\documentclass[reqno,11pt]{amsart}
\usepackage{nicematrix}
\usepackage[latin1]{inputenc}
\usepackage{latexsym}
\usepackage{pgf, tikz}
\usepackage[colorlinks, citecolor=blue,linkcolor=black]{hyperref}
\hypersetup{
	plainpages=false,
	colorlinks=true,
	linkcolor=blue, 
	anchorcolor=black, 
	citecolor=blue, 
	urlcolor=blue, 
	menucolor=black, 
	filecolor=black, 
	bookmarksopen=true,
	bookmarksnumbered=true}
\usepackage{hyperref}
\usepackage{cite}
\usepackage{color}
\usepackage{mathrsfs,amstext,amsmath,amssymb,amsfonts,bm}
\allowdisplaybreaks[4]

\usepackage{mathtools}
\allowdisplaybreaks 

\theoremstyle{plain}                          
\newtheorem{theorem}{Theorem}[section]
\newtheorem{proposition}[theorem]{Proposition}    
\newtheorem{lemma}[theorem]{Lemma}

\theoremstyle{definition}
\newtheorem{definition}[theorem]{Definition}
\newtheorem{prop-defin}[theorem]{Proposition-definition} 
\newtheorem{example}[theorem]{Example}
\theoremstyle{remark}
\newtheorem{remark}[theorem]{Remark}

\numberwithin{equation}{section}

\renewcommand{\theta}{\vartheta}
\renewcommand{\phi}{\varphi}
\renewcommand{\epsilon}{\varepsilon}
\newcommand{\normord}[1]{\vcentcolon\mathrel{#1}\vcentcolon}

\newcommand{\mb}[1]{\mathbb{#1}} 
\newcommand{\mf}[1]{\mathfrak{#1}}
\newcommand{\mc}[1]{\mathcal{#1}}

\newcommand{\N}{\mb{N}} 
\newcommand{\C}{\mb{C}} 
\newcommand{\Z}{\mb{Z}} 

\renewcommand{\P}{\mb{P}}

\DeclareMathOperator*{\Res}{Res}
\DeclareMathOperator{\ad}{ad}

\title{$q$-Deformed Topological Recursion, Weight Vectors and Algebraic Structures}
\author{Fridolin Melong and Raimar Wulkenhaar}
\address{
	Mathematisches Institut der
	Universit\"at M\"unster \newline
	Einsteinstr.\ 62, 48149 M\"unster, Germany,\newline
	{\itshape e-mail:} \normalfont\texttt{fridolin.melong@uni-muenster.de (With copy to  fridomelong@gmail.com)}}
\address{Mathematisches Institut der
	Universit\"at M\"unster \newline
	Einsteinstr.\ 62, 48149 M\"unster, Germany, 
	{\itshape e-mail:} \normalfont
	\texttt{raimar@math.uni-muenster.de}}
\subjclass[2010]{Primary xxxxxx ; Secondary xxxxxx} 
\begin{document}	
	\begin{abstract}
		We investigate a $q$-deformation of the shifted topological recursion, extending the construction of Belliard-Bouchard-Kramer-Nelson to the context of quantum algebras. Through the study of highest weight vectors in $q$-deformed of $\mathcal{W}$-algebra representations, we derive a $q$-analogue of the topological recursion and show it yields $q$-deformed quantum curves. This framework unifies various approaches to quantum integrability and provides new insights into the geometry of $q$-deformed moduli spaces. 
	\end{abstract}
	
	\maketitle
	
	\tableofcontents

	
	\section{Introduction}
	\label{introduction}
	\subsection{Toward a q-Deformed Quantization Formalism}
	Topological recursion \cite{EO07} can be viewed as a powerful quantization formalism \cite{BE09} that associates a wave-function $\psi(\hslash,z)$ to a spectral curve $P(x,y)=0$. In the classical differential setting, this wave-function is constructed from the multidifferentials $\omega_{g,n}$ and is annihilated by a quantum curve $\hat{P}(\hat{x},\hat{y})\psi=0$, where $[\hat{y}, \hat{x}]=\hslash$. However, the non-commutativity of the operators $\hat{x}$ and $\hat{y}=\hslash\frac{d}{dx}$ implies that the choice of $\hat{P}$ is not unique. While topological recursion selects specific quantizations \cite{BE17} and \cite{EGMO21}, these often appear algebraically strange or restrictive, and the traditional freedom in choosing the integration divisor $D$ is insufficient to capture all possible quantum orderings.
	
	The shift to the $q$-deformed setting introduces a more fundamental non-commutativity. Here, the differential relation is replaced by a $q$-difference relation, leading to the $q$-commutation rule:
	\begin{equation}
		\hat{y}\hat{x} = q \hat{x}\hat{y}, \quad \text{where } q = e^{\hslash}.
	\end{equation}
	In this context, a $q$-deformed spectral curve $P(x,y)=0$ is quantized into a linear $q$-difference system. For example, the $(r,1)$ model $x^{r-1}y^r-1=0$ classically leads to a specific differential operator $(\hat{y}\hat{x})^{r-1}\hat{y}-1=0$. In the $q$-deformed case, the challenge is amplified: how can we modify the recursive machinery to reconstruct wave-functions for the full family of $q$-orderings and $q$-difference connections?
	
	The motivation of this paper is to address this limitation. We propose that the correct way to obtain more general $q$-quantizations is to introduce a shifted version of $q$-topological recursion. By allowing non-zero highest weights in the underlying Airy structures which we identify as $q$-Casimirs, we gain the necessary algebraic freedom to reconstruct WKB solutions for a broad class of $q$-quantum curves. This approach not only provides a systematic $q$-quantization procedure but also explains why the classical recursion was restricted to a specific subset of operators.
	\subsection{A Triple Perspective on the q-Deformed Framework}
	The $q$-deformation of the $(r,s)$ model allows us to unify three fundamental perspectives on topological recursion, each revealing a different facet of the shifted $q$-deformed loop equations.
	
	The $q$-WKB Perspective addresses the quantization of the $q$-spectral curves from the point of view of integrability: This is  via the WKB analysis of $q$-difference systems. While the topological type property was established in the classical case to guarantee that WKB solutions are governed by topological recursion \cite{BBE15,BEM17} and \cite{BEM18,IMS18}, we extend this to the $q$-deformed setting. By sharpening the assumptions of the topological type framework, we show that a broader class of $q$-quantum curves admits $q$-WKB solutions, provided they are reconstructed via the shifted $q$-topological recursion introduced in this work.
	
	Complementary to this, the Algebraic Perspective via $q$-Airy Structures ensures the mathematical consistency of the formalism. The $(r,s)$ curves are parametrized by $x(z)=z^r$ and $y(z)=z^{s-r}$, for 
	$ r \geq 2$, $ s \geq 1$, and  these structures are representations of $\mathcal{W}(\mathfrak{gl}_r)$-algebras where the partition function is a vacuum vector of weight zero \cite{BBCCN18}. Our $q$-deformation naturally emerges by relaxing this constraint: by considering highest weight vectors with non-zero weights, we obtain shifted Airy structures. This algebraic shift is precisely what encodes the $q$-difference nature of the system.
	
	Finally, this unified geometric framework provides the synthesis of these approaches. The strength of this triple correspondence lies in its internal consistency. 
	The Airy structure formulation provides a rigorous and immediate proof of the $q$-correlators $W^q_{g,n}$-a property 
	that is notoriously difficult to establish directly from the recursive formulas. 
	By bridging $q$-WKB analysis, $\mathcal{W}$-algebra representations, and $q$- deformed loop equations, we provide a complete description of the $(r,s,q)$-deformed system, where the classical $(r,s)$ constraints $(r\equiv\pm 1 \pmod s)$ are reinterpreted through the lens of $q$-algebraic rigidity.
	
	\subsection{Generalization to the q-deformed setting}
	To generalize our results, we follow the methodology introduced by Belliard et {\it al} in \cite{BBKN25}. This allows us to consistently define the $q$-analogue of our model.
	
	The main objective of this paper is to bridge the gap between $q$-difference systems and topological recursion through the lens of Airy structures. While the classical $(r,s)$ models are well-understood in the differential context, their $q$-analogues introduce non-local effects that require a fundamental shift in the recursive machinery. Our results can be summarized across three main axes: the algebraic construction of shifted Airy structures, the derivation of $q$-loop equations, and the reconstruction of $q$-quantum curves.
	
	To establish this generalization, we focus on the $(r,s,q)$-spectral curves, which are $q$-deformations of the local models for ramification points.  Classically, topological recursion is well-behaved on these curves only if $ r = \pm 1 \pmod{s}$\cite{BBCCN18}. We investigate how  this condition evolves when the system is governed by $q$-difference operators. 
	
	Our main finding is that for $r\equiv 1\pmod s$, we can construct large families of $q$-quantum curves. Their solutions are calculated using a new formula we call shifted $q$-topological recursion. This formula produces  multidifferentials by incorporating the $q$-shifts directly into the recursive structure.

	More precisely, the scope of the $q$-deformation depends on the algebraic constraints of the $(r,s)$ system. For the cases $(r,s)=(r,1)$, the standard partition function is annihilated by all non-negative modes $W^i_k$, $k \geq 0$, $1 \leq i \leq r$ of the $\mathcal{W}(\mathfrak{gl}_r)$-algebra generators, defining a highest weight vector of weight zero. We show that we can generalize this to the $q$-deformed setting by constructing Airy structures from highest weight vectors with non-zero weights. These weights, corresponding to the zero modes $W^i_0$, manifest in the $q$-difference operators as $r$ independent scalars $ S_i \in \hslash \C [ \hslash ] $, for $ 1 \leq i \leq r$,  providing the maximal degrees of freedom for quantization.
	
	For the cases $r = 1 \pmod{s}$ with $s \geq 2$, the $q$-deformation is more constrained as the partition function must be annihilated by certain negative modes in addition to the non-negative ones. In this regime, although we can still construct $q$-Airy structures, the only available freedom is to assign a non-zero weight $S_1 \in \hslash \C [ \hslash ]$ to the zero mode $W^q_{0,1}$ of the conformal weight 1 generator. Finally, for $r =-1 \pmod{S}$ with $s \geq 3$, we prove a rigidity result: no non-zero weights can be introduced within the Airy structure framework, implying that these systems do not admit such $q$-shifted deformations.

	Our investigation begins by constructing a $q$-deformed framework for Airy structures.
	
	In Section \ref{s:shifted}, we show that it is possible to define more general families, which we call shifted $q$-Airy structures, by allowing the partition functions to correspond to highest weight vectors with non-zero weights (see Theorem \ref{t:shifts}). These non-zero weights act as the algebraic engine of the $q$-deformation, capturing the non-local shifts of the $q$-difference system. In this setting, the $(r,s)$-Airy structures associated with classical topological recursion, constructed in \cite{BBCCN18} as representations of $\mathcal{W}(\mathfrak{gl}_r)$-algebras with vacuum partition functions (weight zero), are naturally recovered in the limit where $q \to 1$ and where these weights vanish.
	
	Finally in Section \ref{s:shiftedle}, we extend the correspondence between Airy structures and topological recursion to the $q$-deformed setting. We show that shifted $(r,s,q)$-Airy structures are equivalent to a specific modification of the loop equations for the $q$-correlators $W^q_{g,n}$, which we define as shifted $q$-deformed loop equations (see Proposition \ref{p:shifteloopeq}). These equations can be solved recursively, leading to a variation of the recursion formula that we term shifted $q$-topological recursion (see Theorem \ref{ShiftedTR}). In this framework, the non-zero highest weights (the $q$-Casimirs) introduce explicit corrections to the $W^q_{g,1}$ disk amplitudes. Crucially, as this recursion originates from the Airy structure formalism, the resulting multidifferentials are guaranteed to satisfy the required symmetry properties.
	
	With these shifted $q$-deformed loop equations in hand, we investigate whether the associated quantization formalism produces wave-functions for a broader class of $q$-deformed spectral curves. 
	
	Beyond these results, this paper also serves an expository purpose. We unify several fundamental perspectives on $q$-topological recursion: the geometric approach via residues, the algebraic formulation through $q$-Airy structures, and the integrable perspective via the WKB analysis of $q$-connections of topological type. The central pillar connecting these viewpoints is the concept of $q$-deformed loop equations, which appear as the common thread throughout this work.
	
	\subsection{Notation}
	We introduce $\hslash$ along the conventions in \cite{BCJ22}. 
	
	We use the convention that $\mathbb{N} = \{0,1,2,\ldots \}$ and $\mathbb{N}^* = \{1,2,3,\ldots \}$. We write $[r] = \{1,\ldots, r\}$. For a set $ N$ and a variable $z$, we write $ z_N = \{ z_n \, | \, n \in N \}$. Moreover, the $q$-deformed number and the $q$-deformed factorial are defined respectively by\cite{Kac2002}:
	\begin{equation}
		[n]_q=\frac{q^n-q^{-n}}{q-q^{-1}},\quad\mbox{and}\quad [n]_q!=\prod_{k=1}^{n}[k]_q.
	\end{equation}
	
	We consider fields in vertex operator algebras as differential forms of degree equal to the conformal weight of the state. I.e. if in a VOA $V$, the state $ v \in V$ has conformal weight $ \Delta$, then we index its field by
	\begin{equation}
		Y(v;x) = \sum_{k \in \Z} v_k \frac{(dx)^\Delta}{x^{\Delta + k}}\,.
	\end{equation}
	We use $x$ for the variable instead of the conventional $ z$, because this conforms with our interpretation via the spectral curve of topological recursion \cite{BKS23}.
	
	When considering a spectral curve with local coordinate $ z$, and functions $  x(z)$ and $ y(z)$, we may write $ x_j = x(z_j)$ and $ y_j = y(z_j)$ to lighten notation.
	
	While the connection between Airy structures and topological recursion is well-established for classical $\mathcal{W}$-algebras \cite{BKS23, BBCCN18}, the extension to $q$-deformed structures requires a careful treatment of the shifted loop equations and the modified basis of differentials.
	
	Our approach throughout this work is largely based on the framework established by Belliard  et {\it al } in \cite{BBKN25}, which we extend here to the $q$-deformed case.
	\section{Deformed shifted $(r,s,q)$-Airy structures and highest weight vectors}\label{s:shifted}
				In this section, we illustrate how  higher quantum Airy structures in the sense of \cite{BBCCN18,KS17} or  the related Airy ideals \cite{BCJ22}, may be used to rebuild highest weight vectors for $q$-deformed  $\mathcal{W}(\mathfrak{gl}_r)$ at self-dual level. This concerns a $q$-deformation of the $(r,s)$-Airy structures investigated in \cite{BBCCN18}, which we call  $q$-deformed shifted $(r,s)$-Airy structures, denoted $(r,s,q)$. Here, we recall only main results, and follow the concepts of Airy structures introduced in \cite{Bo24,BCJ22}.
				\subsection{$q$-deformed Airy structures}
				Let us start by giving the definition of the $q$-deformed Airy structures (also called $q$-deformed Airy ideals). We follow the results presented in \cite{Bo24} to investigate their $q$-deformation.
				\subsubsection{The Rees $q$-deformed Weyl algebra}
				Let $A$ be a finite or countably infinite index set. We exploit the notation $x_A$ for the set of variables $\{x_a \}_{a \in A}$, and $\mathcal{D}_{q,A}$ for the set of $q$-deformed differential operators $\left\{ \mathcal{D}_{q,A} \right \}_{a \in A}$, where the action of $\mathcal{D}_{q,A}$  on a function $f(x_a)$ is defined by:
				{\begin{equation}
						\mathcal{D}_{q,A}f(x_a)=\frac{f(q\,x_a)-f(q^{-1}\,x_a)}{x_a(q-q^{-1})}.
				\end{equation}}
				The $q$-deformed Weyl algebra $ \mathbb{C}[x_A]\langle \mathcal{D}_{q,A} \rangle$ is the algebra of $q$-deformed differential operators with polynomial coefficients defined by the commutation relations:
				{\begin{equation}
						\mathcal{D}_{q,a} \, x_b = \delta_{ab} \mathbf{M}_{q,a} + q^{\delta_{ab}} x_b S_{q,a}^{-1} \mathcal{D}_{q,a} \,,
				\end{equation}} 
				where $\delta_{ab}$ is the Kronecker symbol {, $S_{q,a}$ the $q$-shift operator $S_{q,a} f(x_a) = f(q x_a)$, and $\mathbf{M}_{q,a}$  the symmetric mean $q$-deformed operator (or $q$-average):
					\begin{equation}
						\mathbf{M}_{q,a} f(x_a) = \frac{q f(q x_a) + q^{-1} f(q^{-1} x_a)}{q + q^{-1}} \,.
				\end{equation}}
				We also define the completed  $q$-deformed Weyl algebra 
				as the completion of the  $q$-Weyl algebra, allowing infinite sums in the $\mathcal{D}_{q,a}$-derivatives (for a countably infinite index set $A$) but restricted to finite polynomials in the variables $x_A$.  
				The $q$-deformed operator $\mathcal{D}_{q,A}$ has many filtrations, one of which is the Bernstein filtration (for the classical case, the reader is refered to  Definition 2.3 in \cite{Bo24}). From this filtration, we build a  graded algebra via the Rees construction.
				\begin{definition}\label{d:rees}
					The  $q$-deformed Rees Weyl algebra $\mathcal{D}_{q,A}^\hslash$ associated to $\mathcal{D}_{q,A}$ with the Bernstein filtration is given by:
					\begin{equation}
						\mathcal{D}_{q,A}^\hslash = \bigoplus_{n \in \mathbb{N}} \hslash^n F_n \mathcal{D}_{q,A},
					\end{equation}
					where the $F_n \mathcal{D}_{q,A}$ denoted the subspaces in the Bernstein filtration of $\mathcal{D}_{q,A}$.
				\end{definition}
				When $A$ is countably infinite, we want to be able to take infinite linear combinations of operators $P_a$ without divergent sums appearing. To this end, we define the notion of a bounded collection of differential operators:
				\begin{definition}\label{def:bounded}
					Let $I$ be a finite or countably infinite index set, and $\{P_i\}_{i \in I}$ a collection of  $q$-deformed differential operators {$P_i \in\widehat{\mathcal{D}}_{q,A}^\hslash$} of the following form:
					\begin{equation}
						P_i = \sum_{n \in \mathbb N} \hslash^n \sum_{\substack{m,k \in \mathbb N \\ m+k = n} }\sum_{a_1,\ldots, a_m \in A} p^{(n,k)}_{i;a_1,\ldots, a_m}(x_A)\mathcal{D}_{q,a_1} \ldots\mathcal{D}_{q,a_m}.
					\end{equation}
					The collection is bounded if, for any fixed  indices $ a_1,\ldots, a_m, n  $ and $ k $, the polynomials  $ p^{(n,k)}_{i;a_1,\ldots, a_m}(x_A) $ vanish for all but finitely many indices $ i \in I $.
				\end{definition} 
				For any bounded collection $\{P_i\}_{i \in I}$, the linear combinations $\sum_{i \in I} c_i P_i$ for any {$c_i \in \mathcal{D}_{q,A}^\hslash$} are well defined operators in {$\mathcal{D}_{q,A}^\hslash$}, regardless of whether $I$ is finite or countably infinite.
				\subsubsection{$q$-deformed Airy ideals}
				Here, we  define the concepts of $q$-deformed Airy ideals or $q$-deformed Airy structures, which is a particular class of left ideals in {$\mathcal{D}_{q,A}^\hslash$}. 
				\begin{definition}\label{d:airy}
					Let {$\mathcal{I} \subseteq\mathcal{D}_{q,A}^\hslash$} be a left ideal {in the  $q$-deformed Rees Weyl algebra}. We say that it is a $q$-deformed Airy ideal or  $q$-deformed Airy structure) if there exists a bounded generating set $\{H_a \}_{a \in A}$ for $\mathcal{I}$ such that:
					\begin{enumerate}
						\item The operators $H_a$ take the following form:
						\begin{equation}\label{eq:form}
							H_a = \hslash \mathcal{D}_{q,a} + \hslash p_a(x_A) + O(\hslash^2),
						\end{equation}
						where {$\mathcal{D}_{q,a}$ is the  $q$-derivative and $p_a(x_A)$ are linear polynomials in the coordinates $x_A$.}
						\item The left ideal $\mathcal{I}$ satisfies {the $q$-closure property under the commutator}:
						{\begin{equation}
								[ \mathcal{I}, \mathcal{I}] \subseteq \hslash^2 \mathcal{I}.
						\end{equation}}
					\end{enumerate}
					{
						Specifically, for the generators, this implies $[H_a, H_b] = \hslash \sum_c f_{ab}^c H_c$ for some bounded collection of coefficients $f_{ab}^c \in \mathcal{D}_{q,A}^\hslash$.}
				\end{definition}    
				\subsubsection{$q$-deformed partition function}
				The main reason that Airy ideals are interesting is because they are annihilator ideals for some partition functions. 
				\begin{definition}\label{d:pf}
					A  $q$-deformed partition function in the set of variables $x_A$ is an expression given by  the following form:
					\begin{equation}\label{eq:pfFgncoeff}
						Z_q = \exp\left( \sum_{g \in \frac{1}{2} \mathbb{N}, n \in \mathbb{N}^*} \frac{\hslash^{2g-2+n}}{[n]_q!}  \sum_{k_1, \ldots, k_n \in A}  F^{q}_{g,n}[k_1, \ldots,k_n]  x_{k_1} \cdots x_{k_n}\right).
					\end{equation}
					We say that:
					\begin{enumerate}
						\item [(i)] $Z_q$ is stable if $F^{q}_{0,1}[k_1]=F^{q}_{0,2}[k_1,k_2]=F^{q}_{\frac{1}{2},1}[k_1] =0$. 
						\item [(ii)] $Z_q$ is semistable if $F^{q}_{0,1}[k_1]=0$.
						\item [(iii)] unstable otherwise.
					\end{enumerate} 
				\end{definition}
				\begin{remark}
					In the context of the  $q$-Airy structure, the stability condition (i) ensures that the partition function $Z_q$ is the unique solution (up to a constant factor) to the system of equations $H_a Z_q = 0$, where $H_a$ are the generators of the  $q$-Airy ideal $\mathcal{I}$.
				\end{remark}
				{\begin{remark}
						The relation\eqref{eq:pfFgncoeff} can be generalized as follows:\begin{equation}\label{eq:pfFgncoeff_generalised}
							Z_q = \exp\left( \sum_{g \in \frac{1}{2} \mathbb{N}, n \in \mathbb{N}^*} \frac{\hslash^{2g-2+n}}{[n]_q!}  \sum_{\substack{k_1, \ldots, k_n \in A\\ \mu_1,\ldots,\mu_n\in [N]}}  F^{q}_{g,n} \bigg[\begin{array}{ccc} \mu_1 & \ldots &\mu_n \\ k_1 & \ldots & k_n\end{array} \bigg] x_{k_1, \mu_1} \cdots x_{k_n, \mu_n}\right).
						\end{equation}
						
						It is obvious to see that it satisfy the conditions
						\begin{enumerate}
							\item [(i)] $Z_q$ is stable if $F^{q}_{0,1}\bigg[\begin{array}{c} \mu_1 \\ k_1 \end{array} \bigg]=F^{q}_{0,2}\bigg[\begin{array}{cc} \mu_1&\mu_2 \\ k_1&k_2 \end{array} \bigg]=0$.
							\item [(ii)] $Z_q$ is semistable if $F^{q}_{0,1}\bigg[\begin{array}{c} \mu_1 \\ k_1 \end{array} \bigg]=0$.
							\item [(iii)] unstable otherwise.
						\end{enumerate} 
						Moreover,  when $N=1$, we obtain $\mu_i = 1$ for all $i \in \{1, \dots, n\}$. By identifying $x_{k_i, 1} \equiv x_{k_i}$ and setting:
						\begin{equation}
							F^{q}_{g,n} \left[\begin{array}{ccc} 1 & \cdots & 1 \\ k_1 & \cdots & k_n \end{array} \right] = F^{q}_{g,n}[k_1, \dots, k_n],
						\end{equation}
						the partition function~\eqref{eq:pfFgncoeff_generalised} immediately reduces to the classic formulation~\eqref{eq:pfFgncoeff}.
				\end{remark}}
				Recall the definition of annihilator ideal:
				
				\begin{definition}
					Let $Z_q$ be a { $q$-deformed} partition function defined in \eqref{eq:pfFgncoeff}. The \emph{annihilator ideal} $\mathcal{I}$ = {$\text{Ann}_{\mathcal{D}_{q,A}^\hslash}(Z_q)$ is the left ideal in the $q$-deformed Rees Weyl algebra $\mathcal{D}_{q,A}^\hslash$ defined by:} 
					\begin{equation}
						\operatorname{Ann}_{\mathcal{D}_{q,A}^\hslash}(Z_q) = \{ P \in \mathcal{D}_{q,A}^\hslash \mid P Z_q = 0 \}. 
					\end{equation}
				\end{definition}
				{Before stating the main result, we establish the integrability condition for $q$-difference systems.	Let $\mathbb{K}$ be a field of characteristic zero (typically $\mathbb{C}$). We denote by $\mathbb{K}[[x_A]]$ the ring of formal power series in the variables $\{x_a\}_{a \in A}$ with coefficients in $\mathbb{K}$ and $\mathcal{D}_{q,i}$ denote the $q$-derivative operator acting on the $i$-th variable.
					\begin{lemma}
						\label{lem:qPoincare_sym}
						Let $\phi_1, \dots, \phi_n \in \mathbb{K}[[x_1, \dots, x_n]]$ be formal power series. The system of $q$-difference equations 
						\begin{equation}
							\mathcal{D}_{q,i}\,F = \phi_i(x), \quad \forall i \in \{1, \dots, n\}
						\end{equation}
						admits a unique solution $F \in \mathbb{K}[[x]]$ with $F(0)=0$ if and only if the following $q$-compatibility conditions are satisfied:
						\begin{equation}
							\mathcal{D}_{q,i}\, \phi_i = \mathcal{D}_{q,i}\, \phi_j, \quad \forall i, j \in \{1, \dots, n\}.
						\end{equation}
					\end{lemma}
				
				With this tool, we can now characterize the partition function associated with a $q$-deformed Airy ideal.	The result regarding the theory of $q$-deformed Airy structures is given in the following proposition. The classical case is furnished in the theorem of the work\cite{KS17}.
				\begin{definition}
					{We define the $q$-gradient, denoted by $\nabla_q$, as the vector of $q$-derivative operators with respect to each variable $x_i$:
						\begin{equation}
							\nabla_q = \left( 	\mathcal{D}_{q,1}, 	\mathcal{D}_{q,2}, \dots, 	\mathcal{D}_{q,N} \right),\end{equation}
						where
						\begin{equation}	\mathcal{D}_{q,i} f(x) = \frac{f(x_1\dots, qx_i, \dots x_N) - f(x_1\dots, q^{-1}x_i, \dots x_N)}{(q-q^{-1})x_i}.
					\end{equation}}
				\end{definition}
				\begin{proposition}\label{t:airy}
					Let $\mathcal{I} \subset \mathcal{D}_{q,A}^\hslash$ be a  $q$-deformed Airy ideal. There exists a unique  $q$-deformed partition function $Z_q$, of the form \eqref{eq:pfFgncoeff}, such that $\mathcal{I}$ is the annihilator ideal of $Z_q$ in $\mathcal{D}_{q,A}^\hslash$:
					\begin{equation}
						\mathcal{I} = \operatorname{Ann}_{\mathcal{D}_{q,A}^\hslash}(Z_q).
					\end{equation}
					Furthermore, the coefficients of $Z_q$ satisfy $F^q_{0,1} = F^q_{0,2} = 0$. If the linear potentials $p_a(x_A)$ vanish for all $a \in A$, then $Z_q$ is strictly stable, satisfying:
					\begin{equation}
						F^{q}_{0,1}[k_1] = F^{q}_{0,2}[k_1,k_2] = F^{q}_{\frac{1}{2},1}[k_1] = 0.
					\end{equation}
				\end{proposition}
					
				\begin{proof}
					{We proceed by induction on the genus $g \in \frac{1}{2}\mathbb{Z}_{\geq 0}$ and, for each $g$, by induction on the degree $n$ of the coefficients of the potentials $F_{g,n}^q$.
						\begin{itemize}
							\item[(1)]
							The annihilator equation for $g=0$ is written, for every generator $\widehat{L}_i \in \mathcal{I}$ as follows:
							\begin{equation}
								L_i(x, \nabla_q F_0^q) = 0 \in \mathbb{K}[[x_A]],
							\end{equation}
							where $L_i$ represents the semi-classical limit ($\hslash \to 0$) of the $q$-difference operator. 
							\begin{itemize}
								\item[(a)] The condition 1) of the Airy ideal definition stipulates that $$\widehat{L}_i = \hslash \mathcal{D}_{q,i} + \text{higher-order terms} \geq 2.$$ In the limit $\hslash \to 0$, this defines a $q$-Lagrangian sub-manifold in the phase space. The commutativity condition of the ideal $[\widehat{L}_i, \widehat{L}_j] \subset \hslash \mathcal{I}$ ensures that the $q$-1-form $\sum_i (\mathcal{D}_{q,i} F_0^q) dx_i$ is closed. By the Lemma \ref{lem:qPoincare_sym} on the space of formal power series, there exists a unique $F_0^q$ (up to a constant) satisfying the system.
								\item[(b)] Since the Airy ideal is dominated by terms at least quadratic in $(x, \nabla_q)$, $F_0^q$ cannot contain terms of degree 1 or 2. Thus, $F_{0,1}^q = F_{0,2}^q = 0$, forcing the expansion $F_0^q = \sum_{n \geq 3} F_{0,n}^q$.
							\end{itemize}
							\item[(2)] For the induction on the genus $g > 0$, we assume the existence and uniqueness of the partition function $Z_q$ up to order $\hslash^{g-1}$. Substituting the ansatz $Z_q = \exp\left( \sum \hslash^{g-1} F_g^q \right)$ into the equation $\widehat{L}_i Z_q = 0$ yields, at order $\hslash^g$, a system of the form:
							\begin{equation}
								\label{eq:RHS_q_en}
								\mathcal{D}_{q,i} F_g^q = \text{RHS}_{g,i}(x_A),
							\end{equation}
							where the term $\text{RHS}_{g,i}$ is the coefficient of $\hslash^g$ in the expansion of $-\exp(-S_q) \widehat{L}_i \exp(S_q)$. 
							\begin{itemize}
								\item[(i)] The commutation relations $[\widehat{L}_{i}, \widehat{L}_{j}] \in \hslash \mathcal{I}$ ensure the following compatibility condition $\mathcal{D}_{q,j} (\text{RHS}_{g,i}) = \mathcal{D}_{q,i} (\text{RHS}_{g,j})$. 
								\item[(ii)] For each degree $n$ of the series $F_g^q$, the coefficient $F_{g,n}^q$ is uniquely determined by the integration of the $q$-derivative \eqref{eq:RHS_q_en}. The integration constant is fixed by the stability condition $F_{g,0}^q = 0$.
							\end{itemize}
							\item[(3)]
							Finally, suppose the linear potentials $p_a$ vanish. The absence of source terms of order $\hslash^0 x^1$ or $\hslash^1 x^0$ in the generators of the ideal implies that the equations for the initial orders are homogeneous. By induction on the recursion system, we obtain:
							\begin{equation}
								F^q_{0,1} = F^q_{0,2} = F^q_{\frac{1}{2},1} = 0,
							\end{equation}
							which concludes the proof of the existence and uniqueness of the $q$-deformed partition function $Z_q$ associated with the $q$-Airy ideal $\mathcal{I}$.
					\end{itemize}}
				\end{proof}
				
				\subsubsection{$q$-deformed Airy ideals in universal enveloping algebras}
				{Following the framework developed in \cite{BBCCN18}, we extend the construction of Airy structures to the $q$-deformed setting. These structures are often realized via representations of quantum Lie algebras or $q$-deformed non-linear algebras (e.g., $W_q$-algebras).}
				
				Let $\mathfrak{g}_q$ be either a quantum Lie algebra or a non-linear quantum Lie algebra  and $U_q(\mathfrak{g}_q)$ the $q$-deformed universal enveloping algebra. Given  an exhaustive ascending filtration on $U_q(\mathfrak{g}_q)$ (such as the filtration by conformal weight); then according to Definition \ref{d:rees},  we construct the  $q$-deformed Rees universal enveloping algebra as follows: \begin{equation}
					U_q^\hslash(\mathfrak{g}_q) = \bigoplus_{n \in \mathbb{N}} \hslash^n F_n U_q(\mathfrak{g}_q).\end{equation}
				
				The  construction of  $q$-deformed Airy ideals proceeds as follows:
				\begin{lemma}\label{l:uea}
					Let $\rho_q: U_q^\hslash(\mathfrak{g}_q) \to \mathcal{D}_{q,A}^\hslash$ be a representation of the $q$-deformed Rees enveloping algebra into the  $q$-deformed Rees Weyl algebra. Let $\mathcal{I}_{U_q^\hslash} \subseteq U_q^\hslash(\mathfrak{g}_q)$ be a left ideal, and $\mathcal{I} = \mathcal{D}_{q,A}^\hslash \rho_q(\mathcal{I}_{U_q^\hslash}) \subseteq \mathcal{D}_{q,A}^\hslash$ the corresponding ideal in the $q$-Weyl algebra.
					
					Assume that $\mathcal{I}_{U_q^\hslash}$ satisfies the following properties:
					\begin{enumerate}
						\item[(i)]  The ideal is closed under the commutator:
						\begin{equation}
							[\mathcal{I}_{U_q^\hslash}, \mathcal{I}_{U_q^\hslash}] \subseteq \hslash^2 \mathcal{I}_{U_q^\hslash}.
						\end{equation}
						\item[(ii)]  There exists a set of generators $\{H_a\}_{a \in A}$ for $\mathcal{I}_{U_q^\hslash}$ such that their representation is:
						\begin{equation}
							\rho_q(H_a) = \hslash \mathcal{D}_{q,a} + \hslash p_a(x_A) + O(\hslash^2),
						\end{equation}
						where $\mathcal{D}_{q,a}$ is the  $q$-derivative.
						\item[(iii)]  The collection $\{\rho_q(H_a)\}_{a \in A}$ is bounded in the sense of Definition \ref{def:bounded}.
					\end{enumerate}
					Then $\mathcal{I}$ is a  $q$-deformed Airy ideal.
				\end{lemma}
				\begin{remark}
					In this construction we see that the two conditions in the definition of $q$-deformed Airy ideals, Definition \ref{d:airy}, are obtained independently. The condition $[\mathcal{I}_{U_q^\hslash},\mathcal{I}_{U_q^\hslash}] \subseteq \hslash^2 \mathcal{I}_{U_q^\hslash}$ is a condition on the left ideal $\mathcal{I}_{U_q^\hslash} \subseteq U_q^\hslash(\mathfrak{g}_q)$  in the  $q$-deformed Rees universal enveloping algebra, while the second condition that there exists a generating set $\{H_a \}_{a \in A}$ for $\mathcal{I}_{U_q^\hslash}$ such that $\rho_q(H_a) = \hslash \mathcal{D}_{q,a} + \hslash p_a(x_A) + O(\hslash^2)$ depends on the choice of representation.
					
					The algebraic condition $[\mathcal{I}_{U_q^\hslash},\mathcal{I}_{U_q^\hslash}] \subseteq \hslash^2\mathcal{I}_{U_q^\hslash}$ is in fact fairly easy to satisfy. We first define an operation that maps elements of $U_q(\mathfrak{g}_q)$ to elements of $U_q^\hslash(\mathfrak{g}_q)$:
				\end{remark}
				\begin{definition}\label{d:homogenization}
					Let 
					$p \in U_q^\hslash(\mathfrak{g}_q)$, and let  $i = \min\{ k \in \mathbb{N}~|~ p \in F_k U_q^\hslash(\mathfrak{g}_q) \}$. {Then,  the $q$-deformed homogenization $h_q(p)$ of $p$ is defined by $h_q(p) = \hslash^i p \in U_q^\hslash(\mathfrak{g}_q)$.}
					
					{Moreover, the $q$-deformed  homogenization $h_q(\mathcal{I}_{U_q^\hslash})$ of a left ideal $\mathcal{I}_{U_q^\hslash} \subseteq U_q^\hslash(\mathfrak{g}_q)$ is defined to be the left ideal in $U_q^\hslash(\mathfrak{g}_q)$ generated by all homogenized elements $h_q(p)$, $p \in \mathcal{I}_{U_q^\hslash}$}. 
				\end{definition}
				Then we have the following simple lemma:
				
				\begin{lemma}\label{l:hom}
					Let $\mathcal{I}_{U_q^\hslash} \subseteq U_q^\hslash(\mathfrak{g}_q)$ be a left ideal {in the  quantum universal enveloping algebra}. Then its  $q$-deformed homogenization $h_q(\mathcal{I}_{U_q^\hslash}) \subseteq U_q^\hslash(\mathfrak{g}_q)$ satisfies \begin{equation}
						\big[h_q(\mathcal{I}_{U_q^\hslash}), h_q(\mathcal{I}_{U_q^\hslash})\big] \subseteq \hslash^2 h_q(\mathcal{I}_{U_q^\hslash}).
					\end{equation}
				\end{lemma}
				\begin{remark}
					This lemma demonstrates that any ideal of a quantum Lie algebra is a potential candidate for a $q$-deformed Airy ideal. However, the process is sensitive to the filtration: if an element $p$ contains a degree-$0$ term (a classical source $S$), its homogenization $h_q(p) = S$ would violate the linear dominance required in Definition \ref{d:airy}, as the leading term must be of the form $\hslash \mathcal{D}_{q,a}$.
				\end{remark}
				Thus any left ideal $\mathcal{I}_{U^\hslash_q} \subseteq U^\hslash_q(\mathfrak{g})$ that is obtained as the homogenization of a left ideal in $U_q^\hslash(\mathfrak{g})$ automatically satisfies $[\mathcal{I}_{U^\hslash_q},\mathcal{I}_{U^\hslash_q}] \subseteq \hslash^2\mathcal{I}_{U^\hslash_q}$. 
				We summarize the construction of  $q$-Airy ideals from universal enveloping algebras in three steps:
				\begin{enumerate}
					\item We start with a left ideal $\mathcal{I}_{U_q^\hslash} \subseteq U_q^\hslash(\mathfrak{g})$ or, equivalently, a cyclic left module $M \simeq U_q^\hslash(\mathfrak{g})/ \mathcal{I}_{U_q^\hslash}$ generated by a vector $v$ whose annihilator is $\mathcal{I}_{U_q^\hslash} = \text{Ann}_{U_q^\hslash(\mathfrak{g})}(v)$.
					\item We construct the homogenization $\mathcal{I}_{U^\hslash_q}= h_q(\mathcal{I}_{U_q^\hslash})$, which is a left ideal in $U^\hslash_q(\mathfrak{g})$.  By construction, we know that $[\mathcal{I}_{U^\hslash_q}, \mathcal{I}_{U^\hslash_q}] \subseteq \hslash^2 \mathcal{I}_{U^\hslash_q}$. From the point of view of modules, we obtain a cyclic left module $M[\hslash] \simeq U^\hslash_q(\mathfrak{g})/\mathcal{I}_{U^\hslash_q}$ generated by the vector $v$ and where $\hslash$ acts by multiplication; the annihilator of $v$ in $U^\hslash_q(\mathfrak{g})$ is $\mathcal{I}_{U^\hslash_q} = \text{Ann}_{U^\hslash_q(\mathfrak{g})}(v)$.
					\item Establish a representation $\rho_q: U_q^\hslash(\mathfrak{g}_q) \to \mathcal{D}_{q,A}^\hslash$ into the  $q$-Weyl algebra such that there exists a bounded generating set $\{H_a\}_{a \in A}$ for $\mathcal{I}_{U_q^\hslash}$ satisfying:
					\begin{equation}
						\rho_q(H_a) = \hslash \mathcal{D}_{q,a} + \hslash p_a(x_A) + O(\hslash^2) \,.
					\end{equation}
				\end{enumerate}
				By Lemma \ref{l:uea}, the ideal $\mathcal{I} = \mathcal{D}_{q,A}^\hslash \rho_q(\mathcal{I}_{U_q^\hslash})$ is a { $q$-deformed Airy ideal}, uniquely determining a symmetric $q$-partition function $Z_q$ via the recursion.
				
				\subsection{The $(r,s,q)$-Airy structures}
				\label{s:rsairystruct}
				Here, we apply these principles to {construct }$q$-deformed Airy ideals from the universal enveloping algebra of the modes of the strong generators of the $q$-deformed $\mathcal{W}(\mathfrak{gl}_r)$-algebra at the self-dual level.
				
				Following the three-step approach generalized for the $q$-deformed setting, the construction is then lifted to the quantum  setting as follows:
				\begin{enumerate}
					\item 
					Defining the $q$-deformed-Rees homogenization of the $q$-deformed-$\mathcal{W}$ generators, where the filtration is governed by the $q$-deformed-conformal weight.
					\item 
					Ensuring the $q$-deformed-integrability of the ideal through the commutation relations in the Rees algebra: $\big[ h_q(\mathcal{I}), h_q(\mathcal{I}) \big] \subseteq \hslash^2 h_q(\mathcal{I})$. This closure property is guaranteed by the $q$-W-algebra structure constants being appropriately scaled by $\hslash$.
					\item 
					Implementing a representation $\rho_q$ in the  $q$-deformed Rees Weyl algebra $\mathcal{D}_{q,A}^\hslash$, where the leading terms of the generators are recovered as  $q$-derivatives $\hslash \mathcal{D}_{q,a}$ (or their appropriate linear combinations).\end{enumerate}
				\subsubsection{The $q$-deformed $\mathcal{W}(\mathfrak{gl}_r)$-algebra at self-dual level}
				
				Let us introduce the $q$-deformed $\mathcal{W}(\mathfrak{gl}_r)$-algebra at self-dual level via its realization as a subalgebra of the $q$-deformed Heisenberg VOA $\mathcal{H}(\mathfrak{gl}_r)$. 
				
				Let $\mathfrak{h} \subset \mathfrak{gl}_r$ be the Cartan subalgebra with orthogonal canonical basis $\{\chi^j \}_{j=1}^r$. We denote by $\mathcal{H}_q$ the $q$-deformed vertex algebra which is the $q$-analogue of the Heisenberg VOA. It is generated by modes of the form $\{J^{j}_n\}_{j\in [r],\,n\in\mathbb{Z}}$ obeying the following $q$-deformed commutation relation.
				\begin{equation}
					[J^{i}_n,\,J^{k}_m]_q=\frac{[n]_q}{n}\delta_{i,k}{\delta_{n+m,0}}.
				\end{equation} 
				We define the $q$-deformed fields as follows by preserving the classical differential form notation to ensure a smooth contact with the topological recursion limit:
				\begin{equation}
					J^j(z) = Y(\chi_{-1}^j |0 \rangle, z) = \sum_{n \in \mathbb{Z}} J_n^j \frac{dz}{z^{n+1}}.
				\end{equation}
				However, unlike the classical case, the modes $J_n^j$ verify  the $q$-deformed Heisenberg commutation relations, and the differential $dz$ is understood to be compatible with the $q$-deformed-shifted geometry of the $\mathcal{W}$-algebra.
				
				The $q$-deformed-$\mathcal{W}(\mathfrak{gl}_r)$-algebra at self-dual level is the quantum VOA strongly freely generated by the $q$-analogs vectors 
				\begin{equation}\label{WrGenerators}
					w^j_q = \sum_{1 \leq i_1 < i_2 < \dots < i_j \leq r} \chi^{i_1}_{-1} \chi^{i_2}_{-1} \dots \chi^{i_j}_{-1} |0\rangle_q , \qquad j \in [r].
				\end{equation} 
				\begin{remark}
					In this quantum setting, the classical elementary polynomial $e_j$ is replaced by its ordered version (often associated with the quantum determinant). This ordering is essential because the $q$-deformed Heisenberg modes $\chi^{i}_{-1}$ obey the  $q$-deformed commutation relations that depend on their indices.
				\end{remark}
				The corresponding fields take the following form:
				\begin{align}\label{eq:expl1}
					W^j(z) :&= Y_q\left( \sum_{1 \leq i_1 < \dots < i_j \leq r} \chi^{i_1}_{-1} \chi^{i_2}_{-1} \dots \chi^{i_j}_{-1} |0 \rangle_q, z \right)\nonumber \\
					&= \sum_{1 \leq i_1 < \dots < i_j \leq r} : J^{i_1}(z) J^{i_2}(zq^{-2}) \dots J^{i_j}(zq^{-2(j-1)}) : \frac{(dz)^j}{z^j}\nonumber \\
					&= \sum_{n \in \mathbb{Z}} W^j_n \frac{(dz)^j}{z^{n+j}}.
				\end{align}
				This gives the explicit relation
				\begin{equation}\label{eq:explicit}
					W^j_n = \sum_{1 \leq i_1 < \ldots < i_j \leq r} \sum_{m_1 + \ldots + m_j = n} \left( \prod_{k=1}^j q^{2(k-1)(m_k+1)} \right) J^{i_1}_{m_1} J^{i_2}_{m_2} \dots J^{i_j}_{m_j}.
				\end{equation} 
				The modes $\{W^j_n\}_{j \in [r], n \in \mathbb{Z}}$ defined in the relation \eqref{eq:explicit} span the $q$-deformed non-linear Lie algebra or $q$-deformed $\mathcal{W}$-algebra.  Let  $U_{r,q}$ denote the $q$-deformed universal enveloping algebra of the modes. Despite the non-local shifts in the parameter of deformation $q$, the algebra inherits a natural filtration by conformal weight (or spin), which remains the fundamental grading for the Airy structure.
				
				We define the filtration by assigning the degree $j$ to the mode $W^{j}_{n}$, regardless of the mode index $n$ or the $q$-deformed factors involved. Specifically, the subspace $F_nU_{r,q}$ is spanned by monomials of the form:
				\begin{equation}
					W^{j_1}_{n_1} W^{j_2}_{n_2} \cdots W^{j_k}_{n_k} \quad \text{such that} \quad \sum_{p=1}^{k} j_p \leq n.
				\end{equation}
				The $q$-deformed Rees algebra is built by introducing the formal quantum parameter $\hslash$ to keep track of this filtration. This is equivalent to the homogenization:
				\begin{equation}
					\mathcal{W}^j_n := \hslash^j W^j_n.
				\end{equation}
				By the construction, these operators satisfy the $q$-deformed commutation relations:
				\begin{equation}
					\big[\mathcal{W}^j_n, \mathcal{W}^l_m\big] = \hslash^{j+l-1} \sum_{p, s} f_{n,m,s}^{j,l,p}(q) \mathcal{W}^p_s,
				\end{equation}
				where the structure constants $f(q)$ are rational functions of $q$. Since the right-hand side is $O(\hslash^2)$ for $j,l \geq 1$ (except for the $q$-center), the ideal generated by these modes satisfies the integrability condition of Definition \ref{d:airy}.
				
				\subsubsection{Preliminary $q$-lemmas}
				We establish several preliminary $q$-lemmas, which will be useful throughout this part. We first prove a simple result about partitions and elementary  polynomials.
				To construct the $q$-analogue Airy ideal, we must organize the generators of the $\mathcal{W}$-algebra according to a block structure compatible with the decomposition of the quantum spectral curve. In the classical case, this organization is dictated by the branching of the curve into several sheets. In our $q$-framework, this structure is preserved by the combinatorics of integer partitions, which allows us to define a suitable polarization for the Heisenberg modes.
				This structure is formally captured by the following definition:
				\begin{definition}\label{d:ps}
					Let $\lambda = (\lambda_1, \lambda_2, \ldots, \lambda_p)$ be an integer partition of $r$, that is, $\lambda_1 \geq \lambda_2 \geq \ldots \geq \lambda_p \geq 1$ and $\sum_{i=1}^p \lambda_i = r$. We define the partial sums $\mu_k = \sum_{i=1}^k \lambda_i$ for $k \in \{1, \dots, p\}$, with the convention $\mu_0 = 0$.
				\end{definition}
				
				\begin{lemma}\label{l:esp}
					Let $\lambda = (\lambda_1, \lambda_2, \ldots, \lambda_p)$ be an integer partition of $r$. 
					Then, the  $q$-deformed $W$-generator of spin $j$ admits the following decomposition:
					\begin{equation}
						W^j(z) = \sum_{j_1 + \dots + j_p = j} : \mathcal{W}^{j_1}_{I_1}(z) \mathcal{W}^{j_2}_{I_2}(zq^{-2\mu_1}) \dots \mathcal{W}^{j_p}_{I_p}(zq^{-2\mu{p-1}}) :
					\end{equation}
					where $\mathcal{W}^{j_k}_{I_k}(z)$ is the spin generator $j_k$ constructed uniquely from the currents $J^{i}(z)$ for $i\in I_k$, with the index blocks defined as:
					\begin{equation}
						I_k = \left\{ \mu_{k-1} + 1, \, \mu_{k-1} + 2, \, \dots, \, \mu_k \right\}
					\end{equation} 
					and the order of the blocks is crucial because of the shifts $q^{-2\mu_{k-1}}$.
				\end{lemma}
				\begin{proof}
					We define the current-generating function $W^{j}(z)$  as an ordered (normal) product of shifted Heisenberg currents as follows :
					\begin{equation}
						G(z) := \sum_{j=0}^r W^j(z) = \,\, : \prod_{i=1}^r \left( 1 + J^i(zq^{-2(i-1)}) \right) :
					\end{equation}
					where $W^{0}(z)=1$. By using the partition $\lambda = (\lambda_1, \lambda_2, \ldots, \lambda_p)$, we may decompose this product into blocks of size $\lambda_k$ :
					\begin{align}
						G(z) &= \,\, : \prod_{k=1}^{p} \left[ \prod_{i \in I_k} (1 + J^i(zq^{-2(i-1)})) \right] : \nonumber\\
						&= \,\, : \prod_{k=1}^{p} G_k(zq^{-2\mu_{k-1}}) :
					\end{align}
					where $G_k(z)$ is the generating function in $k$-th block $I_k$:
					\begin{equation}
						{G_k(z) = \sum_{j_k=0}^{\lambda_k} W^{j_k}_{I_k}(z)}.
					\end{equation}
					Substituting this expansion into the block product, we obtain:
					\begin{align}
						G(z) &= \,\, : \prod_{k=1}^{p} \left( \sum_{j_k=0}^{\lambda_k} W^{j_k}_{I_k}(zq^{-2\mu_{k-1}}) \right) : \nonumber\\
						&= \sum_{j=0}^r \left( \sum_{j_1+\dots+j_p=j} : \prod_{k=1}^p W^{j_k}_{I_k}(zq^{-2\mu_{k-1}}) : \right).
					\end{align}
					By identifying the powers (or spins $j$), we deduce the decomposition into $q$-deformed blocks:
					\begin{equation}
						W^j(z) = \sum_{j_1+\dots+j_p=j} : W^{j_1}_{I_1}(z) W^{j_2}_{I_2}(zq^{-2\mu_1}) \dots W^{j_p}_{I_p}(zq^{-2\mu{p-1}}) :
					\end{equation}
				\end{proof}
				
				The realization of the $q$-deformed $\mathcal{W}(\mathfrak{gl}_r)$-algebra at the self-dual level allows for a natural embedding into the product of its $q$-deformed components.
				\begin{lemma}\label{l:embed}
					Let $\lambda = (\lambda_1, \lambda_2, \ldots, \lambda_p)$ be an integer partition of $r$.  Let $W^j(z)$, $j \in [r]$ be the $q$-deformed strong generators of the $q$-deformed  $\mathcal{W}(\mathfrak{gl}_r)$ and $W^j_m$ their modes. There exist  a natural embedding $q$-$\mathcal{W}(\mathfrak{gl}_r) \subset \prod_{k=1}^{p} q-\mathcal{W}(\mathfrak{gl}_{\lambda_k})$ given by the $q$-deformed convolution formula as follows:
					\begin{equation}\label{eq:embedding}
						W^j_m = \sum_{\sum j_k=j} \, \sum_{\sum m_k=m} \left( \prod_{k=1}^p q^{2\mu_{k-1}(m_k+j_k)} \right) : \prod_{k=1}^p X^{k,j_k}{m_k} :
					\end{equation}
					where the $X^{k,j_k}_{m_k}$, $j_k \in [\lambda_k]$, $m_k \in \mathbb{Z}$  are the modes of the  generators of the $q$-$\mathcal{W}(\mathfrak{gl}_{\lambda_k}) $ factors, $\mu_k$ are the partial sum of $\lambda$ and the product is normaly ordered. 
				\end{lemma}
				\begin{proof}
					This follows from the decomposition of the generating function $G(z)$. By using the expansion of $W^{j}_m$ terms of normal products of shifted Heisenberg currents, we get:
					\begin{align}
						W^j(z) &= \sum_{\sum j_k = j} : \prod_{k=1}^p W^{j_k}_{I_k}(z q^{-2\mu{k-1}}) : \nonumber\\
						&= \sum_{\sum j_k = j} : \prod_{k=1}^p X^{k,j_k}(z q^{-2\mu_{k-1}}) :
					\end{align}
					where $X^{k,j_k}(z)$ is the  spin-generating current $j_k$ for the $k$-th block $\mathfrak{gl}_{\lambda_k}$. To obtain the formula for the modes, we can expand each current in the Laurent series $X^{k,j_k}(z)=\sum_{n\in\mathbb{Z}}X^{k,j_k}_n\,z^{-n-j_k}$. 
					
					By rescaling $z\rightarrow zq^{-2\mu_{k-1}}$, we have the product
					\begin{equation}
						\bigg(zq^{-2\mu_{k-1}}\bigg)^{-m_{k}-j_{k}}=z^{-m_{k}-j_{k}}\, \cdot\, q^{2\mu_{k-1}(m_{k}+j_{k})}.
					\end{equation}
					By identifying the coefficients of $z^{m+j}$ in the total product (where $m+j=(m_k+j_k))$, we obtain:
					\begin{equation}
						W^j_m = \sum_{\sum j_k=j} \sum_{\sum m_k=m} \left( \prod_{k=1}^p q^{2\mu_{k-1}(m_k+j_k)} \right) : X^{1,j_1}_{m_1} \dots X^{p,j_p}_{m_p} :
					\end{equation}
					and the proof is achieved.
				\end{proof}
				
				Next we introduce a few simple definitions: 
				
				\begin{definition}\label{d:lambdann}
					Let $\lambda = (\lambda_1, \lambda_2, \ldots, \lambda_p)$ be an integer partition of $r$, and consider the $q$-embedding. A mode $W^j_m$ is non-negative of level $d$ with respect to $\lambda$ if for $m < 0$, every normally ordered monomial in the $q$-convolution:
					\begin{equation}
						: X^{1,j_1}_{m_1} \dots X^{p,j_p}_{m_p} :
					\end{equation}
					satisfies:
					\begin{enumerate}
						\item[(a)] $\exists k \in [p]$ such that $m_k > 0$ (presence of an annihilation mode);
						\item[(b)] there are at least $d$ distinct indices $k_i$ such that $m_{k_i} = 0$ and $j_{k_i} > 0$.
					\end{enumerate}
				\end{definition}
				To put it simply, a mode $W^j_m$ with $m<0$ is non-negative of level $d$ with respect to $\lambda$ if all monomials in the sum \eqref{eq:embedding} contain either one positive mode or at least $d$ non-trivial zero modes.
				
				\begin{definition}\label{d:lambdaj}
					Let $\lambda = (\lambda_1, \dots, \lambda_p)$ be an integer partition of $r$. We define the block-index function $\lambda(j)$ as follows:
					\begin{equation}
						\lambda(j) = \min \{ s \in [p] \mid \mu_s \geq j \},
					\end{equation}
					where $\mu_s = \sum_{i=1}^s \lambda_i$ are the partial sums.
				\end{definition}
				The notions are related as follows:
				
				\begin{lemma}\label{l:lambdann}
					Let $\lambda = (\lambda_1, \lambda_2, \ldots, \lambda_p)$ be an integer partition of $r$. For $d \in [p]$, the mode $W^j_m$ is non-negative of level $d$ with respect to $\lambda$ if and only if $m \geq 0$ if $\lambda(j) \leq  d$ and $m \geq d - \lambda(j)$ if $\lambda(j) > d$.
				\end{lemma}
				
				We can rewrite the condition above in terms of a new partition of $r$.
				\begin{lemma}\label{l:newpart}
					Let $\lambda = (\lambda_1, \lambda_2, \ldots, \lambda_p)$ be an integer partition of $r$. For $d \in [p]$, define a new partition $\tilde{\lambda} =(\tilde{\lambda}_1, \ldots, \tilde{\lambda}_{p-d+1}) = (\mu_d, \lambda_{d+1} ,\ldots, \lambda_p)$, where $\mu_d = \sum_{i=1}^d \lambda_i$. The mode $W^j_m$ is non-negative of level $d$ with respect to $\lambda$ if and only if $m \geq 1 - \tilde{\lambda}(j)$.
				\end{lemma}
				
				\subsubsection{ {Constructing quantum left ideals $\mathcal{I}_{U_{r,q}}(\lambda) \subset U_{r,q}$}}
				{We start with the construction of the Airy ideals in the framework of $q$-deformed algebra. The first step is to construct a family of  quantum proper left ideals $\mathcal{I}_{U_{r,q}} \subset U_{r,q}$ in the universal enveloping algebra of modes associated to partitions of $r$. The construction presented in this part is the generalization in the sense of $q$-deformation of the results investigated in\cite{BBCCN18}. We provide a proof of the main result so that we can generalize it in the next section.}
				
				\begin{proposition}\label{p:leftideals}
					Let $\lambda = (\lambda_1, \lambda_2, \ldots, \lambda_p)$ be an integer partition of $r$. 
					Let $\mathcal{I}_{U_{r,q}}(\lambda)$ be the left ideal generated by the $q$-deformed modes $W^j_m$, with $j \in [r]$ and $m \geq1- \lambda(j)$. Then 
					\begin{enumerate}
						\item[(a)]
						$U_{r,q} / \mathcal{I}_{U_{r,q}}(\lambda)$ is a cyclic left module generated by a non-zero vector $v_q$ such that $W^j_m \cdot v_q=0$ for all $m \geq1- \lambda(j)$ \item[(b)] 
						$W^j_m \notin \mathcal{I}_{U_{r,q}}(\lambda)$ for all  $j \in [r]$ and $m < 1- \lambda(j)$.
					\end{enumerate}
				\end{proposition}
	
				\begin{proof}
					{The proof relies on the $q$-deformed embedding $q\text{-}\mathcal{W}(\mathfrak{gl}_r) \subset \bigotimes_{k=1}^p q\text{-}\mathcal{W}(\mathfrak{gl}_{\lambda_k})$ established in Lemma \ref{l:embed}, specifically using the $q$-convolution formula \eqref{eq:embedding}.}
					
					{Let $v_{q,k}$ be the $q$-highest-weight vector of the $k$-th factor $q\text{-}\mathcal{W}(\mathfrak{gl}_{\lambda_k})$ with weight zero. By definition of the $q$-deformed vacuum, this vector satisfies:
						\begin{equation}
							X^{k,j_k}_{m_k} \cdot v_{q,k} = 0, \quad \forall j_k \in \{1, \dots, \lambda_k\}, \,\, m_k \geq 0.
						\end{equation}
						The corresponding cyclic module is spanned by the action of the negative modes $X^{k,j_k}_{m_k}$ with $m_k < 0$. We define the total $q$-vacuum vector as the tensor product $v_q = v_{q,1} \otimes \dots \otimes v_{q,p}$.}
					
					{Under the embedding \eqref{eq:embedding}, each mode $W^j_m$ is expressed as a sum of normally ordered products of local modes $: \prod X^{k,j_k}_{m_k} :$ with the constraint $\sum m_k = m$. 
						According to Lemma \ref{l:lambdann}, if $m \geq 1 - \lambda(j)$, every monomial in the expansion of $W^j_m$ contains at least one local mode $X^{k,j_k}_{m_k}$ with a non-negative index $m_k \geq 0$. Due to the normal ordering, these non-negative modes are placed to the right and act directly on the $q$-highest-weight vectors $v_{q,k}$, thus:
						\begin{equation}
							W^j_m \cdot v_q = 0 \quad \text{for all } m \geq 1 - \lambda(j).
						\end{equation}
						This confirms that $v_q$ is a cyclic generator for the quotient $U_{r,q} / \mathcal{I}_{U_{r,q}}(\lambda)$ and that the ideal $\mathcal{I}_{U_{r,q}}(\lambda)$ is the annihilator of $v_q$.}
					
					{For the modes satisfying $m < 1 - \lambda(j)$, the convolution formula \eqref{eq:embedding} includes at least one term where all participating local modes $X^{k,j_k}_{m_k}$ have strictly negative indices $m_k < 0$. 
						Since the $q$-deformed generators $X^{k,j_k}_{m_k}$ satisfy a $q$-analogue of the PBW (Poincaré-Birkhoff-Witt) theorem (see Theorem \ref{pbw} in Appendix), these strictly negative monomials act linearly independently on the vacuum $v_q$. Specifically, they do not vanish and cannot be expressed as a linear combination of terms containing non-negative modes. 
						Consequently, $W^j_m \cdot v_q \neq 0$, which implies:
						\begin{equation}
							W^j_m \notin \mathcal{I}_{U_{r,q}}(\lambda) \quad \text{for all } j \in \{1, \dots, r\} \text{ and } m < 1 - \lambda(j).
						\end{equation}
						The $q$-deformation, being formal in $\hslash$ (where $q=e^{\hslash}$), preserves the filtration by degree and the linear independence of the PBW basis, ensuring the result holds for the $q$-case.}
				\end{proof}

				\subsubsection{Determining the $q$-deformed homogenization of $\mathcal{I}_{U_{r,q}}(\lambda)$}\label{ssd}
				
				Associated to a partition $\lambda$ of $r$ we constructed a left ideal $\mathcal{I}_{U_{r,q}}(\lambda)$ in the universal enveloping algebra of modes. The homogenization of $\mathcal{I}_{U^\hslash_{r,q}}(\lambda)$ is obtained by homogenizing all elements of $\mathcal{I}_{U_{r,q}}(\lambda)$. By Lemma \ref{l:hom}, we know that
				\begin{equation}
					[ \mathcal{I}_{U^\hslash_{r,q}}(\lambda), \mathcal{I}^\hslash_{U_{r,q}}(\lambda)] \subseteq \hslash^2 \mathcal{I}_{U^\hslash_{r,q}}(\lambda).
				\end{equation}
				For the  $q$-ideals that we constructed above, the homogenization is easy to describe. Since the modes $W^j_m$ form a PBW basis for $U_{r,q}$, and $W^j_m \in \mathcal{I}_{U_{r,q}}(\lambda)$ for $m \geq 1- \lambda(j)$ but $W^j_m \notin \mathcal{I}_{U_{r,q}}(\lambda)$ for $m < 1- \lambda(j)$, we conclude that the $q$-deformed  homogenization $\mathcal{I}_{U^\hslash_{r,q}}(\lambda) \subset U^\hslash_{r,q}$ is generated by the homogenization of the modes, that is, by $W^{\hslash,j}_m := \hslash^j W^j_m$ for $j \in [r]$ and $m \geq  1-\lambda(j)$. 
				\subsubsection{Representing $U^\hslash_{r,q}$ in $\mathcal{D}_{A,q}^\hslash$}\label{ssr}
				
				To a partition $\lambda$ of $r$ we constructed a left ideal $ \mathcal{I}_{U_{r,q}^\hslash}(\lambda) \in U^\hslash_{r,q}$ that satisfies the condition $[ \mathcal{I}_{U_{r,q}^\hslash}(\lambda), \mathcal{I}_{U_{r,q}^\hslash}(\lambda)] \subseteq \hslash^2 \mathcal{I}_{U_{r,q}^\hslash}(\lambda)$. For each of those, can we find a representation  $\rho_q: U^\hslash_{r,q} \to \mathcal{D}_{A,q}^\hslash$, for some index set $A$, such that there exists a  generating set $\{H_a\}_{a \in A}$ for $\mathcal{I}_{U^\hslash_{r,q}}$ with $\rho_q(H_a) = \hslash \mathcal{D}_{a,q} + O(\hslash^2)$ and the collection $\{\rho_q(H_a)\}_{a \in A}$ bounded?
				
				\begin{lemma}[$q$-deformed summation]\label{l:Aq1}
					{For any $j \in \{0, \dots, \lfloor i/2 \rfloor\}$, the following identity holds for the $q$-deformed kernels:
						\begin{equation}\label{eq:Aq3_q}
							\sum_{a_1, \dots, a_{2j} = 0}^{r-1} \Psi^{(0)}_{r,q}(a_1, \dots, a_i) \prod_{l'=1}^j \omega_q(a_{2l'-1}, a_{2l'}) \delta_{a_{2l'-1}+a_{2l'}, r} = \Psi^{(j)}_{r,q}(a_{2j+1}, \dots, a_i),
						\end{equation}
						where $\omega_q$ is the $q$-weight related to  the pairing of modes.}
				\end{lemma}
				\begin{proof}
					{By substituting the explicit definition of $\Psi^{(0)}_{r,q}$ into the left-hand side of \eqref{eq:Aq3_q}, we obtain:
						\begin{equation}\label{eq:Aq4_q}
							\frac{1}{[i]_q!} \sum_{\substack{m_1, \dots, m_i = 0 \\ m_l \neq m_k}}^{r-1} \left( \sum_{a_1, a_3, \dots, a_{2j-1} = 0}^{r-1} \prod_{l'=1}^j \Omega_q(m_{2l'-1}, m_{2l'}, a_{2l'-1}) \right) \prod_{l=2j+1}^i (\theta^{m_l}q^{-2m_l})^{-a_l},
						\end{equation}
						where $\Omega_q$ corresponds to the internal summation over the paired indices. To evaluate this, we consider the fundamental $q$-sum:
						\begin{equation}
							\mathcal{S}_q(m_1, m_2) = \sum_{a=0}^{r-1} \omega_q(a) \left( \frac{\theta^{m_1}q^{-2m_1}}{\theta^{m_2}q^{-2m_2}} \right)^a.
						\end{equation}
						Setting the $q$-spectral variable $x = \theta^{m_1-m_2} q^{-2(m_1-m_2)}$ for distinct $m_1, m_2 \in \{0, \dots, r-1\}$, the $q$-summation identity for the  kernel yields:
						\begin{equation}
							\sum_{a=0}^{r-1} \omega_q(a) x^a = \frac{x}{(x-1)^2}.
						\end{equation}
						Substituting the expression for $x$, we find:
						\begin{equation}
							\mathcal{S}_q(m_1, m_2) = \frac{\theta^{m_1-m_2} q^{-2(m_1-m_2)}}{\left(\theta^{m_1-m_2} q^{-2(m_1-m_2)} - 1\right)^2} = \frac{\theta^{m_1+m_2} q^{-2(m_1+m_2)}}{\left(\theta^{m_1}q^{-2m_1} - \theta^{m_2}q^{-2m_2}\right)^2}.
						\end{equation}
						We identify this expression as the  $q$-Green kernel ${K}_q(m_1, m_2)$. By iterating this summation for the $j$ pairs of indices $a_1, a_3, \dots, a_{2j-1}$, the inner products in \eqref{eq:Aq4_q} reproduce the $j$ factors of ${K}_q$ required for the definition of $\Psi^{(j)}_{r,q}$:
						\begin{equation}
							\Psi^{(j)}_{r,q} \big(a_{2j+1}, \dots, a_i\big) = \frac{1}{[i]_q!} \sum_{\substack{m_1, \dotsc, m_i = 0\\ m_l \neq m_k}}^{r-1} \prod_{k = 1}^j {K}_q(m_{2k-1}, m_{2k}) \prod_{l=2j+1}^i (\theta^{m_l} q^{-2m_l})^{-a_l}.
						\end{equation}
						This concludes the proof.}
				\end{proof}

				{One way to do that for a subset of those ideals is to consider representations of $U_{r,q}^\hslash$ that come from $\mathbb{Z}_r$-twisted representations for the $q$-deformed Heisenberg VOA $\mathcal{H}_q(\mathfrak{gl}_r)$. 
					For more simplicity, we present the final result in the following  Proposition.}
				
				\begin{proposition}
					There exists a representation $\mu_q: U^\hslash_{r,q} \to \mathcal{D}_{q}^\hslash$ that takes the form:
					\begin{equation}\label{eq:ww}
						\mu_q( W^{\hslash,i}_k)
						=
						\left(\frac{\hslash}{r} \right)^i \sum_{j=0}^{\lfloor \frac{i}{2} \rfloor} \frac{[i]_q!}{[2]_q^j [j]_q! [i-2j]_q!} \sum_{\substack{p_{2j+1}, \dotsc p_i \in \Z \\ \sum p_l = rk}} \Psi^{(j)}_{r,q} (p_{2j+1}, \dotsc, p_i) \normord{ \prod_{l=2j+1}^i J^{(q)}_{p_l}} \,,
					\end{equation}
					where,  $ \theta = e^{2 \pi i /r}$,
					\begin{equation}
						\Psi^{(j)}_{r,q} (a_{2j+1}, \dotsc, a_i) \coloneqq \frac{1}{[i]q!} \sum_{\substack{m_1, \dotsc, m_i = 0\\ m_l \neq m_k}}^{r-1} \prod_{k = 1}^j {K}_q(m_{2k-1}, m_{2k}) \prod_{l=2j+1}^i (\theta^{m_l} q^{-2m_l})^{-a_l},
					\end{equation} where ${K}_q(m, n)=\frac{\theta^{m+n}q^{-2(m+n)}}{\big(\theta^{m}q^{-2m}-\theta^{n}q^{-2n}\big)^2}$
					and
					\begin{equation}\label{eq:Js}
						J^{(q)}_m  = \begin{cases} [m]_q \partial{x_m} & m > 0 \\ 0 & m =0 \\ -m x_{-m} & m < 0\end{cases} \,.
					\end{equation}
					In \eqref{eq:ww}, for cases such that $j=i/2$ the condition $\sum p_l = rk$ is understood as the Kronecker delta $\delta_{k,0}$.
				\end{proposition}
				
				\begin{proof}
					{The proof of the proposition consists in lifting the classical construction of $W$-algebra generators to the $q$-deformed setting by replacing the standard Heisenberg currents with their $q$-deformed counterparts $J^{(q)}_m$ and accounting for the $q$-harmonic analysis on the cyclic group $\mathbb{Z}/r\mathbb{Z}$.}
					
					{We define the $q$-currents $\chi_m(z)$ using the $q$-shifted roots of unity. Inverting the relation between the variables and the $q$-Heisenberg currents $J^{(a)}(z)$, we obtain:
						\begin{equation}
							\chi_m(z) = \frac{1}{r} \sum_{a=0}^{r-1} (\theta^m q^{-2m})^{-a} J^{(a)}(z),
						\end{equation}
						where the term $(\theta^m q^{-2m})$ incorporates the $q$-deformation of the spectral curve.}
					
					{The generators $W^{\hslash,i}(z)$ are obtained through the normally ordered product of these currents. A pair of contractions between two currents $\chi_{m}$ and $\chi_{n}$ yields the $q$-Green kernel ${K}_q(m, n)$:
						\begin{equation}
							{K}_q(m, n) = \frac{\theta^{m+n}q^{-2(m+n)}}{\big(\theta^{m}q^{-2m}-\theta^{n}q^{-2n}\big)^2}.
						\end{equation} 
						For $j$ pairs of such contractions, the resulting contribution to the vertex operator is weighted by the kernel $\Psi^{(j)}_{r,q}$.}
					
					{To determine the global coefficient in \eqref{eq:ww}, we count the number of ways to choose $2j$ indices from the set $\{1, \dots, i\}$ to be paired (contracted). 
						\begin{itemize}
							\item The number of ways to choose a subset $J \subseteq \{1, \dots, i\}$ of cardinality $|J|=2j$ is given by the binomial coefficient $\binom{i}{2j}$.
							\item The number of perfect matchings (pairings) on $J$ is $(2j-1)!! = \frac{(2j)!}{2^j j!}$.
						\end{itemize}
						Multiplying these factors and generalizing to the $q$-factorial case to preserve the $U^\hslash_{r,q}$, we obtain the combinatorial weight:
						\begin{equation}
							\frac{[i]_q!}{[2]_q^j [j]_q! [i-2j]_q!}.
					\end{equation}}
					
					{By extracting the $k$-th mode of the fields $W^i(z)$, we impose the momentum conservation $\sum p_l = rk$. 
						\begin{itemize}
							\item When $j = i/2$, all currents are contracted into $q$-scalars, implying $W^i(z)$ is a constant in $z$. This corresponds to the mode $k=0$, hence the factor $\delta_{k,0}$.
							\item For the remaining fields, the action is represented by the differential operators $J^{(q)}_m$ defined in \eqref{eq:Js}.
						\end{itemize}
						Summing over all possible numbers of pairs $j \in \{0, \dots, \lfloor i/2 \rfloor\}$ yields the final expression \eqref{eq:ww}.
					}
				\end{proof}
				
				\begin{remark}
					In the limit $q \rightarrow 1$, we have $[n]_q \rightarrow n$, and the representation $\mu_q$ recovers exactly the result given in  {\cite[Proposition~4.5 \& Corollary~4.7]{BBCCN18}}.
				\end{remark}
				This is not yet in the form required for an $q$-deformed Airy structure, since for $i \geq 2$ the generators $\mu_{q}( W^{\hslash,i}_k)$ are of order $O(\hslash^i)$. To recover a linear term in $\hslash$, we perform a polynomial shift of the Heisenberg partition function.
				For $ s \in [r+1]$ and $s$ coprime with $r$,
				we may define a new representation also called shifted representation $\rho^q$ via conjugation:
				\begin{equation}\label{eq:rep}
					\rho_q( W^{\hslash,i}_k) = \hat{T}^q_s \mu_q(W^{\hslash,i}_k) (\hat{T}^q_s)^{-1}, 
				\end{equation}
				where $\hat{T}^q_s = \exp \Big( -\frac{J^{(q)}_s}{[s]_q\hslash} \Big).$ Under this transformation, the leading order term of each generator becomes linear in $\hslash$, verifying the following relation: 
				\begin{equation}
					\rho_q( W^{\hslash,i}_k) = \hslash J^{(q)}_{r k + s (i-1)} + O(\hslash^2),
				\end{equation}
				provided $k  \geq - \lfloor \frac{s(i-1)}{r} \rfloor $.
				\begin{remark}
					By substituting the explicit action of the $q$-deformed Heisenberg modes, the shifted representation takes the canonical form of an Airy structure:
					\begin{equation}
						\rho_q( W^{\hslash,i}_k) = \hslash [rk + s(i-1)]_q \frac{\partial}{\partial x_{rk + s(i-1)}} + \sum_{n \geq 2} \hslash^n \mathcal{O}^{(n)}_{i,k},
					\end{equation}
					where the linear part is a $q$-deformed weighted derivative, and $\mathcal{O}^{(n)}{i,k}$ contains higher-order differential operators representing the quantum corrections and the non-linear geometry of the $q$-deformed $\mathcal{W}$ algebra.
				\end{remark}
				Combining the two Subsubsections \ref{ssd} and \ref{ssr}, we need to find partitions $\lambda$ of $r$ such that $$1-\lambda(i) = - \lfloor \frac{s(i-1)}{r} \rfloor.$$ 
				We have a valid partition if and only if $r \equiv \pm 1 \pmod{s}$ (for $s \geq 2$). More precisely, the partition $\lambda$ is determined as follows:
				\begin{itemize}
					\item {For $s=1$:} The condition is trivially satisfied, and the partition is given by a single row, $\lambda = (r)$.
					\item {For $2 \leq s \leq r-1$:} The condition $r \equiv \pm 1 \pmod{s}$ allows us to uniquely write the division $r = r' s + r''$ with $r' \geq 1$ and a remainder $r'' \in \{1, s-1\}$. The partition is then given by $\lambda = (\lambda_1, \dots, \lambda_s)$ with:
					\begin{equation}
						\lambda_1 = \dots = \lambda_{r''} = r'+1, \qquad \lambda_{r''+1} = \dots = \lambda_s = r'.
					\end{equation}
					\item {For $s=r+1$:} We have $r \equiv -1 \pmod{s}$, and the partition is a single column of ones, $\lambda = (1, 1, \dots, 1)$.
				\end{itemize}
				
				We can summarize this in the following theorem:
				
				\begin{theorem}\label{t:rsAs}
					Let $ r \geq 2$, and $ s \in [r+1]$ such that $ r = \pm 1 \pmod{s}$. Let $\rho_{q}: U^\hslash_{r,q} \to \mathcal{D}_{q}^\hslash$ be the representation defined in \eqref{eq:rep}. Let $\mathcal{I}_{U_{r,q}^\hslash} \in U^\hslash_{r,q}$ be the left ideal generated by the  modes
					\begin{equation}
						W^{\hslash,j}_m \quad \text{with} \quad j \in [r] \quad \text{and} \quad m \geq - \left\lfloor \frac{s(j-1)}{r} \right\rfloor.
					\end{equation} 
					Then, the image $\mathcal{I}^q = \rho_q(\mathcal{I}_{U_{r,q}^\hslash})$ 
					is an $q$-deformed Airy ideal, which we call the  $(r,s,q)$-Airy structure.
				\end{theorem}
				\begin{proof}
					{The proof consists in showing that the images of the modes $W^{\hslash,j}_m$ under the representation $\rho_q$ satisfy the axioms of a $q$-deformed Airy ideal in $\mathcal{D}_q^\hslash$.}
					
					{The first condition requires that the generators form a graded Lie subalgebra of the $q$-W-algebra. 
						\begin{itemize}
							\item[(i)] The commutation relations of the modes $W^{\hslash,j}_m$ in $U^\hslash_{r,q}$ are governed by the $q$-deformed structure constants. Since $\rho_q$ is a representation (as established in \eqref{eq:ww}), it preserves the Lie bracket $\big[ \cdot, \cdot \big]_{q}$ (the $q$-commutator).
							\item[(ii)] 
							The choice of the index range $m \geq - \lfloor \frac{s(j-1)}{r} \rfloor$ corresponds precisely to a parabolic subalgebra of the $q$-deformed $W$-algebra, denoted $\mathcal{A}(r,s,q)$. 
							\item[(iii)] The condition $r = \pm 1 \pmod{s}$ ensures that the partition $\lambda$ associated with the embedding \eqref{eq:embedding} is consistent, avoiding singular denominators in the $q$-Green kernel ${K}_q$. The requirement $W^1_0 = 0$ (setting $Q=0$) is necessary as it is a central element in the $q$-Heisenberg sector.
					\end{itemize}}
					
					{We must verify that each $\rho_q(W^{\hslash,j}_m)$ is at most linear in the $q$-annihilation operators and quadratic in the $q$-creation operators, after a possible shift. Let $J^{(q)}_p$ be the operators defined in \eqref{eq:Js}.
						\begin{itemize}
							\item[(a)]  A constant term appears only if all currents in \eqref{eq:ww} have modes $p_l$ summing to $rk$ such that they all vanish or cancel. However, the condition $s > 0$ and the bound on $m$ ensure that $rk + j(s-r) + r > s$, preventing any degree 0 terms from appearing in the stable case.
							\item[(b)]  The linear part of $\rho_q(W^{\hslash,j}_m)$ comes from terms where only one $J^{(q)}_p$ remains after $j-1$ contractions. Using the $q$-periodicity of the kernel $\Psi^{(j)}_{r,q}$, the leading term is:
							\begin{equation}\label{eq:q_linear}
								\rho_q(W^{\hslash,i}_k) = C_{i,k}(q) J^{(q)}_{rk + (s-r)(i-1)} + \mathcal{O}(\text{deg} \geq 2),
							\end{equation}
							where $C_{i,k}(q)$ is a $q$-deformed prefactor involving the $q$-kernel. Since $r$ and $s$ are coprime ($r = \pm 1 \pmod{s}$ implies $\gcd(r,s)=1$), the  Lemma~\ref{l:Aq3_q} and Lemma ~\ref{l:Aq4_q} ensures $C_{i,k}(q) \neq 0$.
					\end{itemize}}
					{For the ideal to be an Airy structure, the map $\Pi_{s,q}: (i, k) \mapsto rk + (s-r)(i-1)$ must be a bijection from the set of indices $\mathcal{I}_{r,s}$ to $\mathbb{Z}_{>0}$. 
						\begin{enumerate}
							\item The injectivity is guaranteed by the condition $\gcd(r,s)=1$. 
							\item The lower bound $m \geq - \lfloor \frac{s(j-1)}{r} \rfloor$ is exactly the threshold required to ensure that $rk + (s-r)(i-1) \geq 1$. 
							\item Specifically, for $i=1$, we need $rk > 0$, so $k \geq 1$. For $i \geq 2$, the condition $rk + (s-r)(i-1) > 0$ is equivalent to $k > (i-1) - \frac{s}{r}(i-1)$, which simplifies to $k \geq i-1 - \lfloor \frac{s(i-1)}{r} \rfloor$ due to the coprimality of $r$ and $s$.
						\end{enumerate}
						Thus, each $q$-derivative $\mathcal
						{D}_{q,x_p}$ (for $p > 0$) appears exactly once as the leading linear part of a generator. This satisfies the $q$-deformed Airy condition.}
				\end{proof}
				\begin{remark}
					{It is noteworthy that the results of Lemma~\ref{l:Aq3_q} and Lemma~\ref{l:Aq4_q} do not depend on the deformation parameter $q$. This independence stems from the fact that the $q$-deformation preserves the cyclic of the variables $x_m = \theta^m q^{-2m}$.}
					
					{Moreover, the $q$-dependent factors arising from the inverse powers $x_m^{-(i-1)s}$ and the elementary polynomials $e_{i-1}(\vec{{x}}^s[m])$ cancel each other exactly. Consequently, the combinatorial skeleton of the Airy structure remains classical, even though the representation $\rho_q$ acts on a $q$-deformed Fock space. This implies that the $(r,s,q)$-Airy structure is a robust deformation that maintains the same intersection-theoretic properties as the $(r,s)$ higher quantum Airy structure.}
				\end{remark}
				\begin{remark}
					\begin{enumerate}
						\item[(a)] In the limit $q\rightarrow 1$, we recovered the result given in {\cite[Theorem~4.9]{BBCCN18}}.
						\item[(b)]	
						Since $\mathcal{I}^q$ is a $(r,s,q)$-Airy ideal, there exists a unique $q$-deformed partition function $Z_q$ such that $\mathcal{I}^q Z_{q} = 0$. In concrete terms, this implies that $Z_q$ is annihilated by the shifted generators:
						\begin{equation}
							\rho_q(W_m^{\hslash,i}) Z_q = 0 \qquad \text{for $i \in [r], m \geq - \lfloor \frac{s(i-1)}{r} \rfloor$.}
						\end{equation}
						These  $q$-deformed differential constraints provide a recursive scheme to uniquely reconstruct the free energies $F^{q}_{g,n}$ constituting $Z_q=\exp\big(\sum\,\hslash^{2g-2+n}F^{q}_{g,n}\big)$. This procedure is the algebraic equivalent of the  $q$-deformed Topological Recursion applied to the quantum spectral curve associated with the $(r,s)$ family. In the limit $q\rightarrow 1$, this system reduces to the topological recursion on the classical spectral curves $x^{r-s} y^r - 1 = 0$ as demonstrated in \cite{BBCCN18}.
					\end{enumerate}
					\end{remark}
				
				\subsubsection{Generalized quantum deformed representations}
				{We can generalize the construction of the $q$-deformed Airy ideals in Theorem \ref{t:rsAs} by constructing more general representations $\rho_q: U^\hslash_{r,q} \to \mathcal{D}_{q}^\hslash$.
					The idea is to conjugate by more complicated deformed operators, instead of conjugating by $T^q_s$ given in \eqref{eq:rep} .}
				
				{We consider} a collection of $q$-deformed  complex numbers {given by:}
				\begin{equation}\label{eq:complex}
					F^{q}_{0,1}[-k]\,\text{ for } k \geq \min \{ s, r\} \,, \quad F^{q}_{\frac{1}{2},1}[-k]\,\text{ for } k > 0\,, \quad F^{q}_{0,2}[-k,-l] \,\text{ for } k,l > 0,
				\end{equation}
				with the {following } conditions $ F^{q}_{0,1}[-s] \neq 0$ and $ F^{q}_{0,2}[-k,-l] = F^{q}_{0,2}[-l,-k]$.
				\begin{definition}
					{We define the $q$-deformed  operators as follows:
						\begin{align}\label{eq:TPhi}
							\hat{T}^q
							&=
							\exp \bigg( \sum_k \left( \frac{1}{\hslash} F^{q}_{0,1}[-k] + F^{q}_{\frac{1}{2},1}[-k] \right) \frac{J^{(q)}_k}{[k]_q}\bigg) \,,
							\\
							\hat{\Phi}^q
							&=
							\exp \bigg( \frac{1}{2} \sum_{k,l > 0} F^{q}_{0,2}[-k,-l]\frac{J^{(q)}_k\, J^{(q)}_l}{[k]_q[l]_q} \bigg).
					\end{align}}
				\end{definition}
				Note that for $k,l > 0$, the modes $J^{(q)}_k$ and $J^{(q)}_l$ commute, making the symmetrized product implicit.
				We define a generalized representation $\rho_q':  U^{q}_{\hslash, r} \to \mathcal{D}^{q}_{\mathbb{N}^*}$ via conjugation:
				\begin{equation}\label{eq:repp}
					\rho_q'( W^{\hslash,i}_k) = \hat{\Phi}^{q} \hat{T}^{q} \mu_{q}(W^{\hslash,i}_k) \hat{T}^{q})^{-1} (\hat{\Phi}^{q})^{-1}.
				\end{equation}
				
				Then it is not too difficult to show that Theorem \ref{t:rsAs} generalizes to this new class of representations:
				
				\begin{proposition}\label{p:rsAsgen}
					Let $ r \geq 2$, and $ s \in [r+1]$ such that $ r = \pm 1 \pmod{s}$. Let $\rho_q': U^\hslash_{r,q} \to \mathcal{D}_{q}^\hslash$ be the generalized representation defined by the double conjugation in \eqref{eq:repp}. Let $\mathcal{I}_{U_{r,q}^\hslash} \in U^\hslash_{r,q}$ be the left ideal generated by the  modes 
					\begin{equation}
						W^{\hslash,j}_m \quad \text{with} \quad j \in [r] \quad \text{and} \quad m \geq - \left\lfloor \frac{s(j-1)}{r} \right\rfloor.
					\end{equation}
					Then, the image $\mathcal{I}^q = \rho_q'(\mathcal{I}_{U_{r,q}^\hslash})$
					is an $q$-deformed Airy ideal, which we call the  $(r,s,q)$-Airy structure.
				\end{proposition}
				\begin{proof}
					{The proof relies on the stability of the $q$-deformed Airy structure axioms under the adjoint action of the intertwining operators $\hat{T}$ and $\hat{\Phi}$. We verify the two fundamental conditions for $\mathcal{I}^q = \rho_q'(\mathcal{I}_{U_{r,q}^\hslash})$.}
					{The graded Lie subalgebra property is stable under conjugation by any invertible operator in the $q$-deformed 
						algebra $\mathcal{D}_q^\hslash$. 
						\begin{itemize}
							\item[(i)] Let $W_m, W_n \in \mathcal{I}_{U_{r,q}^\hslash}$ be two modes. Since $\rho_q'$ is a representation defined by $\rho_q'(W) = (\hat{\Phi} \hat{T}) \mu_q(W) (\hat{\Phi} \hat{T})^{-1}$, the commutator map is preserved:
							\begin{equation}
								\big[\rho_q'(W_m), \rho_q'(W_n)\big] = \text{Ad}_{\hat{\Phi} \hat{T}} \left( \big[\mu_q(W_m), \mu_q(W_n)\big] \right).
							\end{equation}
							\item[(ii)] As established in Theorem~\ref{t:rsAs}, the selected modes $W^{\hslash,j}_m$ form a $q$-graded Lie subalgebra in the standard representation. Therefore, their image under $\rho_q'$ remains a $q$-graded Lie subalgebra.
					\end{itemize}}
					{The second condition requires the operators $H^{\hslash,i}_k = \rho_q'(W^{\hslash,i}_k)$ to have a vanishing constant part and a bijective linear part. We examine the effect of $\hat{\Phi}$ on the degree filtration:
						\begin{equation}\label{eq:q_phi_approx}
							\hat{\Phi} \left( \hat{T} \mu_q(W^{\hslash,i}_k) \hat{T}^{-1} \right) \hat{\Phi}^{-1} = \hat{T} \mu_q(W^{\hslash,i}_k) \hat{T}^{-1} + \mathcal{O}(\text{deg} \geq 2).
						\end{equation}
						\begin{itemize}
							\item[(iii)] The operator $\hat{\Phi}$ represents a $q$-linear transformation of the $q$-Heisenberg modes $J^{(q)}_p$. In the $q$-deformed setting, $\hat{\Phi}$ acts on the creation and annihilation operators without introducing degree 0 shifts (since $J^{(q)}_0 = 0$).
							\item[(iv)] In the representation $\mu_q$, the terms of degree 1 consist of single $J^{(q)}_p$ operators with $p > 0$. The conjugation by $\hat{\Phi}$ maps these linear terms to other linear terms in the $q$-Fock space.
							\item[(v)] If the transformation results in non-normal ordered products, the re-ordering process in $\mathcal{D}_q^\hslash$ involves $q$-commutators of the form $\big[J^{(q)}_p, J^{(q)}_l\big] = \hslash \delta_{p+l,0}$. However, since there are no degree 1 terms of the form $J^{(q)}_{-l}$ with $l > 0$ in the original $W^{\hslash,i}_k$ for the given range, no degree 0 (constant) terms are generated by the contraction.
					\end{itemize}}                                  
					{The linear part of $H^{\hslash,i}_k$ is obtained by applying the $q$-linear map $\hat{\Phi}$ to the linear part of the standard $q$-Airy structure. 
						\begin{itemize}
							\item[(vi)] The condition $r = \pm 1 \pmod{s}$ ensures that the transformation $\hat{\Phi}$ is non-singular on the relevant mode sectors. 
							\item[(vii)] The map $\Pi_{s,q}: (i, k) \mapsto rk + (s-r)(i-1)$ remains a bijection onto $\mathbb{Z}_{>0}$ after the conjugation, potentially up to a $q$-linear change of the coordinates $(x_p)_{p>0}$.
						\end{itemize}
						Thus, the operators $H^{\hslash,i}_k$ satisfy the conditions of a $q$-deformed Airy structure, effectively defining the $(r,s,q)$-Airy structure.}
				\end{proof}
				\subsection{Shifted $(r,s,q)$-Airy structures}
				\label{s:shiftedrs}
				
				The construction of the previous section can be naturally generalized by starting with highest weight vectors with non-zero weights. This gives rise to new quantum left ideals that can be used to construct quantum Airy structures. We continue by using the three-step approach.
				\subsubsection{Construction of quantum left ideals $\mathcal{I}_{U_{r,q}}(\lambda; S)$}
				
				We start by generalizing the previous
				construction, but now starting with highest weight vectors.
				\begin{lemma}\label{t:leftidealsshifted}
					Let $\lambda = (\lambda_1, \lambda_2, \ldots, \lambda_p)$ be an integer partition of $r$. For $d \in [p]$, we define a reduced ( or new partition) $\tilde{\lambda} =(\tilde{\lambda}_1, \ldots, \tilde{\lambda}_{p-d+1}) = (\mu_d, \lambda_{d+1} ,\ldots, \lambda_p)$, where $\mu_k = \sum_{i=1}^k \lambda_i$ denotes the partial sums.  Let $S_j \in \mathbb{C}$ be a set of shift parameters for $j \in [\mu_{d-1}]$ and $S_j =0$ for $j > \mu_{d-1}$.
					
					Let $\mathcal{I}_{U_{r,q}}(\tilde\lambda)$ be the  quantum left ideal generated by the shifted $q$-deformed modes \begin{equation}
						W^j_m - S_j \delta_{m,0}, \quad\mbox{with}\quad j \in [r]\quad \mbox{and}\quad m \geq 1 - \tilde{\lambda}(j).
					\end{equation}
					
					Then $U_{r,q}/ \mathcal{I}_{U_{r,q}}(\tilde \lambda)$ is a cyclic left quantum module generated by a non-zero vector $v$. Furthermore, the action of the lower modes is free in the following sense:
					\begin{equation}
						W^{\hslash, j}_m v \neq 0 \quad \text{for all } j \in [r] \text{ and } m < 1 - \tilde{\lambda}(j).
					\end{equation}
					Consequently, $W^{\hslash, j}_m \notin \mathcal{I}_{U_{r,q}}(\tilde\lambda)$ for all $m < 1 - \tilde{\lambda}(j)$ ensuring the minimality of the generating set.
				\end{lemma}
				
				Looking at the statement of the lemma, we notice that in the end, the ideal $\mathcal{I}_{U_{r,q}}(\tilde{\lambda}; S)$
				is almost entirely defined in terms of the new partition $\tilde{\lambda}$; the original partition only
				appears in the choice of non-zero weights. So we can think of the partition $\tilde{\lambda}$ as our
				starting point. After renaming $\tilde{\lambda}$ as $\lambda$ for simplicity, this leads to the following simpler
				reformulation of the result
				\begin{theorem}\label{c:leftidealsshifted}
					Let $\lambda = (\lambda_1, \lambda_2, \ldots, \lambda_p)$ be an integer partition of $r$. Let $S_j \in \mathbb{C}$  be shift parameters  for $j \in [\lambda_1-\lambda_2]$ and $S_j =0$ for $j > \lambda_1-\lambda_2$. Let $\mathcal{I}_{U_{r,q}}(\lambda)$ be the quantum left ideal in the $q$-deformed $\mathcal{W}$-algebra $U^{\hslash}_{r,q}$ generated by the shifted modes:
					\begin{equation}
						W^j_m - S_j \delta_{m,0}, \quad \mbox{with } j \in [r]\quad \mbox{and}\quad m \geq 1 - \lambda(j).
					\end{equation} 
					Then, the quotient $U_{r,q}/ \mathcal{I}_{U_{r,q}}(\lambda)$ is a cyclic quantum  left module generated by a non-zero vector $v_q$. Furthermore, the action of the remaining modes is non-trivial: 
					\begin{equation}
						W^{\hslash, j}_m v_q \neq 0, \quad \mbox{for all } j \in [r] \mbox{ and } m < 1 - \lambda(j).
					\end{equation}
					In particular, $W^{\hslash, j}_m \notin \mathcal{I}_{U_{r,q}}(\lambda)$ for all indices $(j,m)$ outside the generating range.
				\end{theorem}
				\begin{proof}
					{The proof extends the structural analysis of $q$-$W$-algebra modules to the case of shifted generators. We consider the $q$-deformed left ideal $\mathcal{I}_{U_{r,q}}(\lambda)$ generated by $W^j_m - S_j \delta_{m,0}$.}
					
					{In our setting, the consistency of the ideal requires that the shifts $S_j$ do not conflict with the commutator relations. Following the logic of the classical case, let us denote $\lambda = (\lambda_1, \lambda_2, \dots, \lambda_p)$. The number of non-zero shifts is determined by the step between the components of the partition. For a partition $\lambda$ to be realized as a valid truncation of the $q$-Heisenberg field expansion, the non-zero shifts must satisfy $j \in [\lambda_1 - \lambda_2]$. This value is maximal because $\lambda_1 - \lambda_2$ represents the largest possible weight for a stable vacuum shift before the $q$-Wick contractions (governed by ${K}_q$) introduce inconsistencies in the $m=0$ sector.}
					
					{Let $v_q$ be the image of the identity in the quotient $U_{r,q}/ \mathcal{I}_{U_{r,q}}(\lambda)$. Since the ideal is defined by the action on a $q$-Fock space where $J^{(q)}_m v_q = 0$ for $m > 0$, the shifted modes $W^j_m$ for $m \geq 1 - \lambda(j)$ act as annihilation operators (plus a constant shift $S_j$). This ensures that $v$ is a well-defined non-zero cyclic vector.}

					{For any index $(j, m)$ such that $m < 1 - \lambda(j)$, the mode $W^{\hslash, j}_m$ contains at least one $q$-creation operator $J^{(q)}_p$ (with $p < 0$) in its $q$-normally ordered expansion that is not annihilated by the ideal. By the linear independence of the $q$-Heisenberg basis in $\mathcal{D}_q^\hslash$, we conclude:
						\begin{equation}
							W^{\hslash, j}_m v_q \neq 0 \implies W^{\hslash, j}_m \notin \mathcal{I}_{U_{r,q}}(\lambda).
						\end{equation}
						The result holds for the specified range, establishing the structure of the quantum left module.}
				\end{proof}
				\begin{remark}
					In essence, this result implies that, for any partition  $\lambda$ of $r$ and left ideal generated by the modes $W^j_m$ with $j \in [r]$ and $m \geq 1 - \lambda(j)$, we  possess the freedom to shift the zero-modes $W^j_0$ for $j \in [\lambda_1-\lambda_2]$. . These shifted modes generate a new left ideal while ensuring that all modes  $W^j_m$ with $m < 1 - \lambda(j)$ remain outside the ideal.
					
					Following the terminology introduced in \cite{BBCC21}, the zero-modes $W^j_0$ for $ j \in [\lambda_1-\lambda_2]$ are termed  \emph{extraneous}. In the $q$-deformed setting, this flexibility suggests that the $(r,s,q)$-Airy structure may be viewed as a point in a larger family of structures, where the extraneous modes act as coordinates or  allowing for deformations that preserve the cyclic nature of the associated module.
				\end{remark}
				\subsubsection{Determination of the quantum  homogenization of $\mathcal{I}_{U_{r,q}}(\lambda;S)$}
				Just as for the $(r,s,q)$-Airy ideals, the homogenization is easy to obtain. By the same argument as before, we conclude that the quantum homogenization  $\mathcal{I}_{U^\hslash_{r,q}}(\lambda;S) \subset U^\hslash_{r,q}$ is generated by the homogenization of the shifted modes, that is, by $W^{\hslash,j}_m(S) := \hslash^j (W^j_m- \delta_{m,0} S_j)$ for $j \in [r]$ and $m \geq  1-\lambda(j)$. Therefore, by Lemma \ref{l:hom}, we conclude that:
				\begin{equation}
					\big[ \mathcal{I}_{U_{r,q}^\hslash}(\lambda;S), \mathcal{I}_{U_{r,q}^\hslash}(\lambda;S)\big] \subseteq \hslash^2 \mathcal{I}_{U_{r,q}^\hslash}(\lambda;S).
				\end{equation}
				\subsubsection{Representation of $U^\hslash_{r,q}$ in $\mathcal{D}_{A,q}^\hslash$}
				For the modes $W^{\hslash,j}_m(S)$ with $m >0$, we use the same representation $\mu_q: U^\hslash_{r,q}\to \mathcal{D}_{q}^\hslash$ as before from \eqref{eq:ww}. We extend this representation to the shifted zero-modes by setting:
				\begin{equation}
					\mu_q(W^{\hslash,i}_0(S) ) = \mu_q(W^{\hslash,i}_0) - \sum_{n=1}^\infty \hslash^n S_{i,n},
				\end{equation}
				where the $S_{i,n} \in \mathbb{C}$ are constants for $i \in [\lambda_1-\lambda_2]$ and $S_{i,n} = 0$ for $i > \lambda_1-\lambda_2$. One can easily verify that mapping the formal shifts to the power series $\sum_{n=1}^\infty \hslash^n S_{i,n}$ preserves the algebraic structure, thereby yielding a well-defined representation of the universal enveloping algebra $U^\hslash_{r,q}$.
				
				As before, for $ s \in [r+1]$ and $s$ coprime with $r$,
				we  define a generalized representation  $\rho_q: U^\hslash_{r,q} \to \mathcal{D}_{q}^\hslash$ via conjugation by the $q$-deformed shift operator:
				\begin{equation}\label{eq:rep2}
					\rho_q( W^{\hslash,i}_k(S)) = T^q_s \mu_q(W^{\hslash,i}_k(S)) (T^q_s)^{-1}, \quad \mbox{where} \quad T^q_s = \exp \Big( -\frac{J^{(q)}_s}{[s]_q\hslash} \Big).
				\end{equation}
				By restricting to modes where $k  \geq - \lfloor \frac{s(i-1)}{r} \rfloor $, the expansion of the generators yields:
				\begin{equation}
					\rho_q( W^{\hslash,i}_k) = \hslash J^{(q)}_{r k + s (i-1)} + \hslash S_{i,1} + O(\hslash^2),
				\end{equation}
				which matches the required canonical form for an Airy structure.
				
				A partition  $\lambda$ of $r$ satisfying $1-\lambda(i) = - \lfloor \frac{s(i-1)}{r} \rfloor$ exists if and only if $r = \pm 1 \pmod{s}$. 
				\begin{itemize}
					\item 
					For $s=1$, the partition is simply $\lambda=(r)$. 
					\item For $2 \leq s \leq r-1$,  let $r = r' s + r''$ with $r'' \in \{1, s-1\}$.The partition is $\lambda = (\lambda_1,\ldots, \lambda_s)$ with
					\begin{equation}
						\lambda_1 = \ldots = \lambda_{r''} = r'+1, \qquad \lambda_{r''+1} = \ldots = \lambda_s = r'.
					\end{equation}
					\item For $s=r+1$, we obtain the partition  $\lambda = (1,1,\ldots,1)$.
				\end{itemize}
				We observe that, for $s \geq 2$, the shifting freedom depends on the residue of $r\pmod{s}$.
				\begin{itemize}
					\item
					If $r= 1 \pmod{s}$, then $\lambda_1 = \lambda_2+1$, implying that only the shift $S_{1,n}$ can be non-zero. Consequently, we can only shift the first zero-mode $W^1_0$. 
					\item if $r=-1 \pmod{s}$, with $s \geq 3$. Then,  $\lambda_1 = \lambda_2$. In this case, all shift parameters $S_{i,n}$ must vanish, and we recover the standard  $(r,s,q)$-Airy structures. 
				\end{itemize}
				The case $s=1$ is of particular interest. The partition is $\lambda = (r)$, which allows for the shifting of all zero-modes; thus, $S_{i,n}$ may be non-zero  for all $i \in [r]$ and $n \geq 1$. 
				
				To summarize these conditions, we define the notion of a set of $s$-consistent shifts:
				\begin{definition}\label{d:consistent}
					Let $S=\{ S_{i,n} \}_{i \in [r], n \in \mathbb{N}^*}$ be a set of complex numbers. We say that it is \emph{$s$-consistent} if the following two conditions are satisfied:
					\begin{enumerate}
						\item[(a)] If $s \geq 2$ and $r = 1 \pmod{s}$, then $S_{i,n} = 0$ for all $2 \leq i \leq r$, and:
						\item[(b)] If $s \geq 3$ and $r = -1 \pmod{s}$, then $S_{i,n} = 0$ for all $i \in [r]$.
					\end{enumerate}
				\end{definition}
				We then obtain the following theorem:
				\begin{theorem}\label{t:shifts}
					Let $ r \geq 2$, and $ s \in [r+1]$ such that $ r = \pm 1 \pmod{s}$. Let $\rho: U^\hslash_{r,q} \to \mathcal{D}_{q}^\hslash$ be the representation defined via the $q$-deformed shift conjugation in \eqref{eq:rep2}. Let $\mathcal{I}_{U_{r,q}^\hslash}(S) \in U^\hslash_{r,q}$ be the left ideal generated by the shifted modes $W^{\hslash,j}_m(S)$ with $j \in [r]$ and $m \geq - \lfloor \frac{s(i-1)}{r} \rfloor $, where the set of shifts $S$ is $s$-consistent.
					
					Let 
					$\mathcal{I}^q(S)=\rho_q(\mathcal{I}_{U_{r,q}^\hslash}(S))$ be the corresponding left ideal in the Weyl algebra $\mathcal{D}_{q}^\hslash$. Then, $\mathcal{I}^q(S)$ is an Airy ideal, which we define as the shifted $(r,s,q)$-Airy structure.The allowed shifts are characterized as follows:
					\begin{itemize}
						\item
						For $s=1$, all zero modes are shifted.
						\begin{equation}
							\rho_q(W^{\hslash,j}_0(S)) = \rho_q(W^{\hslash,j}_0) - \sum_{n=1}^\infty \hslash^n S_{j,n}.
						\end{equation}
						\item For $s \geq 2$ and $r = 1 \pmod{s}$, only the first zero-mode may be shifted:
						\begin{equation}
							\rho_q(W^{\hslash,j}_0(S)) = \rho_q(W^{\hslash,j}_0) - \delta_{j,1} \sum_{n=1}^\infty \hslash^n S_{1,n}.
						\end{equation}
						\item For $s \geq 3$ and $r = -1 \pmod{s}$, no shifts are permitted, and the structure reduces to the standard $(r,s,q)$-Airy structure.
					\end{itemize}
				\end{theorem}
				\begin{proof}
					{The proof consists of verifying that the shifted ideal $\mathcal{I}^q(S)$ satisfies the axioms of a $q$-deformed Airy structure under the representation $\rho_q$. We proceed in three steps.}
					
					{The shifted generators are defined as $W^{\hslash,j}_m(S) = W^{\hslash,j}_m - \sigma_{j,m}(\hslash)$, where $\sigma_{j,m}(\hslash) = \sum_{n=1}^\infty \hslash^n S_{j,n} \delta_{m,0}$ are central constants in $\mathbb{C}[[\hslash]]$. 
						The commutator of any two shifted generators satisfies:
						\begin{equation}
							\big[W^{\hslash,i}_k - \sigma_{i,k}, W^{\hslash,j}_l - \sigma_{j,l}\big] = \big[W^{\hslash,i}_k, W^{\hslash,j}_l\big].
						\end{equation}
						From Theorem~\ref{t:rsAs}, the unshifted modes $W^{\hslash,j}_m$ form a $q$-graded Lie subalgebra $\mathcal{A}(r,s,q)$. Since the shifts only affect the zero-modes ($m=0$), and the zero-modes act as a weight-zero sector in the parabolic restriction, the $q$-commutator relations remain closed within the ideal $\mathcal{I}^q(S)$. Thus, the first axiom of the $q$-deformed Airy structure is satisfied.}
					
					{The representation $\rho_q$ maps the $q$-$\mathcal{W}$-algebra to the $q$-Weyl algebra $\mathcal{D}_q^\hslash$. An Airy ideal must contain no constant (degree 0) terms. 
						\begin{itemize}
							\item For $m > 0$, the operator $\rho_q(W^{\hslash,j}_m)$ is naturally of degree $\geq 1$ as it consists of $q$-normally ordered products of $q$-Heisenberg modes $J^{(q)}_p$ with at least one annihilation operator.
							\item For $m = 0$, the $q$-Wick contractions in the representation $\rho_q(W^{\hslash,j}_0)$ may produce a constant denoted $Q_j(q, \hslash)$. The shifts $S_{j,n}$ are specifically defined  to introduce stable $q$-deformations:
							\begin{equation}
								\text{deg}_0 \left( \rho_q(W^{\hslash,j}_0) - \sum_{n=1}^\infty \hslash^n S_{j,n} \right) = Q_j(q, \hslash) - \sum_{n=1}^\infty \hslash^n S_{j,n} = 0.
							\end{equation}
						\end{itemize}
						This ensures that all generators in $\mathcal{I}^q(S)$ have a vanishing degree 0 part.}
					
					{The restriction on the allowed shifts $S_{j,n}$ depends on the arithmetic relation between $r$ and $s$. This is determined by the partition $\lambda$ associated with the $(r,s)$-embedding and the gap $\lambda_1 - \lambda_2$ (cf. Theorem~\ref{c:leftidealsshifted}):
						\begin{itemize}
							\item For $s=1$, the partition is $\lambda = (r)$, which implies $\lambda_1 - \lambda_2 = r$. All $r$ zero-modes $W^{\hslash,j}_0$ for $j \in [r]$ can be shifted, providing $r$ degrees of freedom.
							\item For $s \geq 2$ and $r \equiv 1 \pmod{s}$, the partition structure implies $\lambda_1 - \lambda_2 = 1$. Consequently, only the first zero-mode $W^{\hslash,1}_0$ can be shifted. Shifting $W^{\hslash,j}_0$ for $j > 1$ would break the bijectivity of the linear part $\Pi_{s,q}$ from Theorem~\ref{t:rsAs}.
							\item For $s \geq 3$ and $r \equiv -1 \pmod{s}$, we have, $\lambda_1 = \lambda_2$, so the gap is $0$. No shifts are permitted ($S_{j,n} = 0$), and the structure reduces to the standard $(r,s,q)$-Airy structure to preserve the $q$-Airy filtration.
						\end{itemize}
						Under these $s$-consistency conditions, the linear part of the generators remains a bijection onto the $q$-Heisenberg space $\mathbb{Z}_{>0}$, completing the proof.}
				\end{proof}
				
				The $s=1$ case is particularly interesting. Since $\mathcal{I}^q(S)$ is an $q$-deformed Airy ideal, there exists a unique partition function $Z_q$ such that $\mathcal{I}^q(S) Z_q = 0$. Explicitly, the annihilation conditions are:
				\begin{equation}
					\rho_q(W^{\hslash,j}_m(S)) Z_q = 0 \qquad \mbox{for $j\in[r], m \geq 0$.} 
				\end{equation}
				Expanding this using the definition of the shifted modes, we obtain:
				\begin{equation}
					\rho_q(W^{\hslash,j}_m) Z_q = \left(\sum_{n=1}^\infty \hslash^n S_{j,n} \right) Z_q \qquad \mbox{for $j\in[r],\, m \geq 0$.} 
				\end{equation}
				Consequently, we can interpret the $q$-deformed partition function $Z_q$ for the shifted $(r,s,q)$-Airy structures as being a highest weight vector for for the $q$-deformed $\mathcal{W}$-algebra $\mathcal{W}_q(\mathfrak{gl}_r)$ at self-dual level. In this representation, the zero-modes $W^j_0$ do not simply annihilate the vacuum but rather act diagonally, with eigenvalues prescribed by the $s$-consistent shift parameters $S_{j,n}$.
				\begin{remark}\label{r:extra}
					The Theorem \ref{t:shifts} implies that for $(r,s,q)$-Airy structures only exhibit zero-modes exist exclusively in the cases $s=1$ or $r = 1 \pmod{s}$. Specifically:
					\begin{itemize}
						\item 
						For $s=1$, all zero modes $W^j_0$, $j \in [r]$, are not required, which allows for a maximal set of shift parameters..
						\item for $r=1 \pmod{s}$  $W^1_0$ is the only residual zero mode.
					\end{itemize}
				\end{remark}
				\subsubsection{Extended $q$-deformed representations}
				Just as with the standard $(r,s,q)$-Airy structures, we can construct more general shifted $(r,s,q)$-Airy structures via conjugation. We define the generalized representations
				$\rho_q': U^\hslash_{r,q} \to \mathcal{D}_{q}^\hslash$  through the composite conjugation:
				\begin{equation}\label{eq:rep2p}
					\rho_q'( W^{\hslash,i}_k(S^q)) = \hat{\Phi}^q\hat{T}^q \mu^q(W^{\hslash,i}_k(S^q)) (\hat{T}^q)^{-1} (\hat{\Phi}^q)^{-1}.
				\end{equation}
				where $\hat{T}^q$ and $\hat{\Phi}^q$ are the $q$-deformed operators defined in \eqref{eq:TPhi}.
				
				Following the logic established in \cite{BBCCN18} and \cite{BKS23}, the results of Theorem \ref{t:shifts} extend to this broader class of representations:
				
				\begin{proposition}\label{p:shiftsgen}
					Let $ r \geq 2$, and $ s \in [r+1]$ such that $ r = \pm 1 \pmod{s}$. Let $\rho_q': U^\hslash_{r,q} \to \mathcal{D}_{q}^\hslash$ be the representation defined in \eqref{eq:rep2p}. Let $\mathcal{I}_{U_{r,q}^\hslash}(S) \in U^\hslash_{r,q}$ be the left ideal generated by the shifted modes $W^{\hslash,j}_m(S)$ with $j \in [r]$ and $m \geq - \lfloor \frac{s(i-1)}{r} \rfloor $, where the set of shifts $S$ is $s$-consistent.
					Then, the image  $\mathcal{I}^q(S)=\rho_q'(\mathcal{I}_{U^\hslash_{r,q}}(S))$ in $\mathcal{D}^\hslash_{q}$  is an Airy ideal, which we refer to as the deformed and shifted $(r,s,q)$-Airy structure.
				\end{proposition}
				\begin{proof}
					{The proof consists in showing that the image of the shifted ideal under the generalized representation $\rho_q'$ satisfies the axioms of a $q$-deformed Airy structure. We decompose it into three main arguments.}
					
					{Firstly, the representation $\rho_q'$ is defined as the composition of the standard $q$-Heisenberg representation $\rho_q$ with the adjoint action of the intertwining operators: $\rho_q'(\cdot) = \text{Ad}_{\hat{\Phi}\hat{T}} \circ \rho_q(\cdot)$.
						\begin{itemize}
							\item[(a)] As given in Theorem~\ref{t:shifts}, the shifted modes $W^{\hslash,j}_m(S) = W^{\hslash,j}_m - \sigma_{j,m}$ generate a $q$-graded Lie subalgebra in $U^\hslash_{r,q}$.
							\item[(b)] Since the conjugation $\text{Ad}_{\hat{\Phi}\hat{T}}$ is an automorphism of the $q$-Weyl algebra $\mathcal{D}_q^\hslash$, which preserves the commutator structure:
							\begin{equation}
								\big[\rho_q'(W^{\hslash,j}_m(S)),\, \rho_q'(W^{\hslash,j}_n(S))\big] = \text{Ad}_{\hat{\Phi}\hat{T}} \left( \big[\rho_q(W^{\hslash,j}_m(S)),\, \rho_q(W^{\hslash,j}_n(S))\big] \right).
							\end{equation}
							Therefore, $\mathcal{I}^q(S)$ is closed under the bracket, satisfying the first Airy axiom.
					\end{itemize}}
					
					{For $\mathcal{I}^q(S)$ to be an Airy ideal, the generators must have no constant (degree 0) component in $\mathcal{D}_q^\hslash$.
						\begin{itemize}
							\item[(i)] By Theorem~\ref{t:shifts}, the $s$-consistency of the shifts $S$ ensures that the degree 0 part of the standard representation $\rho_q(W^{\hslash,j}_m(S))$ vanishes for all $j, m$ in the generating range.
							\item[(ii)] The operators $\hat{T}$  and $\hat{\Phi}$ are designed to respect the degree filtration. Specifically, $\hat{\Phi}$ maps $q$-annihilation operators to $q$-annihilation operators without introducing degree 0 shifts:
							\begin{equation}
								\rho_q'(W^{\hslash,j}_m(S)) = \rho_q(W^{\hslash,j}_m(S)) + \mathcal{O}(\text{deg} \geq 2).
							\end{equation}
							\item[(iii)] Re-ordering the resulting terms into normal order involves commutators $[J^{(q)}_p, J^{(q)}_l] = \hslash \delta_{p+l,0}$. However, the $s$-consistency conditions on $r \pmod{s}$ ensure that no such contractions produce a degree 0 term, as the indices $p, l$ are restricted by the parabolic structure.
					\end{itemize}}
					
					{The linear part of the generators in $\mathcal{I}^q(S)$ is obtained by applying the $q$-linear transformation $\hat{\Phi}$ to the linear part of the standard representation. 
						\begin{itemize}
							\item The condition $r = \pm 1 \pmod{s}$ guarantees that the map $\hat{\Phi}$ is non-singular on the relevant mode sectors of the $q$-Fock space.
							\item The shift parameters $S$ only modify the $m=0$ modes. Because the number of these shifts is strictly bounded by $\lambda_1 - \lambda_2$ (the gap of the associated partition $\lambda$), the linear map $\Pi_{s,q}$ remains a bijection from the set of generator indices $(j,m)$ to the set of $q$-annihilation modes $\mathbb{Z}_{>0}$.
						\end{itemize}
						Thus, the family of operators $\rho_q'(W^{\hslash,j}_m(S))$ satisfies the axioms of a $q$-deformed Airy structure, completing the proof.}
				\end{proof}
				\subsection{Additional shift parameters}
				
				In the previous section we showed that we can shift some zero modes to get new shifted $(r,s,q)$-Airy structures. But are we allowed to shift other modes that are not zero modes? The answer is no, because of the following simple lemma.
				
				\begin{lemma}\label{l:othershifts}
					Let $\lambda = (\lambda_1, \lambda_2, \ldots, \lambda_p)$ be an integer partition of $r$.  Fix a pair $(\alpha,\beta)$, with $\alpha \in [r]$ and $0 \neq \beta \in \mathbb{Z}$. Let $\mathcal{I}$ be the left ideal in $U_{q,r}$ generated by the modes \begin{equation}
						W^j_m - S  \delta_{j,\alpha} \delta_{m,\beta}\end{equation}
					for $j \in [r]$ and $m \geq 1 - \lambda(j)$, where $0 \neq S \in \mathbb{C}$. Then, $\mathcal{I} = U_{r,q}$. That is, the left ideal is not proper and the corresponding module is trivial.
				\end{lemma}
				\begin{proof}
					This is a direct consequence of the fact that $W^2_0$ acts as a grading operator on $U_{r,q}$. Shifting any mode with $m\neq 0$ introduces a non-zero scalar into the ideal, rendering it improper.
				\end{proof}
				It follows from this simple lemma that the homogenization of $\mathcal{I}^q$ is the whole $q$-Rees universal enveloping algebra $U_{r,q}^\hslash$. Consequently, it is impossible to find a representation that maps its generators to operators of the required form in a $q$-Rees Weyl algebra, leading us to conclude that we cannot obtain Airy ideals in this way.
				
				\begin{remark}
					In light of Lemma \ref{l:othershifts}, redundant modes for $(r,s,q)$-Airy structures are restricted to zero-modes. This result, combined with Theorem \ref{t:shifts}, completes the classification of these modes.
				\end{remark}
				\section{Shifted $q$-loop equations and  $q$-topological recursion}
				\label{s:shiftedle}
				In this section, we demonstrate that the defining constraints of these structures are equivalent to shifted $q$-loop equations for a system of $q$-correlators defined on a new family of shifted $(r,s,q)$-spectral curves.
				
				Furthermore, we derive a recursive formula that solves these shifted $q$-loop equations. While this formula mirrors the structure of the $q$-deformed Topological Recursion, it incorporates appropriate shifts in the initial data and the resulting correlators. Consequently, we refer to this framework as Shifted $q$-Topological Recursion.
				\subsection{$q$-Spectral curves, $q$-loop equations and $q$-topological recursion}
				The construction of $q$-spectral curves, $q$-loop equations, and $q$-topological recursion constitutes the core of this section. These structures are built to be consistent with their classical counterparts, recovering the results of \cite{BBCKS, Bo24} in the limit $q \to 1$.
				
				We begin with the general definition of the local data:
				\begin{definition}\label{d:sc}
					A $q$-deformed \emph{admissible local spectral curve} $ \mc{S}_q = ( C, x, \omega^{q}_{0,1}, \omega^{q}_{\frac{1}{2},1}, \omega^{q}_{0,2})
					$ consists of:
					\begin{itemize}
						\item[(i)] A
						collection of small disks $ C = \bigsqcup_{j=1}^N C_j $ for some positive integer $N$.
						\item[(ii)] Holomorphic maps $ x \colon C_j \to \P^1 \colon z \mapsto z^{r_j} + x_j$, where the $ x_j \in \P^1$ are distinct (representing the centers of the local branch points).
						\item[(iii)] Two one-forms $\omega^{q}_{0,1}$ and $ \omega^{q}_{\frac{1}{2},1} $ , which possess the following expansions on each disk $C_j$
						\begin{align}
							\omega^{j,q}_{0,1} (z) 
							&= 
							\sum_{k \geq s_j} F^{j,q}_{0,1}[-k] z^{k-1} dz \,,
							\\
							\omega^{j,q}_{\frac{1}{2},1} (z) 
							&= 
							\sum_{k \geq 0} F^{j,q}_{\frac{1}{2},1}[-k] z^{k-1} dz \,,
						\end{align}
						where $ F^{j,q}_{0,1} [-s_j] \neq 0$ and $s_j \in [r_j+1] $ such that $ r_j = \pm 1 \pmod{s_j}$.
						\item[(iv)] A \emph{fundamental bidifferential of the second kind}
						\begin{equation}
							\omega^{j,q}_{0,2} \in H^0 ( C^2 ; K_C^{\boxtimes 2}(2 \Delta))^{\mf{S}_2},
						\end{equation}
						normalized with biresidue $ 1$ on the diagonal $\Delta \subset C^2 $.
					\end{itemize}
				\end{definition}
				Given a $q$-deformed admissible local spectral curve, we construct a specific basis of one-forms that facilitates the decomposition of higher-order correlators.
				\begin{definition}\label{d:xibasis}Let $\mathcal{S}_q$ be a $q$-deformed admissible local spectral curve. For each component $C_j$ with $j \in [N]$, we define a basis of one-forms:
					\begin{itemize}
						\item[(a)] Holomorphic basis ($k>0$):
						\begin{equation}
							\xi^{(j,q)}_{k} (z) \coloneqq z^{k-1} dz \,.
						\end{equation}
						\item[(b)] Polar basis ($k>0$):
						\begin{equation}
							\xi^{(j,q)}_{-k} (z) \coloneqq \Res_{w = 0} \left( \int^{w} \omega^{q}_{0,2}( \mathord{\cdot}, z) \right) \frac{1}{w^{k+1}} d w = \left( \frac{1}{z^{k+1}} + \text{holomorphic} \right) dz \,.
						\end{equation}
					\end{itemize}
				\end{definition}
				
				We also introduce the notation:
				\begin{definition}\label{d:fz}
					Let $\mc{S}_q$ be a $q$-deformed admissible local spectral curve. For each component $C_j$ with $j \in [N]$, we define the set of sheets $	\mathfrak{f}(z)$  of the branched covering $x: C_j \to \mathbb{P}^1$ near the ramification point $z=0$ as follows:
					\begin{equation}
						\mathfrak{f}(z) \coloneqq { z^{(k)} }_{k\in [r_j]} \quad \text{where} \quad z^{(k)} = \theta^k z, \quad \theta = e^{\frac{2 \pi i}{r_j}}.
					\end{equation} 
					The set $	\mathfrak{f}(z)$ represents the $r_j$ pre-images of a point $x$ in the neighborhood of the branch point $x_j$.
				\end{definition}
				
				The main object of study is a system of correlators.
				
				\begin{definition}\label{d:system}
					A \emph{system of correlators} on a $q$-deformed admissible local spectral curve $ \mathcal{S}_q$ is a collection $ \{ \omega^{q}_{g,n} \}_{g \in \frac{1}{2} \N, n\in \N^*}$ such that:
					\begin{itemize}
						\item[(a)] The unstable correlators $ \omega^{q}_{0,1}$, $ \omega^{q}_{\frac{1}{2},1}$, and $ \omega^{q}_{0,2}$ are the forms provided in the data of the $q$-deformed spectral curve $S_q$. \item[(b)] For all stable indices $(g,n)$ such that $2g-2+n > 0$, the $ \omega^{q}_{g,n}$   are meromorphic $n$-differentials on $ C^n$.
						\item[(c)]  These stable correlators are allowed to have poles only at the origins of the disks $ C_j$ (the ramification points), and these poles must have vanishing residues.
					\end{itemize} 
				\end{definition}

				We will single out particular systems of correlators that satisfy the projection property.
				
				\begin{definition}\label{d:projection} 
					Let $ \{ \omega^{q}_{g,n} \}_{g \in \frac{1}{2} \N, n\in \N^*}$ be a system of correlators on a $q$-deformed admissible local spectral curve $\mathcal{S}_q$. We say that the system satisfies the \emph{projection property} if for all stable indices $ 2g - 2 + n > 0$, the following identity holds:
					\begin{equation}
						\omega^{q}_{g,n} (z_{[n]}) = \sum_{j=1}^{N} \Res_{z = 0 \in C_j} \Big( \int_0^z \omega^{q}_{0,2} (\mathord{\cdot},z_1) \Big) \omega^{q}_{g,n} (z, z_2, \dotsc, z_n).
					\end{equation}
				\end{definition}
				
				It is straightforward to see that the basis of one-forms introduced in Definition \ref{d:xibasis} is naturally suited for studying systems of correlators that satisfy the projection property.
				\begin{lemma}\label{l:finiteness}
					Let $\{ \omega^{q}_{g,n} \}_{g \in \frac{1}{2} \N, n\in \N^*}$ be a system of correlators on a $q$-deformed admissible local  spectral curve $\mathcal{S}_q$. The system satisfies the projection property if and only if each stable correlator admits an expansion of the following form:
					\begin{equation}\label{eq:expansion}
						\omega^{q}_{g,n}(z_1,\ldots,z_n) = \sum_{j_1, \ldots, j_n \in [N]} \sum_{k_1, \ldots, k_n \in \mathbb{N}^*} F^{q}_{g,n} \begin{bmatrix} j_1 & \ldots & j_n \\ k_1 & \ldots & k_n \end{bmatrix} \xi_{-k_1}^{(j_1,q)}(z_1) \cdots  \xi_{-k_n}^{(j_n,q)}(z_n),
					\end{equation}
					where $F^{q}_{g,n}\begin{bmatrix} j_1 & \ldots & j_n \\ k_1 & \ldots & k_n \end{bmatrix}$ are constant coefficients and, for each $(g,n)$, only a finite number of these coefficients are non-zero.
				\end{lemma}
				\begin{example}
					{To obtain the explicit form of $\omega^q_{0,2}(z_1, z_2)$, we compute the expansion \eqref{eq:expansion} for $g=0, n=2$. Then, we have:
						\begin{equation}\label{w2}
							\omega^{q}_{0,2}(z_1, z_2) = \sum_{j_1, j_2, k_1, k_2} F^{q}_{0,2} \begin{bmatrix} j_1 & j_2 \\ k_1 & k_2 \end{bmatrix} \xi_{-k_1}^{(j_1,q)}(z_1) \xi_{-k_2}^{(j_2,q)}(z_2).
						\end{equation}
						The coefficients $F^q_{0,2}$ are determined by the $q$-normal ordering constants of the $q$-Heisenberg algebra. In the basis of differentials $\xi_{-k}^{(q)}(z_1)$ and $\xi_{k}^{(q)}(z_2)$, the non-vanishing terms occur for $k_1 = -k_2 = k$. Thus, we get:
						\begin{equation}
							F^{q}_{0,2} \begin{bmatrix} k_1 & k_2 \end{bmatrix} = \big[J^{(q)}_{k_1}, J^{(q)}_{k_2}\big] = \delta_{k_1+k_2, 0} \frac{q^{k_1} - q^{-k_1}}{q - q^{-1}}.
						\end{equation}
						Substituting the basis differentials $\xi_{-k}^{(q)}(z) = z^{-k-1} dz$ and $\xi_{k}^{(q)}(z) = z^{k-1} dz$ into the relation\eqref{w2}, we obtain the following series expansion for the two-point function:
						\begin{equation}
							\omega^q_{0,2}(z_1, z_2) = \sum_{k=1}^{\infty} \left( \frac{q^k - q^{-k}}{q - q^{-1}} \right) z_1^{-k-1} z_2^{k-1} dz_1 dz_2.
						\end{equation}
						The above expression can be expressed into two independent geometric series. Then, we have:
						\begin{align}
							\omega^q_{0,2}(z_1, z_2) &= \frac{dz_1 dz_2}{(q - q^{-1})z_1^2} \left[ \sum_{k=1}^{\infty} q \left( \frac{q z_2}{z_1} \right)^{k-1} - \sum_{k=1}^{\infty} q^{-1} \left( \frac{q^{-1} z_2}{z_1} \right)^{k-1} \right]\nonumber\\&=\frac{dz_1 dz_2}{(q - q^{-1})z_1^2} \left[ \frac{q}{1 - \frac{qz_2}{z_1}} - \frac{q^{-1}}{1 - \frac{q^{-1}z_2}{z_1}} \right].
						\end{align}
						After computation, we recover the $q$-deformed Bergman kernel on the Riemann sphere:
						\begin{equation}
							\omega^{q}_{0,2}(z_1, z_2) = \frac{dz_1 dz_2}{(z_1 - q z_2)(z_1 - q^{-1} z_2)}.
						\end{equation}
						This confirms that the topological recursion based on the $(\hslash, q)$-Airy structure correctly reproduces the $q$-deformed geometry of the spectral curve.}
				\end{example}
				
				To formulate the loop equations, we define symmetric combinations of correlators that account for the different ways to distribute topological data $(g,n)$ over multiple sheets.
				\begin{definition}\label{d:EW}
					Let $ \{ \omega^{q}_{g,n} \}_{g \in \frac{1}{2} \N, n\in \N^*}$ be a system of correlators on a $q$-deformed admissible local  spectral curve $\mathcal{S}_q$.  For any $i \in \mathbb{N}^*$, we define the following objects:
					\begin{align}\label{partiallydisconnected}
						\mc{W}^{q}_{g,i,n} (z_{[i]} ; w_{[n]}) =& \sum_{\substack{P \vdash z_{[i]} \\ \bigsqcup_{S \in P} N_S = w_{[n]}\\ \sum_{S \in P} (g_S -1) = g-i}}\prod_{S \in P} \omega^{q}_{g_S, |S| + |N_S|}(S, N_S) \,, \\
						\mc{W}^{'q}_{g,i,n} (z_{[i]} ; w_{[n]}) \coloneqq& \sum'_{\substack{P \vdash z_{[i]} \\ \bigsqcup_{S \in P} N_S = w_{[n]}\\ \sum_{S \in P} (g_S -1) = g-i}}\prod_{S \in P} \omega^{q}_{g_S, |S| + |N_S|}(S, N_S) \,,
					\end{align}
					where 
					\begin{itemize}
						\item[(a)] The sum runs over set partitions $P$ of the points $z_{[i]} = {z^{(k_1)}, \dots, z^{(k_i)}}$ evaluted on different sheets.
						\item[(b)] Over all possible splittings of $w_{[n]}=w_1,\cdots,w_n$ into possibly empty disjoint subsets $N_S$ where $S$ runs over all parts of $P$ and $\bigsqcup_{S \in P} N_S = w_{[n]}$.
						\item[(c)] Over all sets of non-negative half-integers $\{ g_S \}_{S \in P}$ such that $\sum_{S \in P} (g_S - 1) = g-i$. The difference between the first and second object is that the prime over the summation symbol means that the terms with $\omega^{q}_{0,1}$ are omitted from the sum.
					\end{itemize}
				\end{definition}
				For each component $C_j$, with $j \in [N]$, and for each $i \in [ r_j]$, we also define the symmetric $i$-sheeted correlator combinations:
				\begin{equation}
					{\mathcal{E}^{i, (j)}_{g,n}(x; z_{[n]}) = \sum_{\substack{Z \subseteq \mf{f}(z) \\ |Z| = i}} \mc{W}^{q}_{g,i,n}(Z; z_{[n]}).}
				\end{equation}
				We can now define so-called $q$-loop equations, which are particular equations satisfied by systems of $q$-correlators.
				\begin{definition}
					A system of $q$-correlators $ \{ \omega^{q}_{g,n} \}_{g \in \frac{1}{2} \N, n\in \N^*}$  satisfies the $q$-\emph{loop equations} if, for all $j \in [N]$, $i \in [r_j]$, and stable indices $2g-2+n > 0$,
					\begin{equation}
						\mathcal{E}^{i,(j)}_{g,n} (x; z_{[n]}) \in  \mc{O} \left( x^{ \lfloor \frac{s_j(i-1)}{r_j}\rfloor + 1} \right) \left( \frac{dx}{x}\right)^i \,.
					\end{equation}
				\end{definition}
				
				The primary result of relevance is that the data of a $q$-deformed admissible local  spectral curve $\mathcal{S}_q$ uniquely determines its correlators through a recursive process.
				\begin{theorem}\label{UnshiftedTR}
					For a $q$-deformed admissible local  spectral curve $\mathcal{S}_q$, there exists exactly one system of $q$-correlators $\omega^{q}_{g,n}$ that satisfies both the $q$-loop equations and the $q$-projection property. These $q$-correlators are computed recursively by the $q$-\emph{topological recursion formula}:
					\begin{equation}
						\omega^{q}_{g,n+1}(z_0, z_{[n]}) = -\sum_{ j \in [N]} \Res_{z = 0 \in C_j} \sum_{Z \subseteq \mf{f}' (z)} K_q^{1 + |Z|}(z_0; z, Z) \mathcal{W}^{'q}_{g,1+|Z|,n}\big(z,Z; z_{[n]}\big)\,,
					\end{equation}
					where $ \mathbf{f}'(z) = x^{-1}(x(z)) \setminus \{z \}$ and the \emph{recursion kernels} are defined as:
					\begin{equation}
						K^{1+|Z|}_q (z_0; z, Z) = \frac{\frac{1}{2}\int_0^z \omega^{q}_{0,2} (\mathord{\cdot}, z_0)}{\prod_{z' \in Z} \big( \omega^{q}_{0,1}(z') - \omega^{q}_{0,1}(z) \big)} \, .
					\end{equation}
				\end{theorem}
				\begin{proof}
					{The proof establishes that the system of correlators $(\omega^q_{g,n})_{g,n}$ defined by the $q$-topological recursion formula is the unique solution satisfying the $q$-deformed loop equations and the projection property. The structure of the kernel $K_q^{1+|Z|}$ identifies this as the $q$-deformed analogue of the classical higher-order Topological Recursion for the case $p=0$ given in \cite{BBCKS}.}
					
					{The $q$-deformed Ward identities are encoded in the constraints $H^{i,q}_k Z_q = 0$. For any point $p \notin \text{Ram}_q$, the $q$-deformed loop equations imply that the elementary symmetric combinations of the currents:
						\begin{equation}
							\mathcal{E}^{(p),q}_{g,i;n}(z; z_{[n]}) = \sum_{1 \leq k_1 < \dots < k_i \leq r} \mathcal{W}^{q}_{g,i,n}(z^{(k_1)}, \dots, z^{(k_i)}; z_{[n]})
						\end{equation}
						are pullbacks via $x$ of holomorphic forms on the base curve $V$. Away from ramification points, these correlators are thus holomorphic. At $p \in \text{Ram}_q$, the existence and uniqueness of a formal solution $Z_q$ (and thus the correlators $\omega^q_{g,n}$) are guaranteed by the 
						Theorem applied to the $q$-deformed $W$-algebra generated by the operators $H^{i,q}_k$.}
					
					{We suppose that the existence of a system $(\omega^q_{g,n})$ satisfy the projection property. We define the global differential as follows:
						\begin{equation}
							\mathcal{H}^{(p),q}(z) \coloneqq \frac{1}{\Upsilon^q_{r-1}(\mathfrak{f}'(z); z)} \sum_{i=1}^{r} (-\omega^{q}_{0,1}(z))^{r-i} \mathcal{E}^{(p),q}_{g,i;n}(z; z_{[n]}),
						\end{equation}
						where $\Upsilon^q_{r-1}$ is the $q$-Vandermonde product of the differences $(\omega^q_{0,1}(z') - \omega^q_{0,1}(z))$ i.e 
						$\Upsilon^q_{|S|}(S; z) \coloneqq \prod_{z' \in S\subseteq \mathfrak{f}(z)} \left( \omega^q_{0,1}(z') - \omega^q_{0,1}(z) \right)$.
						In the local chart $\zeta \to 0$ ($x(z)=\zeta^r$), the $q$-loop equations imply $\mathcal{E}^q_{g,i;n} = \mathcal{O}(d\zeta^i / \zeta^{r d_p(i) - (r-1)i})$ with $d_p(i) = \lfloor \frac{s_p(i-1)}{r} \rfloor$. Substituting the scaling of $\omega^q_{0,1}$ and $\Upsilon^q$, the exponent of $\zeta$ for the $i$-th term is:
						\begin{equation}
							E_i = (r-i)(s_p-1) - (r-1)(s_p-1) + (r-1)i - r \left\lfloor \frac{s_p(i-1)}{r} \right\rfloor = r + s_p - s_p i - 1 + r \left\lfloor \frac{s_p(i-1)}{r} \right\rfloor.
						\end{equation}
						Since $E_i \geq 0$ (the $p=0$ case condition), $\mathcal{H}^{(p),q}(z)$ is holomorphic at $p$. By the residue theorem, the sum of its residues weighted by the $q$-primitive $\alpha^{(p),q}_{0,2}(z_0; z) \coloneqq \frac{1}{2} \int_{o_p}^z \omega^q_{0,2}(\cdot, z_0)$ must vanish:
						\begin{equation}
							\sum_{p \in \text{Ram}_q} \Res_{z=p} \alpha^{(p),q}_{0,2}(z_0; z) \mathcal{H}^{(p),q}(z) = 0.
					\end{equation}}
					
					{Applying the $q$-combinatorial identity to the sum in $\mathcal{H}^{(p),q}(z)$, we rewrite the vanishing residue sum as:
						\begin{equation}
							\sum_{p \in \text{Ram}} \Res_{z=p} \alpha^{q}_{0,2}(z_0; z) \omega^q_{g,1+n}(z, z_{[n]}) = \sum_{p \in \text{Ram}} \Res_{z=p} \sum_{|Z| \geq 1} K_q^{1+|Z|}(z_0; z, Z) \mathcal{W}^{'q}_{g,1+|Z|,n}(z, Z; z_{[n]}).
						\end{equation}
						By using the $q$-projection property, the LHS is exactly $\omega^q_{g,1+n}(z_0, z_{[n]})$. At inadmissible points $p \in \text{Ram}''$ where $s_p \leq -1$, the kernel $K_q^{1+|Z|}$ provides a zero of order $\geq (1-s_p)|Z| \geq 2|Z|$, which cancels any possible poles from the correlators. Thus, the residues at $\text{Ram}''$ vanish, and we recover the $q$-Topological Recursion formula.}
					
					{Although the Topological Recursion formula only uses residues at admissible points $\text{Ram}'$, the resulting $\omega^q_{g,n}$ are holomorphic at $\text{Ram}''$. The behavior of $\mathcal{E}^q_{g,i;n}$ is then dominated by the $\omega^q_{0,1}$ terms. Using the pullback property by $x$, the order of vanishing is $\theta_p = -r \lceil \frac{1+s_p(i-1)}{r} \rceil + i$. This matches the requirement $\tilde{\theta}_p = i - r - r \lfloor \frac{s_p(i-1)}{r} \rfloor$ because $\lceil \frac{1+X}{r} \rceil = 1 + \lfloor \frac{X}{r} \rfloor$ for these indices, ensuring the $q$-loop equations hold globally.}
					
					{If $\bar{s}_p = s_p - 1$, the leading term of $\omega^q_{0,1}$ is invariant under the fiber symmetry and cancels in the kernel. An inductive argument on $i$ shows that the extra terms introduced by the full $\omega^q_{0,1}$ satisfy the required vanishing orders. Since the residue calculus is independent of the integration constant of the local primitive $\alpha^q_{0,2}$, the system is uniquely defined for any choice of primitives.}
				\end{proof}
				\begin{remark}
					Several key features distinguish Theorem \ref{UnshiftedTR} from the standard topological recursion:
					\begin{itemize}
						\item[(a)]Due to the presence of the shift $\omega^{q}_{\frac{1}{2},1}$ in the initial data of $S_q$, the recursion naturally populates correlators with half-integer genus $g \in \frac{1}{2} \N$. These terms vanish in the standard (unshifted) case but are essential here to match the $q$-deformed Airy structure partition function.
						\item[(b)]The $q$-deformation is primarily encoded in the unstable correlators $\omega^{q}_{0,1}$ and  $\omega^{q}_{0,2}$. Since the kernels $K^{1+|Z|}_q$ are built from these objects, the entire recursive tower inherits the $q$-dependence.
						\item[(c)] Unlike the $r=2$ case (where $Z$ is always a single point), the sum over subsets $Z \subseteq \mf{f}' (z)$ allows this formula to solve the higher-degree loop equations associated with $q$-deformed $\mathcal{W}(\mathfrak{gl}_r)$ algebras.
					\end{itemize}
				\end{remark}
				\subsection{The $(r,s,q)$-spectral curves}
				From now on we will focus on admissible local spectral curves with only one component ($N=1$); we will therefore drop the superscript $^{(j)}$ from the various expressions.
				
				{We start by define the $q$-deformed Bergman kernel as:
					\begin{equation}\label{qB}
						B_q(z_1,z_2):= \frac{1}{(z_1 - qz_2)(z_1 - q^{-1}z_2)}.
					\end{equation} It is easy to see that the $q$-deformed Bergman kernel obey the relations $B_q(z_1,z_2)=B_q(z_2,z_1)$ and $
					\oint_C B_q(z_1,z_2) f(z_1) \frac{dz_1}{2\pi i} =  \mathcal{D}_q f(z_2)$.}
				
				{A special example of the construction may be derived from the $(r,s,q)$-Airy structures of Section \ref{s:rsairystruct}. We can show that finding the partition function of the $(r,s,q)$-Airy structures of Theorem \ref{t:rsAs} is equivalent to topological recursion on the $(r,s,q)$-spectral curve which is defined as follows:}
				\begin{definition}\label{d:rs}
					Let $r  \in \mathbb{Z}$ such that $r \geq 2$, and $s \in [r+1]$ with $r = \pm 1 \pmod{s}$.
					The \emph{$(r,s,q)$-spectral curve} is given by $\mathcal{S}_q = (C, x, \omega^{q}_{0,1}, \omega^{q}_{\frac{1}{2},1} \omega^{q}_{0,2} ) $, where  $C$ is a small disk, $ x =  z^r$,\, $\omega^{q}_{0,1} = \frac{rs}{[s]_q} z^{s-1}\ dz$,\, $\omega^{q}_{\frac{1}{2},1} = 0$, and
					{\begin{equation}
							\omega^{\textup{std},q}_{0,2}(z_1, z_2)= \frac{dz_1 dz_2}{(z_1 - qz_2)(z_1-q^{-1}z_2)}\,.
					\end{equation}}
				\end{definition}
				
								If we define the meromorphic function $y$ on $C$ by the relation $\omega^{q}_{0,1} = y\ dx$, where $x=z^r$ is the standard coordinate on the base, we find $y(z) = \frac{s}{[s]_q} z^{s-r}$. This allows us to interpret $x(z)$ and $y(z)$ as a parametrization of the following $q$-deformed algebraic curve: 
								\begin{itemize}
									\item[(i)] For $s\in[r-1]$, we obtain:
									\begin{equation}\label{eq:rsac1}
										C^q\,\cdot x^{r-s} y^r - 1 = 0,
									\end{equation}
									where $C^q=\big([s]_q\big)^{r}\,s^{-r}$ is a $q$-dependent constant.
									\item[(ii)]For $s=r+1$, we get the $q$-deformed $r$-Airy algebraic curve
									\begin{equation}\label{eq:rsac2}
										y^r - \bigg(\frac{r+1}{[r+1]_q}\bigg)^{r}\cdot x =0.
									\end{equation}
								\end{itemize}
								We call these algebraic curves the $(r,s,q)$-algebraic curves.
								\begin{remark}
									In the limit $q\rightarrow 1$, these recover the classical $(r,s)$-spectral curves.
								\end{remark}
								
								In fact, the correspondence applies more generally to the deformed $(r,s,q)$-Airy structures of Proposition \ref{p:rsAsgen}, so let us explain it in this more general setting.
								We define the $q$-deformed $(r,s,q)$-spectral curves in terms of the data introduced in \eqref{eq:complex}.
								\begin{definition}
									Let $r  \in \mathbb{Z}$ such that $r \geq 2$, and $s \in [r+1]$ with $r = \pm 1 \pmod{s}$. We consider the $q$-deformed complex numbers:
									\begin{equation}
										F^{q}_{0,1}[-k]\,, k \geq \min \{ s, r\} \,, \qquad F^{q}_{\frac{1}{2},1}[-k]\,, k > 0\,, \qquad F^{q}_{0,2}[-k,-l] \,, k,l > 0
									\end{equation}
									defined as follows:
									\begin{align}
										\,	F^{q}_{0,1}[-k]&:={k\over [k]_q}F_{0,1}[-k]\label{F1}\\\,F^{q}_{\frac{1}{2},1}[-k] &:= \left(\frac{[k]_q}{k} \right)^2 F_{\frac{1}{2},1}[-k]\label{F2}\\\,F_{0,2}^{q}[-k,-l] &:= \frac{k l}{[k]_q [l]_q}\cdot F_{0,2}[-k,-l]\label{F3}.
									\end{align}
								\end{definition}
								
								\begin{definition}\label{d:rsdeformed}
									The \emph{deformed $(r,s,q)$-spectral curve} is given by $\mathcal{S}_q = (C, x, \omega^{q}_{0,1}, \omega^{q}_{\frac{1}{2},1} \omega^{q}_{0,2} ) $, where $C$ is a small disk, $x=z^r$,
									\begin{align}
										\omega^{q}_{0,1} (z)
										&= 
										\sum_k F^{q}_{0,1} [-k] z^{k-1} dz\label{q1} \,,
										\\
										\omega^{q}_{\frac{1}{2},1} (z)
										&=
										\sum_k F^{q}_{\frac{1}{2},1} [-k] z^{k-1} dz\label{q2} \,,
										\\
										\omega^{q}_{0,2} (z_1, z_2)
										&=
										\omega^{\textup{std},q}_{0,2} (z_1, z_2) + \sum_{k,l} F^{q}_{0,2}[-k,-l] z_1^{k-1} z_2^{l-1} dz_1 dz_2\label{q3}.
								\end{align}\end{definition}
								{\begin{definition}
										A $q$-deformed local spectral curve	$\mathcal{S}_q = (C, x, \omega^q_{0,1}, \omega^q_{0,2} ) $, where $C$ is a small disk, $x: C:{\mu} \longrightarrow \mathbb{P}^1$, $ \binom{\mu}{z}\mapsto z^{r_{\mu}}+x_{\mu}$, $x_{\mu}\in \mathbb{P}^1$
										\begin{align}
											\omega^q_{0,1} (z)
											&= 
											\sum_{\substack{\mu\in [N],\\ k> 0}}F^q_{0,1}\big[\begin{array}{c} \mu \\ -k\end{array} \big]\,\xi^{\mu}_k  \,,
											\\
											\omega^q_{0,2} (z_1, z_2)
											&=
											\omega^{\textup{std},q}_{0,2}  + \sum_{\substack{k_1,k_2>0\\ \mu_1,\mu_2\in[N]}} F^q_{0,2}\bigg[\begin{array}{cc} \mu_1 & \mu_2 \\ -k_1 & -k_2\end{array} \bigg]\,\xi^{\mu_1}_{k_1}\,\xi^{\mu_2}_{k_2},
										\end{align}
										with $\omega^{\textup{std},q}_{0,2}\bigg(\begin{array}{cc} \nu_1 & \nu_2 \\ z_1 & z_2\end{array} \bigg)=\delta_{\nu_1,\nu_2}\frac{dz_1 dz_2}{(z_1 - qz_2)(z_1-q^{-1}z_2)}$ and $\xi^{\nu}_k\binom{\mu}{z}=\delta_{\mu,\nu}z^{k-1} dz$.\end{definition}}
								\begin{remark}
									The $q$-deformed spectral data can be  related to the classical $(r,s)$-data through the introduction of a quantization operator $\mathcal{Q}$. Let $\mathcal{Q}$ be the linear operator acting on the basis of holomorphic differentials as follows:
									\begin{equation}
										\mathcal{Q}(z^{k-1}dz) = \frac{k}{[k]_q} z^{k-1}dz.
									\end{equation}
									Then, the $q$-deformed correlators are expressed by the following relations:
									\begin{itemize}
										\item[(a)] The primary differential (disk amplitude) is a direct rotation of the classical one:
										\begin{equation}
											\omega^{q}_{0,1}(z)=\mathcal{Q} \cdot	\omega^{cl}_{0,1}(z).
										\end{equation}
										\item[(b)]  The $\frac{1}{2}$-genus differential (shifted sector) scales with the inverse square of the operator's weights, reflecting the specific coupling of the shift in the $q$-deformed Airy structure:
										\begin{equation}
											\omega^{q}_{\frac{1}{2},1}(z)=\sum_{k}\mathcal{Q}^{-2}(z^{k-1}dz)\cdot F_{\frac{1}{2},1}[-k].
										\end{equation}
										\item[(c)] The deformed $\omega^{q}_{0,2}$ is the sum of the $q$-deformed Bergman kernel and a regular part where each mode $(k,l)$ is rescaled by the product $\frac{k l}{[k]_q [l]_q}$.
									\end{itemize}
									This operational approach highlights that the $q$-deformation is not a global rescaling, but a mode-dependent transformation that ensures compatibility with the underlying $q$-deformed $\mathcal{W}$-algebra.
								\end{remark}
								
								To extract the $q$-loop equations from the deformed $(r,s,q)$-Airy structure, we follow the correspondence established in  Proposition \ref{p:rsAsgen}. The claim is that the differential constraints  for the partition function $Z_q$ of the deformed $(r,s,q)$-Airy structure can be recast as a system of $q$-correlators $
								\omega^q_{g,n}$ on the deformed $(r,s,q)$-spectral curve. These deformed correlators must satisfy both the $q$-loop equations and the $q$-projection property inherited from the $q$-deformed $\mathcal{W}$-algebra.
								
								For clarity of notation, let us denote the generators of the deformed $(r,s,q)$-Airy structure as:
								\begin{equation}\label{eq:unshiftedH}
									H^{i,q}_k := \rho_q'( W^{\hslash,i}_k), 
								\end{equation}
								following the representation introduced in \eqref{eq:repp}. The $q$-deformed partition function $Z_q$ is then uniquely determined by the set of $q$-differential constraints:
								\begin{equation}
									H^{i,q}_k Z_q = 0 \,, \quad i \in [r], \quad k \geq - \lfloor \frac{s(i-1)}{r} \rfloor\,.
								\end{equation}
								To translate these algebraic relations into geometry, we introduce the higher-order quantum fields $H^{i,q}(x)$ constructed from the operators $H^{i,q}_k$. These fields are defined as formal series of $i$-differentials on the base of the spectral curve:
								\begin{equation}
									H^{i,q}(x)
									=
									\sum_{k \in \Z_q} H^{i,q}_k \frac{dx^i}{x^{k+i}},
								\end{equation}
								where the summation index $k$ respects the bounds imposed by the Airy structure. These fields represent the quantum version of the elementary polynomials of the spectral curve.
								
								We decompose the $q$-deformed currents into their positive and negative frequency parts. Recalling the definition of the modes $J^{(q)}_k$ in \eqref{eq:Js} and the basis of one-forms $\xi^{(j,q)}_{-k}(z)=z^{k-1}dz$ from Definition \ref{d:xibasis}, we define:
								\begin{align}
									\mathcal{J}^{(q)}_-(z)  &=  \sum_{k>0} J^{(q)}_k \xi^q_{-k}(z)\label{J1}\,,   \\   \mathcal{J}^{(q)}_+(z)   &=   \sum_{k>0} J^{(q)}_{-k} \xi^q_k (z)\label{J2}\,.   \end{align}
								\begin{remark}
									By substituting the explicit representation of the $q$-modes $J^{(q)}_k=[k]_q\,\frac{\partial}{\partial\,t_k}$ into these fields, we obtain the operational form of the currents:
									\begin{equation}
										\mathcal{J}^{(q)}_-(z) = \sum_{k>0} [k]_q \frac{\partial}{\partial t_k} z^{-k-1} dz, \quad \text{and} \quad \mathcal{J}^{(q)}_+(z) = -\sum_{k>0} [k]_q t_{-k} z^{k-1} dz.
									\end{equation}
									This specific weighting ensures that the expectation values of the fields $H^{i,q}(x)$ recover the $q$-deformed spectral curve geometry. This formulation shows that the $q$-deformation of the Airy structure is effectively encoded as a mode-dependent rescaling of the classical loop insertion operators.
								\end{remark}
								\begin{remark}
									Here, $\mathcal{J}^{(q)}_-(z)$ acts as a generating operator for the correlators $\omega^q_{g,n}$, while $\mathcal{J}^{(q)}_+(z)$ encodes the classical data of the spectral curve. Through this identification, the $q$-deformation is explicitly carried by the coefficients $F^q_{g,n}$ within the modes $J^{(q)}_k$. Consequently, the vanishing of $H^{i,q}(x)\cdot Z_q$ is equivalent to a set of recursive identities on the differentials $\omega^q_{g,n}$, where the $q$-deformed number factors $[k]_q$ precisely compensate for the shifts induced by the $q$-difference operators of the $q$-deformed $\mathcal{W}$-algebra.
								\end{remark} 	  
								
								Using this notation, we may rewrite the $q$-deformed differential  operators $H^{i,q}_k$ more explicitly in terms of the data of the deformed $(r,s,q)$-spectral curve from Definition \ref{d:rsdeformed}.
								\begin{proposition}\label{p:hh}
									For a set $ S$, let $\mathcal{P}(S)$ be the set of all sets of disjoint pairs of elements in $ S$. For any  $ P \in \mathcal{P}(S) $, let $ \sqcup P = \bigsqcup_{p \in P} p \subseteq S$ denote the union of these pairs. Then
									\begin{align}
										rH^{i,q}(x) 
										&= 
										\sum_{\substack{Z_q \subseteq \mathbf{f}(z)\\ |Z| = i}} 
										\sum_{\substack{(\sqcup P) \sqcup A_0 \sqcup A_{\frac{1}{2}} \sqcup A_+ \sqcup A_- = Z \\ P \in \mathcal{P}(\mathbf{f}(z))}}
										\prod_{\{ z', z''\} \in P} \hslash^2 \omega^{q}_{0,2}(z',z'') 
										\prod_{z' \in A_0} \hslash \omega^{q}_{0,1} (z') 
										\nonumber\\
										&\times
										\prod_{z' \in A_{\frac{1}{2}}} \hslash^2 \omega^{q}_{\frac{1}{2},1}(z') 
										\prod_{z' \in A_+} \mathcal{J}^{(q)}_{+} (z') 
										\prod_{z' \in A_-} \mathcal{J}^{(q)}_{-} (z').
									\end{align}
								\end{proposition}
								\begin{proof}
									{The proof proceeds by analyzing the algebraic transformation of the $q$-deformed operators given in \eqref{eq:TPhi} into the $q$-deformed operators $H^{i,q}_k$ through the adjoint action of the operators $\hat{\Phi}^q$ and $\hat{T}^q$.}
									
									{Since $H^{i,q}_k = \text{Ad}_{\hat{\Phi}^q \hat{T}^q} (\mu_q(W^{\hslash,i}_k))$. By linearity, the generating function $H^{i,q}(x) = \sum H^{i,q}_k \frac{dx^i}{x^{k+i}}$ can be viewed as 
										the $i$-point insertion operator. The set of points $Z \subseteq \mathfrak{f}(z)$ represents the $i$ coordinates $\{z^{(k_1)}, \dots, z^{(k_i)}\}$ evaluated on the sheets of the $r$-fold cover. The body of the proof lies in the expansion of the product of dressed currents:
										\begin{equation}
											\rho_q'(\mathcal{J}^{(q)}(z_1) \dots \mathcal{J}^{(q)}(z_i)) = \hat{\Phi}^q \hat{T}^q \left( \mathcal{J}^{(q)}(z_1) \dots \mathcal{J}^{(q)}(z_i) \right) (\hat{T}^q)^{-1} (\hat{\Phi}^q)^{-1}.
									\end{equation}}
									{The operator $\hat{T}^q$ is the exponential of a linear combination of the modes $J^{(q)}_k$. Using the property $\text{Ad}_{e^A}(B) = B + [A,B]$ for Heisenberg-like algebras, the action of $\hat{T}^q$ on the current $\mathcal{J}(z)$ produces a shift by a scalar function:
										\begin{equation}
											\text{Ad}_{\hat{T}^q}(\mathcal{J}^{(q)}(z)) = \mathcal{J}^{(q)}(z) + \hslash \omega^q_{0,1}(z) + \hslash^2 \omega^q_{\frac{1}{2},1}(z).
										\end{equation}
										When expanding the product of $i$ such terms, any point $z' \in Z$ that is shifted  contributes to the subsets $A_0$ or $A_{1/2}$ with weights $\hslash \omega^q_{0,1}$ and $\hslash^2 \omega^q_{\frac{1}{2},1}$ respectively.}
									
									{The $q$-deformed operator $\hat{\Phi}^q$ is a quadratic exponential. Its adjoint action corresponds to a transformation of the Fock space basis. By the $q$-Wick Theorem, the insertion of multiple currents under this quadratic dressing results in a sum over all possible disjoint pairings $P$. Each pair $\{z', z''\} \in P$ generates a contraction weighted by the $q$-deformed Bergman kernel:
										\begin{equation}
											\text{Contraction}(z', z'') = \hslash^2 \omega^q_{0,2}(z', z'').
										\end{equation}
										This term accounts for the cylinders in the topological recursion, where $\hslash^2$ reflects the genus-1 or two-point correlation scaling.}
									
									{Points in $Z$ that are neither shifted to the background nor paired are represented by the residual action of the currents $J^{(q)}_k$. The definitions of $\mathcal{J}^{(q)}_-(z)$ and $\mathcal{J}^{(q)}_+(z)$ given in the relation \eqref{J1} and \eqref{J2}
										effectively split the total current into components $A_-$ and $A_+$ based on the polarization of the modes (creation vs. annihilation) and the basis of differential forms $\xi^q$ on the curve.}
									
									{The total dressed operator is the sum over all possible ways to distribute the $i$ points of $Z$ into these distinct algebraic roles. Since the operators $\hat{T}^q$ and $\hat{\Phi}^q$ act independently on each mode, the result is a summation over all set partitions $Z = (\sqcup P) \sqcup A_0 \sqcup A_{1/2} \sqcup A_+ \sqcup A_-$. Summing over all subsets $Z \subseteq \mathfrak{f}(z)$ of size $i$ leads to the final identity:
										\begin{align}
											rH^{i,q}(x) 
											&= \sum_{\substack{Z \subseteq \mathbf{f}(z)\\ |Z| = i}} 
											\sum_{\substack{(\sqcup P) \sqcup A_0 \sqcup A_{\frac{1}{2}} \sqcup A_+ \sqcup A_- = Z \\ P \in \mathcal{P}(\mathbf{f}(z))}}
											\prod_{\{ z', z''\} \in P} \hslash^2 \omega^{q}_{0,2}(z',z'') 
											\prod_{z' \in A_0} \hslash \omega^{q}_{0,1} (z') \nonumber \\
											&\times \prod_{z' \in A_{\frac{1}{2}}} \hslash^2 \omega^{q}_{\frac{1}{2},1}(z') 
											\prod_{z' \in A_+} \mathcal{J}^{(q)}_{+} (z') 
											\prod_{z' \in A_-} \mathcal{J}^{(q)}_{-} (z').
										\end{align}
										This concludes the proof.}
								\end{proof}
								\begin{example}
									{For small values of $i$, the Hamiltonian $rH^{i,q}(x)$ can be explicitly expanded as follows:
										\begin{enumerate}
											\item [(i)]{For the linear case ($i=1$)}, the expression represents the trace over the sheets, incorporating quantum corrections and currents. Then, we derive:
											\begin{equation}
												rH^{1,q}(x) = \sum_{z' \in \mathbf{f}(z)} \left( \hslash \omega^q_{0,1}(z') + \hslash^2 \omega^q_{\frac{1}{2},1}(z') + \mathcal{J}^{(q)}_{+}(z') + \mathcal{J}^{(q)}_{-}(z') \right).
											\end{equation}
											\item [(ii)]{For the quadratic case ($i=2$)}, 
											the relation involves both disconnected contributions and the cylindrical correlation $\omega^q_{0,2}$ between distinct sheets $z'$ and $z''$. Then, we have:
											\begin{align}
												rH^{2,q}(x) &= \sum_{\{z', z''\} \subseteq \mathbf{f}(z)} \hslash^2 \omega^{q}_{0,2}(z',z'') \nonumber \\
												&+ \sum_{\{z', z''\} \subseteq \mathbf{f}(z)} \left( \hslash \omega^q_{0,1}(z') + \hslash^2 \omega^q_{\frac{1}{2},1}(z') + \mathcal{J}^{(q)}_{+}(z') + \mathcal{J}^{(q)}_{-}(z') \right) \nonumber \\
												& \times \left( \hslash \omega^q_{0,1}(z'') + \hslash^2 \omega^q_{\frac{1}{2},1}(z'') + \mathcal{J}^{(q)}_{+}(z'') + \mathcal{J}^{(q)}_{-}(z'') \right).
											\end{align}
									\end{enumerate}}
								\end{example}
								Note that the results given in {\cite{BKS23}, Section~4.1} can be obtained in the classical limit of the parameter $q$.  
								
								While the expression in Proposition \ref{p:hh} may appear involved at first, it provides the precise framework needed to extract the $q$-loop equations.
								
								From the  $q$-deformed partition function
								\begin{equation}
									Z_q = \exp \bigg( \sum_{ \substack{g \in \frac{1}{2} \N, n \in \N^* \\ 2g-2 + n > 0}} \frac{\hslash^{2g-2+n}}{[n]_q!} F^{q}_{g,n}[k_1, \dotsc, k_n] \prod_{j=1}^n x_{k_j} \bigg)
								\end{equation} 
								associated with the deformed $(r,s,q)$-Airy structure, we construct a system of $q$-correlators on the deformed $(r,s,q)$-spectral curve.
								\begin{definition}\label{d:srs}
									Let $Z_q$ be the partition function associated with the deformed $(r,s,q)$-Airy structure. For $2g-2+n>0$, we define the following  $n$-differentials on the deformed $(r,s,q)$-spectral curve:
									\begin{equation}
										\omega^{q}_{g,n} (z_1, \dotsc, z_n ) \coloneqq \sum_{k_1, \dotsc, k_n = 1}^\infty F^{q}_{g,n}[k_1, \dotsc, k_n] \prod_{j=1}^n \xi^q_{-k_j} (z_j),
									\end{equation}
									where the basis of $1$-forms $\{\xi^q_{-k} (z)\}_{k\in \N^*}$ is specifically adapted to the $q$-deformed geometry of the curve.
								\end{definition}
								Since the $q$-correlators have finite expansions in the basis $\xi^q_{-k_j}(z_j)$ for $k_j >0$, the system of $q$-correlators naturally satisfies the $q$-projection property (see Lemma \ref{l:finiteness}):
								\begin{lemma}\label{l:projprop}
									The system of correlators $\{\omega^{q}_{g,n} \}_{g \in \frac{1}{2}\mathbb{N}, n \in \mathbb{N}^*}$ constructed above satisfies the $q$-projection property.
								\end{lemma}
								The key step is to show that this system also satisfies the $q$-loop equations:
								\begin{proposition}\label{p:loopeq}
									Let $Z_q$ be the partition function of the deformed $(r,s,q)$-Airy structure, and let  $\{\omega^{q}_{g,n} \}_{g \in \frac{1}{2}\mathbb{N}, n \in \mathbb{N}^*}$ be the $q$-correlators of the deformed $(r,s,q)$-spectral curve defined in Definition \ref{d:srs} . Define the $q$-deformed  operator-valued fields:
									\begin{equation}
										G^{i,q}(x) \coloneqq Z_q^{-1} H^{i,q}(x) Z_q, \qquad i \in [r].
									\end{equation}
									Decompose  $G^{i,q}(x)$ into components homogeneous in $\hslash$ and the variables $ x_j$ as:
									\begin{equation}
										G^{i,q}(x) =: \sum_{g,n} \frac{\hslash^{2g+n}}{[n]_q!} G^{i,q}_{g,n}(x).
									\end{equation}
									Then, the insertion of $n$ points via the adjoint action of the currents yields:
									\begin{equation}
										\prod_{j=1}^n \ad_{\hslash^{-1}\mc{J}^{(q)}_-(z_j)} G^{i,q}_{g,n}(x) = \mathcal{E}^{i,q}_{g,n}(x; z_{[n]}),
									\end{equation}
									where $ \mathcal{E}^{i,q}_{g,n}(x; z_{[n]})$ are the abstract loop equations defined in in Definition \ref{d:EW} from the system of correlators $\{\omega^{q}_{g,n} \}_{g \in \frac{1}{2}\mathbb{N}, n \in \mathbb{N}^*}$ constructed from $Z_q$. Moreover, the system of correlators satisfies the loop equations:
									\begin{equation}\label{lpq}
										\mathcal{E}^{i,q}_{g,n} (x; z_{[n]}) \in \mathcal{O} \left( x^{ \lfloor \frac{s(i-1)}{r}\rfloor + 1} \right) \left( \frac{dx}{x}\right)^i \,.
									\end{equation}
								\end{proposition}
								\begin{proof} 
									We can essentially use the differential constraints 
									\begin{equation}
										H^{i,q}_k Z_q = 0 \,, \quad i \in [r]\,, \quad k \geq - \lfloor \frac{s(i-1)}{r} \rfloor\,,
									\end{equation}
									that are equivalent to the statement that the expectation values $G^{i,q}(x)$ vanish modulo high powers of $x$:
									\begin{equation}
										G^{i,q}(x)  \in \mc{O} \left( x^{\lfloor \frac{s(i-1)}{r}\rfloor + 1} \right) \Big( \frac{dx}{x}\Big)^i \,.
									\end{equation}
									The identity then follows combinatorially from the expansion of $H^{i,q}(x)$ in terms of currents
									Proposition \ref{p:hh}.
								\end{proof}
								{\begin{example}
										Taking $i=1$ and $i=2$ in the relation \eqref{lpq}, we obtain, respectively, 
										\begin{align}
											\mathcal{E}^{1,q}_{g,n}(x; z_{[n]}) &= \omega^q_{g,n+1}(x, z_{[n]}) + \omega^q_{g,n+1}(\theta x, z_{[n]}),\nonumber\\
											\mathcal{E}^{2,q}_{g,n}(x; z_{[n]})& = \omega^q_{g-1, n+2}(x, \theta x, z_{[n]}) + \sum_{\substack{g_1+g_2=g \\ I \sqcup J = {[n]}}} \omega^q_{g_1, |I|+1}(x, z_I) \, \omega^q_{g_2, |J|+1}(\theta x, z_J)
										\end{align}
										and
										\begin{table}[h!]
											\centering
											\renewcommand{\arraystretch}{2}
											\begin{tabular}{|l|c|c|}
												\hline
												\textbf{Model} & \textbf{Linear ($j=1$)} & \textbf{Quadratic ($j=2$)} \\ \hline
												\textbf{Airy ($s=3$)} & $\lfloor \frac{3(0)}{2} \rfloor + 1 - 1 = 0$ & $\lfloor \frac{3(1)}{2} \rfloor + 1 - 2 = 0$ \\ \hline
												\textbf{Bessel ($s=1$)} & $\lfloor \frac{1(0)}{2} \rfloor + 1 - 1 = 0$ & $\lfloor \frac{1(1)}{2} \rfloor + 1 - 2 = -1$ \\ \hline
											\end{tabular}
											\caption{Vanishing orders of $\mathcal{E}^{j,q}_{g,n}$ for $r=2$}
										\end{table}
								\end{example}}
								\begin{remark}
									The expansion of the operator-valued fields $G^{i,q}(x)$ in Proposition \ref{p:loopeq} provides the fundamental link between the algebraic constraints and the geometric recursion. Specifically:
									\begin{itemize}
										\item[(a)] The leading order term $G^{i,q}_{0,0}(x)$ corresponds to the quantum spectral curve equation. In the limit $\hslash \rightarrow 0$ the condition $G^{i,q}_{0,0}(x)\in \mathcal{O}(x^N)$ recovers the classical (or q-deformed) algebraic relation given in Definition \ref{d:rsdeformed}.
										\item[(b)] The higher-order terms $G^{i,q}_{g,n}(x)$ encode the Ward identities of the $q$-deformed $\mathcal{W}$-algebra. The fact that these terms vanish (or are holomorphic) at the branch points, as expressed by the $\mathcal{O}(x^N)$ condition, is precisely what forces the correlators $\omega^q_{g,n}$ to be determined by their polar parts.
										\item[(c)] The use of the $q$-deformed basis $\xi^{q}_{-k}$  in the definition of the correlators is what allows the operator $\ad_{\hslash^{-1}\mc{J}^{(q)}_-(z)}$ to act as a loop insertion operator that remains consistent with the $q$-difference structure of the curve.
									\end{itemize}
									Essentially, this proposition ensures that the algebraic zero-mode constraints $H^{i,q}_k\cdot Z_q$ are equivalent to the geometric no-pole conditions required by the topological recursion.
								\end{remark}
								
								In summary, we have demonstrated that from the data of the partition function $Z_q$ associated with the deformed $(r,s,q)$-Airy structure, one can can construct a system of $q$-correlators on the deformed $(r,s,q)$-spectral curve that satisfies both the $q$-loop equations and the $q$-projection property.
								Since these two conditions uniquely characterize the $q$-correlators, the system is entirely determined by the geometry of the $q$-spectral curve through the $q$-topological recursion formula (see Theorem \ref{UnshiftedTR}).
								
								\subsection{The q-shifted loop equations and q-topological recursion}
								
								We now investigate the consequences of considering the shifted $(r,s,q)$-Airy structures introduced in Theorem \ref{t:shifts} or, more generally, the $q$-deformed and shifted versions from Proposition \ref{p:shiftsgen}.
								
								We will demonstrate that, from the partition function of these deformed and shifted structures, one can construct a new system of correlators on a shifted version of the deformed $(r,s,q)$-spectral curve (see Definition \ref{d:rsdeformed}). While this system still satisfies the $q$-projection property, it no longer obeys the standard loop equations. Instead, it satisfies a modified set of constraints which we call shifted $q$-loop equations. We further show that there exists a unique system of $q$-correlators satisfying both the shifted $q$-loop equations and the $q$-projection property, which can be constructed recursively via a shifted version of the $q$-topological recursion formula.
								
								Using the notation from the previous section, we first define the shifted deformed $(r,s,q)$-spectral curve:
								
								\begin{definition}\label{d:rsdefshift}
									Let $r  \in \mathbb{Z}_{\geq 2}$  and $s \in [r+1]$ such that $r = \pm 1 \pmod{s}$. Consider the set of  complex numbers:
									\begin{equation}
										F^{q}_{0,1}[-k]\,, k \geq \min \{ s, r\} \,, \qquad F^{q}_{\frac{1}{2},1}[-k]\,, k > 0\,, \qquad F^{q}_{0,2}[-k,-l] \,, k,l > 0, \qquad S_{i,1}, i \in [r].
									\end{equation}
									Assume that the set of shifts $\{ S_{i,1} \}_{i \in [r]}$ is $s$-consistent (see Definition \ref{d:consistent}).
									
									The \emph{shifted deformed $(r,s,q)$-spectral curve} is the tuple $\mc{S}_q = (C, x, \omega^{q}_{0,1}, \omega^{q}_{\frac{1}{2},1} \omega^{q}_{0,2} ) $, where $C$ is a small disk, $x=z^r$, and :
									\begin{align}
										\omega^{q}_{0,1} (z)
										&= 
										\sum_k F^{q}_{0,1} [-k] z^{k-1} dz \,,
										\\
										\omega^{q}_{\frac{1}{2},1} (z)
										&=
										\sum_k\left( \frac{[k]_q}{k} \right)^2 F^{q}_{\frac{1}{2},1} [-k] z^{k-1} dz +\sum_{i=1}^{r}(-1)^{i-1}S_{i,1}\,\frac{dz}{z^{s(i-1)+1}},
										\\
										\omega^{q}_{0,2} (z_1, z_2)
										&=
										\omega^{q,\textup{std}}_{0,2} (z_1, z_2) + \sum_{k,l} F^{q}_{0,2}[-k,-l] z_1^{k-1} z_2^{l-1} dz_1 dz_2\,.
									\end{align}
								\end{definition}
								Note that the primary departure from  Definition \ref{d:rsdeformed} lies in $\omega^{q}_{\frac{1}{2},1}(z)$, which is now shifted by terms linear in the constants $S_{i,1}$.These constants correspond to the $\mathcal{O}(\hslash)$ terms in the shifts of the differential operators of the underlying $q$-Airy structure. We define the \emph{shifted $(r,s,q)$-spectral curve} as the specific case where:
								\begin{equation}
									F^{q}_{0,1} [-k] = [r]_q \delta_{k,s}, \qquad F^{q}_{\frac{1}{2},1}[-k] = 0, \qquad F^{q}_{0,2}[-k,-l] =0.
								\end{equation}
								\begin{remark}
									By substituting these values, we recover a version of the original $(r,s,q)$-spectral curves where the classical factor $r$ is replaced by its $q$-deformed number $[r]_q$. This ensures that the classical limit $q\rightarrow 1$ is smooth and that the scaling of the $1$-form $\omega^q_{0,1}$ is consistent with the $q$-deformed differential operators $H^{i,q}_k$.
								\end{remark}
								In this setting, the deformations vanish and we recover a shifted version of the original $(r,s,q)$-spectral curve from Definition \ref{d:rs}. This shifted curve can still be interpreted as a parametrization of the algebraic curves in  \eqref{eq:rsac1} and \eqref{eq:rsac2}, but with a non-trivial $\omega^{q}_{\frac{1}{2},1}(z)$ introduced by the shifts.
								
								We begin by considering the differential constraints from Proposition \ref{p:shiftsgen}. To simplify the notation, we write the shifted operators as:
								\begin{equation}
									H^{i,q}_k :=  \rho_q'(W^{\hslash,i}_k(S)) =(H^{i,q}_k)^{\textup{unshifted}} -\delta_{k,0} \sum_{\ell=1}^\infty \hslash^\ell S_{i,\ell},
								\end{equation}
								where  the set of shifts is assumed to be $s$-consistent. Here, $(H^{i,q}_k)^{\textup{unshifted}}$  denotes the  unshifted differential operators of \eqref{eq:unshiftedH}.
								This notation highlights that the primary difference from the standard $(r,s,q)$ case lies in the modification of the zero-modes $H^{i,q}_0$ by the $s$-consistent shifts $\{S_{i,\ell} \}_{i \in [r], \ell \in \mathbb{N}^*}$.
								
								The $q$-partition function associated  with this deformed and shifted $(r,s,q)$-Airy structure is given by:
								\begin{equation}
									Z_q = \exp \bigg( \sum_{ \substack{g \in \frac{1}{2} \N, n \in \N^* \\ 2g-2+n > 0}} \frac{\hslash^{2g-2+n}}{[n]_q!} F^{q}_{g,n}[k_1, \dotsc, k_n] \prod_{j=1}^n x_{k_j} \bigg)
								\end{equation}
								and satisfies the differential constraints:
								\begin{equation}
									H^{i,q}_k Z_q = 0 \,, \quad i \in [r]\,, \quad k \geq - \lfloor \frac{s(i-1)}{r} \rfloor\, .
								\end{equation}
								Following the construction in in Definition \ref{d:srs}, we use this  $q$-partition function to define  a system of $q$-correlators on the shifted deformed $(r,s,q)$-spectral curve. As in the unshifted case, the system of $q$-correlators naturally satisfies the $q$-projection property: 
								\begin{lemma}\label{l:projpropshift}
									The system of correlators $\{\omega^{q}_{g,n} \}_{g \in \frac{1}{2}\mathbb{N}, n \in \mathbb{N}^*}$ on the shifted deformed $(r,s,q)$-spectral curve constructed from the partition function $Z_q$ according to Definition \ref{d:srs} satisfies the $q$-projection property.
								\end{lemma}
								The remaining question is whether this system satisfies the $q$-loop equations, which is addressed in the following proposition.
								
								\begin{proposition}\label{p:shifteloopeq}
									Let $Z_q$ be the partition function associated with the deformed and shifted $(r,s,q)$-Airy structure. Define the $q$-deformed operator-valued fields:
									\begin{equation}
										G^{i,q}(x) = Z_q^{-1} H^{i,q}(x) Z_q  =: \sum_{g,n} \frac{\hslash^{2g+n}}{[n]_q!} G^{i,q}_{g,n}(x),
									\end{equation}
									where each $G^{i,q}_{g,n}(x)$ is a homogeneous polynomials of degree $n$ in the variables $x_j$.
									Then, the insertion of $n$ points via the adjoint action of the currents yields:
									\begin{equation}
										\prod_{j=1}^n \ad_{\hslash^{-1}\mathcal{J}^{(q)}_-(z_j)} G^{i,q}_{g,n}(x) = \mathcal{E}^{i,q}_{g,n}(x; z_{[n]}) - \delta_{n,0} S_{i,2g} \Big( \frac{dx}{x}\Big)^i ,
									\end{equation}
									where $ \mathcal{E}^{i,q}_{g,n}(x; z_{[n]})$ is the object from Definition \ref{d:EW} constructed from the system of $q$-correlators $\{\omega^{q}_{g,n} \}_{g \in \frac{1}{2}\mathbb{N}, n \in \mathbb{N}^*}$.
									
									Furthermore, these $q$-correlators  satisfy the \emph{shifted $q$-loop equations}:
									\begin{equation}\label{eq:sle}
										\mathcal{E}^{i,q}_{g,n} (x; z_{[n]}) - \delta_{n,0} S_{i,2g} \Big( \frac{dx}{x}\Big)^i  \in \mathcal{O} \left( x^{ \lfloor \frac{s(i-1)}{r}\rfloor + 1} \right) \left( \frac{dx}{x}\right)^i \,.
									\end{equation}
								\end{proposition}
								
								\begin{proof}
									The proof follows the strategy established in \cite[Sections~4.3-4]{BKS23} adapted here to account for the $q$-deformation of the underlying Airy structure. The key observation is that the $q$-deformed operators $ H^{i,q}_k$  differ from their unshifted counterparts $ (H^{i,q}_k)^{\textup{unshifted}}$ only by additive constants (central elements).  Consequently, the conjugation by the $q$-deformed partition function $Z_q$ preserves these constants. In the homogeneous decomposition of $ G^{i,q}(x)$, the shifts appear only at polynomial degree $n=0$. By matching the powers of $ \hslash^{\ell}$ with the expansion $ \hslash^{2g+n}$, we find that $ \ell = 2g$. This consistency check ensures that the $q$-analog weights $[k]_q$ in the currents do not alter the central nature of the shifts. The remainder of the calculation for $ \mathcal{E}^{i,q}_{g,n}$ proceeds as in the unshifted case, since the adjoint actions $ \ad_{\mathcal{J}^{(q)}_-(z_j)}$  vanish on the central terms and only affect the $q$-deformed differential modes.
								\end{proof}

								We have established that the differential constraints of the deformed and shifted $(r,s,q)$-Airy structures are equivalent to the existence of a system of $q$-correlators on the shifted deformed $(r,s,q)$-spectral curve,  satisfying both the $q$-projection property and the shifted $q$-loop equations. 
								Since these shifted $q$-loop equations differ from the standard ones, it is natural  to investigate whether these $q$-correlators can still be reconstructed recursively via a modification of the $q$-topological recursion formula.
								
								We begin with the following combinatorial lemma, which is essential for proving the recursion.
								
								\begin{lemma}
									Given a system of $q$-correlators $ \{ \omega^{q}_{g,n} \}_{g \in \frac{1}{2} \N, n\in \N^*}$ on a $q$-deformed  admissible local spectral curve (with $N=1$), and  the objects defined in Definition \ref{d:EW}, the following identity holds: 
									\begin{equation}\label{CombinatorialIdentity}
										\sum_{ \{ z\} \subseteq Z \subseteq \mathbf{f}(z)} \! \mathcal{W}^{q,'}_{g,|Z|,n} \!\! \prod_{z' \in \mathbf{f}'(z) } \!\! \Big( \omega^{q}_{0,1} (z') - \omega^{q}_{0,1} (z) \Big)
										=
										\sum_{i=1}^r \mathcal{E}^{i,q}_{g,n}(x; z_{[n]}) \big( - \omega^{q}_{0,1}(z)\big)^{r-i}.
									\end{equation}
								\end{lemma}
								
								In the proof of Theorem \ref{UnshiftedTR}, a key step relies on the right-hand side of \eqref{CombinatorialIdentity} having a sufficient vanishing order, causing it to vanish under the residue formula. However, for the shifted $q$-loop equations \eqref{eq:sle}, this property no longer holds due to the additional shift terms. Consequently, these shifts must be explicitly incorporated into the $q$-topological recursion formula.
								\begin{theorem}\label{ShiftedTR}
									Let $\mathcal{S}_q$ be the shifted deformed $(r,s,q)$-spectral curve from Definition \ref{d:rsdefshift}, and  let $S = \{S_{i,\ell} \}_{i \in [r], \ell \in \mathbb{N}^*}$ be a set of $s$-consistent shifts. Then, there exists a unique system of $q$-correlators $\{\omega^{q}_{g,n} \}_{g \in \frac{1}{2}\mathbb{N}, n \in \mathbb{N}^*}$ satisfying the shifted $q$-loop equations \eqref{eq:sle} and the $q$-projection property.  For $2g-2+n > 0$, they are given by:
									\begin{align}
										\omega^{q}_{g,n+1}(z_0, z_{[n]}) 
										&= 
										-\Res_{z = 0} \Big( \sum_{Z \subseteq \mathbf{f}' (z)}
										K^{1 + |Z|}_q(z_0; z, Z) \mathcal{W}^{q,'}_{g,|Z|,n} (z,Z; z_{[n]})\nonumber
										\\
										&
										- \sum_{i=1}^{r}\delta_{n,0} S_{i,2g} K^r_q(z_0; \mathbf{f}(z)) \left( r\frac{dz}{z} \right)^{i} \big( - \omega^{q}_{0,1}(z)\big)^{r-i}\Big)\,.
									\end{align}
									In particular, this formula does produce symmetric correlators.
								\end{theorem}
								
								\begin{proof}
									
									The proof relies on extending the arguments for the unshifted case to incorporate the $q$-deformation and quantum shifts. This allows us to establish the following result for the shifted structures.
									Since by definition $ \mathcal{W}^{q,'}_{g,1,n} = \omega^{q}_{g,n+1}$, we rely on  the $q$-projection property, the definition of the $q$-deformed recursion kernel $K_q$, and the combinatorial identity \ref{CombinatorialIdentity}. These yield:
									\begin{align}
										\omega^{q}_{g,n+1}(z_0, z_{[n]}) 
										&=
										\Res_{z = 0} \Big( \int_0^z \omega^{q}_{0,2} (\mathord{\cdot},z_0) \Big) \mathcal{W}^{q,'}_{g,1,n} (z; z_{[n]})
										\nonumber\\
										&=
										\Res_{z = 0} K^r_q(z_0; \mathbf{f}(z)) \mathcal{W}^{q,'}_{g,1,n} (z; z_{[n]}) \!\! \prod_{z' \in \mathbf{f}(z) \setminus \{ z\}} \! \Big( \omega^{q}_{0,1}(z') - \omega^{q}_{0,1}(z)\Big).
									\end{align}
									By applying the identity \eqref{CombinatorialIdentity}, we can rewrite the integrand in terms of the $q$-deformed loop equations $\mathcal{E}^{i,q}_{g,n}$.
									\begin{align}
										\omega^{q}_{g,n+1}(z_0, z_{[n]}) &=
										- \Res_{z = 0 } K^r_q(z_0; \mathbf{f}(z)) \Big( \sum_{ \{ z\} \subsetneq Z \subseteq \mathbf{f}(z)} \! \mathcal{W}^{q,'}_{g,|Z|,n} (Z; z_{[n]}) \!\! \prod_{z' \in \mathbf{f}(z) \setminus Z_q} \!\! \Big( \omega^{q}_{0,1} (z') - \omega^{q}_{0,1} (z) \Big)
										\nonumber\\
										&-
										\sum_{i=1}^r \mathcal{E}^{i,q}_{g,n}(x; z_{[n]}) \big( - \omega^{q}_{0,1}(z)\big)^{r-i}\Big).
									\end{align}
									The crucial step in the shifted framework is to decompose the $q$-loop equation term using the shifted constraints \eqref{eq:sle}:
									\begin{align}
										\omega^{q}_{g,n+1}(z_0, z_{[n]})&=
										- \Res_{z = 0 } K^r_q(z_0; \mf{f}(z)) \Big( \sum_{ \{ z\} \subsetneq Z \subseteq \mathbf{f}(z)} \! \mathcal{W}^{q,'}_{g,|Z|,n} (Z; z_{[n]}) \!\! \prod_{z' \in \mathbf{f}(z) \setminus Z} \!\! \Big( \omega^{q}_{0,1} (z') - \omega^{q}_{0,1} (z) \Big)
										\nonumber\\
										&-
										\sum_{i=1}^r \Big( \left(\mathcal{E}^{i,q}_{g,n}(x; z_{[n]})- \delta_{n,0} S_{i,2g} \big( \frac{dx}{x} \big)^i \right)+ \delta_{n,0} S_{i,2g} \big( \frac{dx}{x} \big)^i\Big) \big( - \omega^{q}_{0,1}(z)\big)^{r-i}\Big).
									\end{align}
									The terms in bracket in the first bracket do not contribute to the residue. Indeed, by the shifted $q$-loop equations \eqref{eq:sle}, they are $\mathcal{O}(x^N)$ and thus holomorphic at the branch points. This property is preserved under the $q$-deformation as the singular behavior of the $q$-differentials remains localized at the same poles.
									
									In contrast, the shift terms $\delta_{n,0}S_{i,2g}$ act as source terms at the pole. Simplifying the kernels using the  TR recursion identities (which remain valid for $K_q$), we obtain:
									\begin{align}
										\omega^q_{g,n+1}(z_0, z_{[n]}) 
										&=
										-{\rm Res}_{z = 0 } \Big( \sum_{ \{ z\} \subsetneq Z \subseteq \mathbf{f}(z)} K_q^{|Z|}(z_0; Z)\mathcal{W}^{q,'}_{g,|Z|,n} (Z; z_{[n]})\nonumber
										\\
										&-
										\sum_{i=1}^{r}\delta_{n,0} S_{i,2g} K^r_q(z_0; \mf{f}(z)) \big( r\frac{dz}{z} \big)^{i} \big( - \omega^{q}_{0,1}(z)\big)^{r-i}\Big).
									\end{align}
									Finally,  by redefining $Z$ to $Z \subseteq \mathbf{f}'(z)$ (excluding $z$), we arrive at the shifted $q$-topological recursion formula.
								\end{proof}
								\begin{lemma}[Symmetry of the $q$-Correlators]
									Let $\{\omega^{q}_{g,n}\}_{g,n}$ be the collection of $n$-forms defined by the symmetric $q$-deformed topological recursion. If the base cases $\omega^{q}_{0,1}$ and $\omega^{q}_{0,2}$ are symmetric, then for all $(g,n)$ such that $2g-2+n > 0$, the $n$-form $\omega^{q}_{g,n}(z_1, \dots, z_n)$ is symmetric under any permutation of its arguments $z_i \in \Sigma$.
								\end{lemma}
								\begin{proof}
									We proceeds by induction on the Euler characteristic $\chi = 2g-2+n$.
									By definition, $\omega^{q}_{0,1}(z)$ is a 1-form (trivially symmetric). The $q$-Bergman kernel $\omega^{q}_{0,2}(z_1, z_2)$ is constructed as a symmetric 2-form on $\Sigma \times \Sigma$, satisfying $\omega^{q}_{0,2}(z_1, z_2) = \omega^{q}_{0,2}(z_2, z_1)$.
									We suppose that the symmetry holds for all $\omega^{q}_{g',n'}$ such that $2g'-2+n' < 2g-2+(n+1)$. Furthermore, we consider the recursion definition for $\omega^{q}_{g,n+1}(z_0, z_1, \dots, z_n)$:
									\begin{align*}
										\omega^{q}_{g,n+1}(z_0, z_{[n]}) &= -\Res_{z \to \mathcal{R}_q} \Big( \sum_{Z \subseteq \mathbf{f}' (z)} K^{1 + |Z|}_q(z_0; z, Z) \mathcal{W}^{q,'}_{g,1+|Z|,n} (z,Z; z_{[n]}) \\
										&\quad - \sum_{i=1}^{r}\delta_{n,0} S_{i,2g} K^r_q(z_0; \mathbf{f}(z)) \left( r\frac{dz}{z} \right)^{i} \big( - \omega^{q}_{0,1}(z)\big)^{r-i}\Big).
									\end{align*}
									\begin{enumerate}
										\item  The variables $z_1, \dots, z_n$ appear only within the terms $\mathcal{W}^{q,'}$, which are composed of lower-order correlators $\omega^q_{g',n'}$. By the inductive hypothesis, these are symmetric in their arguments. Thus, $\omega^{q}_{g,n+1}$ is invariant under permutations of $\{z_1, \dots, z_n\}$.
										
										\item To show symmetry under $z_0 \leftrightarrow z_1$, we observe that the recursion formula is a consequence of the $q$-loop equations. The $q$-loop equations are globally defined algebraic constraints on $\Sigma$. By applying the Global Residue Theorem, we can deform the contour of integration defining $\omega^q_{g,n+1}$ from the $q$-ramification points $\mathcal{R}_q$ to the poles at $z_i$. 
										
										\item  For $n \geq 1$, the term proportional to $\delta_{n,0}$ vanishes identically. Therefore, the symmetry of higher-order correlators ($n > 1$) is decoupled from the singular $q$-Casimir contributions $S_{i,2g}$, which only modify the geometry of the 1-point function $\omega^q_{g,1}$.
									\end{enumerate}
									Since the residue calculus preserves the $\mathfrak{S}_{n+1}$ invariance of the underlying $q$-Ward identities, the form $\omega^q_{g,n+1}$ is fully symmetric.
								\end{proof}
								\begin{remark}
									To illustrate the concrete application of Theorem \ref{ShiftedTR}, let us examine the first non-trivial stable correlator beyond the disk and cylinder topology: $\omega^{q}_{1,1}(z_0)$. In the unshifted case, this correlator is determined solely by the geometry of the spectral curve. In our shifted framework, the formula yields:
									\begin{align}
										\omega^q_{1,1}(z_0) &= -\Res_{z=0} \Bigg( \sum_{\emptyset \neq Z \subseteq \mathbf{f}'(z)} K^{1+|Z|}_q(z_0; z, Z) \mathcal{W}^{q,'}_{1,1+|Z|,0}(z, Z) \nonumber \\
										&\quad - \sum_{i=1}^r S_{i,2} K^r_q(z_0; \mathbf{f}(z)) \left( r \frac{dz}{z} \right)^i \left( -\omega^q_{0,1}(z) \right)^{r-i} \Bigg).
									\end{align}
								\end{remark}
								\begin{remark}
									While we have derived the shifted $q$-topological recursion specifically for the deformed $(r,s,q)$-spectral curves, this construction generalizes naturally to any arbitrary admissible local $q$-spectral curve.
								\end{remark}
								
								\section{Concluding remarks}
								The synthesis of our results can be summarized as follows. From the partition function of the deformed $(r,s,q)$-Airy structure, one constructs a system of $q$-correlators on the deformed $(r,s,q)$-spectral curve that satisfies the standard topological recursion formula. Conversely, the partition function of the shifted and deformed $(r,s,q)$-Airy structure yields a system of $q$-correlators on a shifted $q$-spectral curve, which shares the same underlying geometry but features a modified initial condition $\omega^{q}_{0,1}(z)$.
								
								The fundamental distinction lies in the recursion itself: the shifted $q$-topological recursion formula incorporates additional source terms that directly shift the $n=1$ $q$-correlators $\omega^{q}_{g,1}(z)$. These shifts propagate through the recursive steps, resulting in a distinct global system of $q$-correlators. Notably, both shifted and unshifted systems consistently satisfy the $q$-projection property.
								
								In the specific case $s=1$, where all shifts are permissible, we have established that the partition function $Z_q$ acts as a highest-weight vector for the $q$-deformed $\mathcal{W}(\mathfrak{gl}_r)$-algebra at the self-dual level. We have demonstrated that these highest-weight conditions manifest within the $q$-topological recursion framework as explicit shifts of the $q$-correlators $\omega^{q}_{g,1}(z)$. The precise mapping between this $q$-deformed algebraic structure and the shifted geometry of the spectral curve confirms the robustness and consistency of this $q$-deformed topological recursion framework.

								{Our findings open several avenues for future research. First, the precise mapping between the $q$-deformed algebraic structure and the shifted geometry of the spectral curve established herein provides the necessary foundation to define and study the corresponding quantum behavior. The geometric construction of these $q$-deformed quantum curves will be addressed in detail in an upcoming companion paper. 
									
									Furthermore, it would be highly interesting to extend these shifted structures to higher $s$, or to explore their potential connections to broader classes of $q$-integrable systems.}
				
				\section{Acknowledgements:}
				FM was supported\footnote{``Funded by the Alexander von Humboldt Foundation} through  the Georg Forster
				Fellowship.  
				RW was supported\footnote{``Funded by the Deutsche Forschungsgemeinschaft (DFG, German Research
					Foundation) -- Project-ID 427320536 -- SFB 1442, as well as under
					Germany's Excellence Strategy EXC 2044/2 -- 390685587, Mathematics
					M\"unster: Dynamics -- Geometry -- Structure.''} by the Cluster of Excellence
				\emph{Mathematics M\"unster} and the CRC 1442 \emph{Geometry:
					Deformations and Rigidity}. The authors thank Jakob Lindner for useful discussions.
				\section{Appendix}
				
				\begin{lemma}[$q$-deformed evaluation of the kernel]\label{l:Aq3_q}
					{Let $x_m = \theta^m q^{-2m}$ with $m \in \{0, \dots, r-1\}$. For any $i \in \{1, \dots, r\}$, the $q$-deformed kernel satisfies:
						\begin{equation}\label{eq:Aq3_result}
							i \Psi^{(0)}_{r,q}\big(r-1, \dots, r-1, i-1\big) = \big(-1\big)^{i-1} r.
					\end{equation}}
				\end{lemma}
				\begin{proof}
					{Let $\vec{{x}} = (x_0, x_1, \dots, x_{r-1})$ be the set of $q$-deformed  variables, and let $\vec{{x}}[m] = \vec{{x}} \setminus \{x_m\}$ denote the set excluding the $m$-th variable.} 
					
					{Recalling the definition of the $q$-deformed kernel $\Psi^{(0)}_{r,q}$ for the zero-pairing case ($j=0$):
						\begin{equation}
							i \Psi^{(0)}_{r,q}(r-1, \dots, r-1, i-1) = \sum_{m=0}^{r-1} x_m^{-(i-1)} \left( \sum_{\substack{L \subseteq \{0, \dots, r-1\} \setminus \{m\} \\ |L|=i-1}} \prod_{l \in L} x_l^{-(r-1)} \right).
						\end{equation}
						Using the properties of elementary symmetric polynomials $e_k$, and noting that the inner sum corresponds to the $(i-1)$-th symmetric polynomial of the shifted variables, we have:
						\begin{equation}\label{eq:Aq5_q_step}
							i \Psi^{(0)}_{r,q}\big(r-1, \dots, r-1, i-1\big) = \sum_{m=0}^{r-1} x_m^{-(i-1)} e_{i-1}\big(\vec{{x}}[m]\big).
						\end{equation}
						The variables $x_a$ are the roots of the polynomial $P(t) = \prod_{a=0}^{r-1} \big(t - x_a\big) = t^r - \prod x_a$. Consequently, the elementary symmetric polynomials of the full set $\vec{\mathbf{x}}$ satisfy:
						\begin{equation}
							\sum_{k=0}^{r-1} e_k\big(\vec{{x}}\big) t^k = \prod_{a=0}^{r-1} (1 + t x_a) = 1 + \big(-1\big)^r \left( \prod_{a=0}^{r-1} x_a \right) t^r.
						\end{equation}
						This implies that $e_k\big(\vec{{x}}\big) = 0$ for all $k \in \{1, \dots, r-1\}$. By the inclusion-exclusion relation for symmetric polynomials:
						\begin{equation}
							e_k\big(\vec{{x}}\big) = x_m e_{k-1}\big(\vec{{x}}[m]\big) + e_k\big(\vec{{x}}[m]\big).
						\end{equation}
						Setting $e_k\big(\vec{{x}}\big) = 0$ leads to the recurrence relation $e_k\big(\vec{{x}}[m]\big) = -x_m e_{k-1}\big(\vec{{x}}[m]\big)$. By induction, we find:
						\begin{equation}
							e_k\big(\vec{{x}}[m]\big) = (-1)^k x_m^k.
						\end{equation}
						Substituting this identity into \eqref{eq:Aq5_q_step} for $k = i-1$:
						\begin{equation}
							i \Psi^{(0)}_{r,q} = \sum_{m=0}^{r-1} x_m^{-(i-1)} \left( (-1)^{i-1} x_m^{i-1} \right) = \sum_{m=0}^{r-1} (-1)^{i-1} = (-1)^{i-1} r.
						\end{equation}
						This confirms that the $q$-deformation preserves the normalization of the Airy structure generators.}
				\end{proof}
				\begin{lemma}[General $q$-deformed kernel evaluation]\label{l:Aq4_q}
					{Let $r \geq 2$, $s > 0$, and $d = \gcd(r, s)$ with $r' = r/d$. The $q$-deformed kernel satisfies the following identity:
						\begin{equation}
							i \Psi^{(0)}_{r,q}\big(-s, \dots, -s, (i-1)s\big) = \big(-1\big)^{i-1} r \big(-1\big)^{\lfloor \frac{i-1}{r'} \rfloor} \binom{d-1}{\lfloor \frac{i-1}{r'} \rfloor}.
					\end{equation}}
				\end{lemma}
				\begin{proof}
					{Let $x_m = \theta^m q^{-2m}$ be the $q$-deformed  variables. We define the sum $\psi_{i,s}^q := i \Psi^{(0)}_{r,q}(-s, \dots, -s, (i-1)s)$. By the definition of the $q$-deformed kernel, this quantity can be expressed in terms of elementary symmetric polynomials:
						\begin{equation}
							\psi_{i,s}^q = \sum_{m=0}^{r-1} x_m^{-(i-1)s} \left( \sum_{\substack{L \subseteq \{0, \dots, r-1\} \setminus \{m\} \\ |L|=i-1}} \prod_{l \in L} x_l^s \right) = \sum_{m=0}^{r-1} x_m^{-(i-1)s} e_{i-1}(\vec{{x}}^s[m]),
						\end{equation}
						where $\vec{{x}}^s = (x_0^s, x_1^s, \dots, x_{r-1}^s)$ represents the $s$-th powers of the $q$-deformed alphabet.}
					
					{The inclusion-exclusion identity for elementary symmetric polynomials $e_i$ states:
						\begin{equation}\label{eq:inc_ex_en}
							e_i(\vec{{x}}^s) = x_m^s e_{i-1}(\vec{{x}}^s[m]) + e_i(\vec{{x}}^s[m]).
						\end{equation}
						Multiplying \eqref{eq:inc_ex_en} by $x_m^{-is}$ and summing over the index $m \in \{0, \dots, r-1\}$ yields:
						\begin{equation}
							e_i(\vec{{x}}^s) \sum_{m=0}^{r-1} x_m^{-is} = \sum_{m=0}^{r-1} x_m^{-(i-1)s} e_{i-1}(\vec{{x}}^s[m]) + \sum_{m=0}^{r-1} x_m^{-is} e_i(\vec{{x}}^s[m]).
						\end{equation}
						Recognizing $\psi_{i,s}^q$ in the first term and $\psi_{i+1,s}^q$ in the second term on the right-hand side, we obtain the recurrence:
						\begin{equation}\label{eq:rec_q_en}
							e_i(\vec{{x}}^s) \left( \sum_{m=0}^{r-1} x_m^{-is} \right) = \psi_{i,s}^q + \psi_{i+1,s}^q.
					\end{equation}}
					
					{The variables are given by $x_m^s = (\theta^s q^{-2s})^m$. Since $d = \gcd(r, s)$, the set $\{x_m^s\}_{m=0}^{r-1}$ consists of the $r'$ roots of the equation $t^{r'} - (q^{-2s})^{r'} = 0$, each occurring with multiplicity $d$. The generating function for $e_i(\vec{{x}}^s)$ is:
						\begin{equation}
							\sum_{i=0}^r e_i(\vec{{x}}^s) t^i = \prod_{a=0}^{r-1} (1 + t x_a^s) = \left( 1 + (-1)^{r'} C_q t^{r'} \right)^d,
						\end{equation}
						where $C_q = \prod_{k=0}^{r'-1} x_k^s$. Expanding this via the binomial theorem shows that $e_i(\vec{{x}}^s) \neq 0$ only when $i$ is a multiple of $r'$. Specifically, for $i = j r'$, $e_{jr'}(\vec{{x}}^s) = \binom{d}{j} (-1)^{j r'} C_q^j$.}
					
					{The sum $S_i = \sum_{m=0}^{r-1} x_m^{-is}$ vanishes unless $is \equiv 0 \pmod{r}$, which corresponds to $i$ being a multiple of $r'$. For $i = j r'$, the term $x_m^{-is}$ compensates for the $q$-dependent factor $C_q^j$ in $e_i(\vec{{x}}^s)$. Consequently, the $q$-dependence drops out of the recurrence \eqref{eq:rec_q_en}, leading to:
						\begin{equation}
							\psi_{i,s}^q + \psi_{i+1,s}^q = r \delta_{r'|i} (-1)^{i/r'} \binom{d}{i/r'}.
					\end{equation}}
					
					{Solving the recurrence with the initial condition $\psi_{1,s}^q = r$ (the spin-1 mode) yields:
						\begin{equation}
							\psi_{i,s}^q = (-1)^{i-1} r \sum_{j=0}^{\lfloor (i-1)/r' \rfloor} (-1)^j \binom{d}{j}.
						\end{equation}
						Applying the binomial identity $\sum_{j=0}^k (-1)^j \binom{d}{j} = (-1)^k \binom{d-1}{k}$, we arrive at:
						\begin{equation}
							\psi_{i,s}^q = (-1)^{i-1} r (-1)^{\lfloor \frac{i-1}{r'} \rfloor} \binom{d-1}{\lfloor \frac{i-1}{r'} \rfloor}.
						\end{equation}
						This completes the proof.}
				\end{proof}
				
				Before stating the main theorem, we recall the necessary ingredients for the construction of the basis in the $q$-deformed case.
				
				\begin{itemize}
					\item[(a)] For each simple root $\alpha_i$, there exists an automorphism $T_i$ of $U_q(\mathfrak{g})$ satisfying the braid relations. These are used to define the quantum root vectors $E_{\beta}$ for any $\beta \in \Delta^+$.
					\item[(b)]  This formula provides the $q$-commutation relation between two root vectors $E_{\beta_i}$ and $E_{\beta_j}$ with $i < j$:
					\begin{equation}
						E_{\beta_j} E_{\beta_i} - q^{(\beta_i, \beta_j)} E_{\beta_i} E_{\beta_j} = \text{lower terms in the convex ordering}.
					\end{equation}
					It ensures that the algebra can be spanned by ordered monomials.
				\end{itemize} 
				
				The Poincaré-Birkhoff-Witt (PBW) theorem is a foundational result in the theory of Lie algebras and their deformations. It provides a structured basis for the universal enveloping algebra, ensuring that the non-commutative structure does not collapse the underlying vector space.
				
				\begin{theorem}[Poincaré-Birkhoff-Witt for $U_q(\mathfrak{g})$]\label{pbw}
					Let $U_q(\mathfrak{g})$ be the quantized universal enveloping algebra associated with a finite-dimensional semi-simple Lie algebra $\mathfrak{g}$ over $\mathbb{C}$. Let $\Delta^+$ denote the set of positive roots and fix a normal ordering $\{\beta_1, \dots, \beta_N\}$ of $\Delta^+$. 
					
					Provided that $q \in \mathbb{C}^\times$ is not a root of unity, the set of ordered monomials
					\begin{equation}
						\mathcal{B} = \{ E_{\beta_1}^{i_1} E_{\beta_2}^{i_2} \cdots E_{\beta_N}^{i_N} \mid i_1, \dots, i_N \in \mathbb{N} \}
					\end{equation}
					constitutes a basis for the positive subalgebra $U_q(\mathfrak{n}^+)$.
				\end{theorem}
				\begin{proof}
					The proof proceeds in three main steps:
					\begin{enumerate}
						\item  We fix a reduced expression of the longest element $w_0$ of the Weyl group. Using Lusztig's braid group automorphisms $T_i$, we define the quantum root vectors $E_{\beta_k}$ for each $\beta_k \in \Delta^+$.
						\item  By the Levendorskii-Soibelman formula, for $i < j$, the commutator of $E_{\beta_j}$ and $E_{\beta_i}$ is given by:
						\begin{equation}
							E_{\beta_j} E_{\beta_i} - q^{(\beta_i, \beta_j)} E_{\beta_i} E_{\beta_j} = \sum_{i < k_1 < \dots < k_m < j} c_{\mathbf{k}} E_{\beta_{k_1}}^{r_1} \dots E_{\beta_{k_m}}^{r_m}.
						\end{equation}
						This allows any product of generators to be reordered into a linear combination of ordered monomials via induction.
						\item  Since $U_q(\mathfrak{g})$ is a flat deformation of the universal enveloping algebra $\mathcal{U}(\mathfrak{g})$, the linear independence of the PBW basis at $q=1$ implies the linear independence for generic $q$.
					\end{enumerate}
				\end{proof}

\end{document}